\documentclass{article}

\usepackage{iclr2027_conference,times}

\usepackage[utf8]{inputenc}
\usepackage[T1]{fontenc}
\usepackage{graphicx}
\usepackage{hyperref}
\usepackage{url}
\usepackage{booktabs}
\usepackage{amsfonts}
\usepackage{nicefrac}
\usepackage{microtype}

\IfFileExists{headers/config/showoverfull.config}{
	\overfullrule=1cm
}{
}

\usepackage{marginnote}

\usepackage[backgroundcolor=none,linecolor=red,textsize=footnotesize]{todonotes}

\usepackage{etoolbox}

\newbool{includeappendix}
\setbool{includeappendix}{true} %
\IfFileExists{headers/config/noappendix.config}{
	\setbool{includeappendix}{false}
}{}

\newif\ifincludeappendixx
\ifbool{includeappendix}{
	\includeappendixxtrue
}{
	\includeappendixxfalse
}

\usepackage{xr} %
\usepackage{filecontents}

\ifbool{includeappendix}{}{

	\externaldocument{appendix-labels}
}

\usepackage{xspace}

\newcommand{\eg}{e.g., }
\newcommand{\ie}{i.e., }

\newcommand{\nth}[1]{\ensuremath{#1^\text{th}}}

\usepackage{acro} %

\usepackage{listings}

\usepackage{textcomp}

\usepackage{xcolor}

\usepackage[scaled=0.8]{beramono}

\definecolor{ckeyword}{HTML}{7F0055}
\definecolor{ccomment}{HTML}{3F7F5F}
\definecolor{cstring}{HTML}{2A0099}

\lstdefinestyle{numbers}{
	numbers=left,
	framexleftmargin=20pt,
	numberstyle=\tiny,
	firstnumber=auto,
	numbersep=1em,
	xleftmargin=2em
}

\lstdefinestyle{layout}{
	frame=none,
	captionpos=b,
}

\lstdefinestyle{comment-style}{
	morecomment=[l]//,
	morecomment=[s]{/*}{*/},
	commentstyle={\color{ccomment}\itshape},
}

\lstdefinestyle{string-style}{
	morestring=[b]",%
	morestring=[b]',%
	stringstyle={\color{cstring}},
	showstringspaces=false,%
}

\lstdefinestyle{keyword-style}{
	keywordstyle={\ttfamily\bfseries},
	morekeywords={
		function,
		constructor,
		int,
		bool,
		return,
		returns,
		uint
	},
	morekeywords = [2]{},
	keywordstyle = [2]{\text},
	sensitive=true,
}

\lstdefinestyle{input-encoding}{
	inputencoding=utf8,
	extendedchars=true,
	literate=
	{ℝ}{$\reals$}1%
	{→}{$\rightarrow$}1%
	{α}{$\alpha$}1%
	{β}{$\beta$}1%
	{λ}{$\lambda$}1%
	{θ}{$\theta$}1%
	{ϕ}{$\phi$}1%
}

\lstdefinestyle{escaping}{
	moredelim={**[is][\color{blue}]{\%}{\%}},
	escapechar=|,
	mathescape=true
}

\lstdefinestyle{default-style}{
	basicstyle=\fontencoding{T1}\ttfamily\footnotesize,
	style=numbers,
	style=layout,
	style=comment-style,
	style=string-style,
	style=keyword-style,
	style=input-encoding,
	style=escaping,
	tabsize=2,
	upquote=true
}

\lstdefinelanguage{BASIC}{
	language=C++,
	style=default-style
}[keywords,comments,strings]%

\usepackage{algorithm}
\usepackage{algpseudocode}
\usepackage{longtable}
\usepackage{colortbl}
\usepackage{wrapfig}
\usepackage{needspace}
\makeatletter
\newcommand\wrapfill{\par
  \ifx\parshape\WF@fudgeparshape
    \nobreak\vskip-\baselineskip\vskip\c@WF@wrappedlines\baselineskip\allowbreak\WFclear
  \fi}

\makeatother
\usepackage{capt-of}

\usepackage{tikz}

\usetikzlibrary{arrows}
\usetikzlibrary{automata}
\usetikzlibrary{calc}
\usetikzlibrary{backgrounds}
\usetikzlibrary{decorations.markings}
\usetikzlibrary{decorations.pathmorphing}
\usetikzlibrary{decorations.pathreplacing}
\usetikzlibrary{fit}
\usetikzlibrary{patterns}
\usetikzlibrary{positioning}
\usetikzlibrary{shadows}
\usetikzlibrary{shapes}
\usetikzlibrary{shapes.geometric}

\usetikzlibrary{arrows.meta,svg.path}
\usepackage{fontawesome5} %
\usepackage{dsfont} %
\usepackage{amsmath} %
\newcommand{\indicator}{\mathds{1}}
\newcommand{\astra}{\textsc{GPT-6 Astra}\xspace}
\newcommand{\opus}{\textsc{Opus 5}\xspace}
\newcommand{\gemini}{\textsc{Gemini-3.8 Flash}\xspace}
\usepackage{amsthm}
\newtheorem{proposition}{Proposition}
\newtheorem{lemma}{Lemma}

\usepackage[most]{tcolorbox}
\usepackage{enumitem}
\definecolor{PromptAccent}{HTML}{1F77B4}
\newtcolorbox{promptbox}[1][]{%
  enhanced,
  breakable,
  colback=white,
  colframe=black!16,
  boxrule=0.55pt,
  arc=1.6mm,
  outer arc=1.6mm,
  left=10pt,right=10pt,
  top=8pt,bottom=8pt,
  borderline west={2.2pt}{0pt}{PromptAccent},
  drop shadow={black!12},
  fonttitle=\sffamily\bfseries\footnotesize,
  coltitle=PromptAccent!90!black,
  colbacktitle=PromptAccent!10,
  title={#1},
  attach boxed title to top left={xshift=10pt,yshift*=-2mm},
  boxed title style={
    enhanced,
    arc=1.2mm,
    boxrule=0.35pt,
    colframe=PromptAccent!35,
    colback=PromptAccent!10,
    left=6pt,right=6pt,top=2pt,bottom=2pt
  },
  before skip=10pt,
  after skip=10pt
}

\renewcommand{\S}{Sec.~}

\usepackage[capitalize,noabbrev]{cleveref}

\crefformat{section}{\S#2#1#3}

\crefrangeformat{section}{\S#3#1#4\crefrangeconjunction\S#5#2#6}

\crefmultiformat{section}{\S#2#1#3}{\crefpairconjunction\S#2#1#3}{\crefmiddleconjunction\S#2#1#3}{\creflastconjunction\S#2#1#3}

\newcommand{\crefrangeconjunction}{--}

\crefname{listing}{Lst.}{listings}
\crefname{line}{Lin.}{Lin.}
\crefname{appendix}{App.}{App.}

\newcommand{\appref}[1]{%
	\ifbool{includeappendix}{\cref{#1}}{the appendix}%
}
\newcommand{\Appref}[1]{%
	\ifbool{includeappendix}{\cref{#1}}{The appendix}%
}

\newcommand{\OurTitle}{AutoMark: Enabling Autoresearch to\\Discover Better LLM Watermarks}

\title{\OurTitle{}}

\author{Thibaud Gloaguen, Robin Staab, Martin Vechev \\
ETH Zurich \\
\texttt{thibaud.gloaguen@inf.ethz.ch}}

\iclrfinalcopy

\begin{document}

\maketitle

\begin{abstract}
  With LLM watermarking being deployed commercially and now required by regulations, improving its reliability and effectiveness has become crucial.
Yet, recent progress in the field of LLM watermarking has increasingly been driven by improving details of existing methods, an effort fundamentally limited by the pace of human researchers.
In this work, we enable for the first time the autonomous discovery of new distortion-free state-of-the-art watermarking schemes.
To enable this, we (i) establish strict criteria to ensure that watermarks are reliable (e.g., they do not have an unexpectedly high false positive rate), (ii) propose rigorous statistical tests to automatically evaluate whether a watermarking scheme satisfies our criteria, and (iii) design an evaluation suite to rank watermarks along three key dimensions: detectability, quality, and robustness.
By running our framework with 3 frontier models (GPT-6 Astra, Opus 5, Gemini-3.8 Flash), we discover over 50 different watermarking schemes, including several that outperform prior works along all key dimensions.
We complement this by a manual study of the discovered schemes, distilling the key ideas into smaller components, and individually studying the impact of each component across dimensions (detectability, quality, robustness) to better understand how the proposed schemes operate.
Importantly, we find that the agents, on top of improving existing ideas, also discover fundamentally new ideas (e.g., aligning watermark scores with random per-request direction).
Overall, our work establishes the first steps of fully autonomous watermarking research, enabling the discovery of more reliable and effective watermarks.
Our code is available \href{https://github.com/eth-sri/automark}{here}, and a blogpost to visualize our results \href{https://www.sri.inf.ethz.ch/blog/automark}{here}.

\end{abstract}

\suppressfloats[t]
\begin{figure}[t]
    \centering
    \vspace{-0.5em}
    \resizebox{\linewidth}{!}{\input{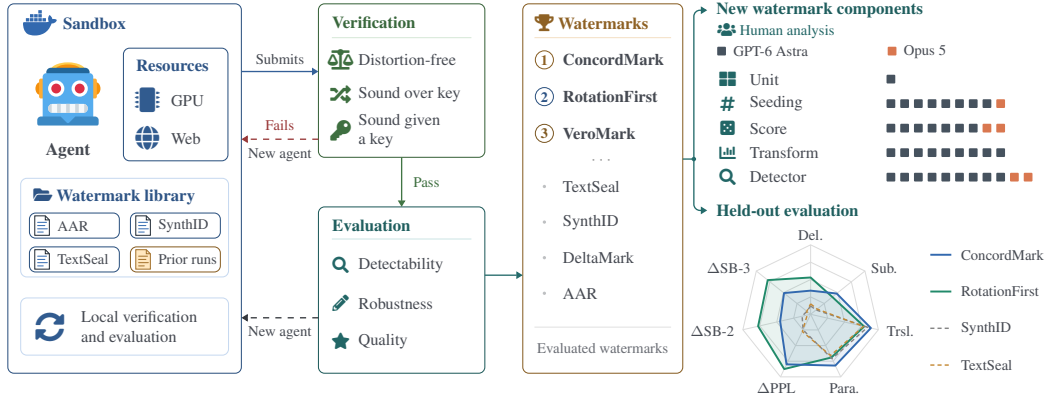}}
    \vspace{-1em}
    \caption{Overview of our work. In an isolated sandbox, an agent iterates on a watermarking scheme. Submissions are first verified and, upon passing, evaluated and added to a leaderboard. We then manually decompose the discovered schemes into 5 components and ablate them for insights. Lastly, we evaluate two selected schemes on a held-out test set and find that they outperform prior methods.}
    \label{fig:automark-overview}
    \vspace{-0.5em}
\end{figure}

\section{Introduction}
\label{sec:introduction}

Large Language Model (LLM) watermarking, which embeds a signal invisible to humans in model outputs, has emerged as the standard for tracing AI-generated text: regulators mandate it~\citep{euiact}, and major providers deploy it~\citep{synthid, anthropic}.
This makes research on new watermarking schemes crucial as they have to be reliable enough to prevent false accusations and robust enough to prevent evasion across real-world deployments.
Recent progress has been steady but limited with improvements mostly driven by refining existing designs, \eg with entropy weighting~\citep{sweet}, or key switching~\citep{textseal}, and the pace of iteration capped by that of human researchers.
Autonomous research agents are thus a natural and promising direction for further accelerating LLM watermarking research.

\paragraph{Autoresearch}
Autoresearch systems let agents propose, implement, and evaluate ideas in a loop~\citep{funsearch, alphaevolve}, and have already led to scientific breakthroughs~\citep{openai2026euler}.
Fields with readily verifiable outputs, such as mathematics with Lean-generated proofs, are especially well suited to this approach.
Watermarking can also benefit from autoresearch, but its evaluation requires more care: useful properties such as detectability, robustness, and quality trade off, so improvement on a single metric need not represent a better watermark.
Also, agents may exploit gaps in the evaluation harness (\ie reward hacking), \eg artificially inflating detectability/robustness by increasing the rate of falsely flagged human texts.
A separate concern is that agents may produce results faster than humans can understand them, especially when their explanations are poor, making the resulting solutions difficult to trust.
In this work, we show that autoresearch, when carefully executed with an adequate harness and methodology, discovers watermarks that are trustworthy, explainable, and Pareto-superior to existing watermarking schemes.

\paragraph{This work: autoresearch for LLM watermarks}
To enable autoresearch for watermarking, we mathematically formalize what makes a watermark valid~(\cref{sec:method:definition_validity}) using two criteria: soundness in expectation over the private key, and, given a fixed text distribution, the probability of sampling a key that leads to false flags should be small.
Given these criteria, we design our autoresearch harness illustrated in~\cref{fig:automark-overview}.
An agent is fully isolated in an environment and discovers new watermarking schemes~(left).
It can submit its scheme to a \emph{private} verification pipeline~(middle) that (i) tests whether the proposed scheme violates our criteria (using tailored tests we designed in \cref{sec:method:verification}), and, upon passing, (ii) evaluates the schemes along three key axes (detectability, robustness, and quality).

Running our harness on three state-of-the-art agents, \astra, \opus, and \gemini, we generated over 50 schemes, with some outperforming prior work on all our evaluated axes.
To understand the key insights of these schemes, we manually classify them all into five atomic components~(\cref{fig:automark-overview}, right) and ablate their individual effects~(\cref{sec:eval:directions,app:components}).
We also select two schemes, ConcordMark and RotationFirst, that we evaluate on a held-out evaluation set~(\cref{sec:eval:best}).
We find that ConcordMark achieve state-of-the-art robustness and is Pareto-superior to all evaluated prior work, whereas RotationFirst has performance similar to prior work at much higher text quality.

\paragraph{Main contributions}
Our main contributions are:
\begin{itemize}
    \item We formalize the criteria for reliable watermarking, and design statistical tests that automatically detect violations (\cref{sec:method:definition_validity,sec:method:verification}).
    \item We build an autoresearch harness for LLM watermarking around this verification (\cref{sec:method:framework}).
    \item We show that agents running our harness push the Pareto frontier of LLM watermarks, with two new schemes  that outperform prior work on a held-out evaluation (\cref{sec:eval:trajectories,sec:eval:best}).
    \item We distill the agents' schemes into 39 new components, ablate each of them, and identify fundamentally new watermarking ideas (\cref{sec:eval:directions}).
\end{itemize}

\section{Related Work}
\label{sec:related}

\paragraph{LLM watermarking}
LLM watermarks modify the sampling procedure of an LLM to embed a (key-dependent) signal, that a later statistical test can detect.
Red-Green watermarks~\citep{kgw} achieve this by distorting the next-token distribution through a pseudo-random green token bias.
In contrast, recent schemes, including all in this work, are distortion-free, \ie they preserve the next-token distribution in expectation over the key.
AAR~\citep{aar} uses the Gumbel-max trick, KTH~\citep{kth} inverse-transform sampling, DiPMark and MCMark~\citep{dipmark,mcmark} reweight the distribution based on score-rankings, and SynthID~\citep{synthid}, deployed in production, uses tournament sampling.
Most recently, TextSeal~\citep{textseal} extends Gumbel-max sampling with dual-key generation and entropy-weighted detection.

\paragraph{Watermark properties}
Beyond detectability, benchmarks and toolkits~\citep{markmywords,waterbench,markllm} evaluate watermarks along several, often competing, properties.
\emph{Robustness} measures whether the watermark remains detectable after the text is edited, \eg by word substitutions, translation, or paraphrasing~\citep{kgw2,kth}.
\emph{Quality} measures the impact of the watermark on the generated text, \eg on perplexity or downstream tasks~\citep{waterbench}.
Distortion-free schemes preserve the distribution of each generation, but, as a repeated context always yields the same scores, they reduce output diversity across generations~\citep{unified_wm_framework} which noticeably affects quality.
\emph{Security} measures whether an adversary can learn the watermark rules to spoof or scrub it~\citep{ws}, and is outside the scope of this work.
Lastly, \emph{reliability} requires the false positive rate of the detector to be controlled, \eg \citet{threebricks} show that the z-tests used in early works underestimate the false positive rate, and propose exact tests with de-duplication of repeated (context, token) pairs.
Yet, in prior works, these guarantees hold only in expectation over the private key: a specific deployed key can still exhibit a high false positive rate.

\paragraph{Autoresearch}
Program search systems~\citep{funsearch,alphaevolve,shinkaevolve} let LLMs iteratively propose programs scored by an automatic evaluator, leading to new results in verifiable domains such as algorithm design and mathematics~\citep{openai2026euler}.
Other works automate the entire research loop, from generating ideas to writing papers~\citep{ai_scientist}, discovering new training objectives~\citep{discopop}, or improving LLM training~\citep{autoresearch}, and benchmarks measure the research capabilities of agents~\citep{mlebench,rebench}.
Yet, all require a strict evaluator as otherwise agents may exploit loopholes, \eg by modifying tests instead of solving the task~\citep{impossiblebench}.
To our knowledge, we are the first to propose an automated pipeline to verify and evaluate watermarks, enabling autoresearch for LLM watermarking.

\section{Method}
\label{sec:method}

In this section, we introduce the necessary background on LLM watermarking~(\cref{sec:method:preliminaries}), explain our design for automatically verifying watermarking schemes~(\cref{sec:method:verification}), and finally detail our autoresearch framework~(\cref{sec:method:framework}).

\subsection{Preliminaries on LLM Watermarks}
\label{sec:method:preliminaries}

In this part, we describe the four key components of most watermarking algorithms.

\paragraph{Mathematical description}
Let $\Sigma$ be the finite vocabulary, and let $\omega \in \Sigma^*$ be a finite sequence of tokens.
At each step $t$ of autoregressive generation, the LLM returns $p_t \in \Delta(\Sigma)$, the conditional probability distribution of the next token given $\omega_{<t}$.
Without a watermark, the next token is sampled according to $p_t$.
With a watermark, we assume the existence of a sequence $(U_{a_t})$ of pseudo-random vectors in $[0,1]^{|\Sigma|}$ whose entries $U_{a_t,v}$ are i.i.d. uniform variables seeded by the context hash $a_t$, and a watermark function $f: [0,1]^{|\Sigma|}\times \Delta(\Sigma) \rightarrow \Delta(\Sigma)$ that transforms the next-token probability distribution $p_t$ into a watermarked probability distribution $f(\tilde{U}_{a_t},p_t)$.
We refer to $f$ as the \emph{logits transformation} and $(U_{a_t})$ as the \emph{watermark score}, where $(\tilde{U}_{a_t})$ are i.i.d. uniform variables derived from $(U_{a_t})$ and in most cases $\tilde{U}_{a_t}=U_{a_t}$.
We give concrete examples in \cref{app:setup:baselines}.
The next token $\omega_t$, sampled according to $f(\tilde{U}_{a_t},p_t)$, will be correlated with the sequence $(U_{a_t})$.

For detection, given any text $\omega \in \Sigma^*$, we recover the corresponding scores $U_{a_t,\omega_t}$ at retained positions $t \in D$, with $N \mathrel{:=} |D|$, after deduplicating repeated (seed, token) pairs.
By assumption, under the null hypothesis (\ie the text was generated without knowledge of the watermark), the retained scores $U_{a_t,\omega_t}$ are i.i.d. uniform random variables, whereas when the watermark is used, they are not.
We can therefore build a statistic and a corresponding statistical test to detect the watermark, \ie a \emph{detector}.

\paragraph{Additional requirements}
Furthermore, we require the watermark to be distortion-free,
\begin{equation}
	\label{eq:distortion_free}
	\forall p \in \Delta(\Sigma), \mathbb{E}_{\tilde{U}}[f(\tilde{U},p)] = p.
\end{equation}
This property ensures that a watermark has minimal impact on quality: with perfect randomness, the quality of the watermarked model is indistinguishable from its base model.
Yet, in practice, as we explain below we can still observe mild quality degradation within distortion-free watermarks.

\paragraph{Practical considerations}
In practice, the sequence $(U_{a_t})$ needs to be computed from the sequence of tokens $\omega$ so that it can be recovered at detection time.
A typical way to derive $(U_{a_t})$ is to use a hash of the preceding tokens and a private watermark key $\xi$ to obtain a seed $a_t$ for a pseudorandom number generator~(\cref{app:hash_function}).
The resulting scores $U_{a_t,v}$ remain approximately i.i.d. across distinct (seed, token) pairs and can be derived directly from the text given the private key.
We refer to this as \emph{context seeding}.
Seeding creates some practical problems: in case of a repeated (seed, token) pair, the watermark scores are no longer i.i.d., making detection unsound under the null.
For a single request, we use a cache mechanism: if a (seed, token) pair is repeated we no longer apply the watermark.
For repeated requests however, we cannot apply this cache as in the limit, the watermark is no longer applied.
This lack of i.i.d. has also direct implications for the distortion-free property~(\cref{eq:distortion_free}) which in practice is therefore also violated across a larger corpus with repeated tokens.
Hence, even for theoretically distortion-free schemes, we have to measure their impact on quality.

\subsection{What Defines a Valid Watermarking Scheme}
\label{sec:method:definition_validity}

To automatically verify watermarking scheme validity at scale, we first need to formally define what makes a watermarking scheme valid.
In particular, we identify two key criteria: \emph{soundness over the key} and \emph{soundness given a key}.
Though our definitions build upon prior work, we find in \cref{app:soundness_key} that several existing schemes fail our more stringent validity checks.
In \cref{app:invalid}, we further motivate our criteria by observing the failure modes of schemes designed by autoresearch without these safeguards.

\paragraph{Soundness over the key}
Soundness over the key is the property that, in expectation over the watermarking key $\xi$, the detector is a valid statistical test.
Formally we require,
\begin{equation}
	\label{eq:soundness_over_key}
	\forall \omega \in \Sigma^*, \forall \alpha \in [0,1], \quad \Pr_{\xi}[P_{\xi}(\omega) \le \alpha] \le \alpha,
\end{equation}
where $P_{\xi}(\omega)$ is the p-value returned by the detector.
Most prior watermarking schemes satisfy this property, which distinguishes statistical watermarks from zero-shot LLM-text detectors. %
However, this property alone does not guarantee that the watermarking scheme is reliable.

For instance, consider a deliberately naive watermarking scheme: using the private key $\xi$, we randomly partition the vocabulary into a hundred classes and declare a text watermarked if its last token lies in the first class.
Averaged over the key, this occurs with probability at most $1/100$, so the scheme satisfies \cref{eq:soundness_over_key}.
Yet its FPR can be high for a specific key: if, for some $\xi$, \textrm{"."} lies in the first class, most human texts will be classified as watermarked.
Thus, while \cref{eq:soundness_over_key} bounds the FPR in expectation over the private key, it does not ensure a well-behaved FPR for each fixed key.

\paragraph{Soundness given a key}
We therefore introduce a novel second requirement for watermarking schemes: soundness given a key.
Ideally, we want to require that for a human text distribution and a p-value threshold $\alpha$, the false positive rate is bounded by $\alpha$.
This is arguably impractical as it would require precisely specifying the human text distribution, which, e.g., noticeably differs between languages.
Instead, we move the randomness to the key: we require that, for \emph{every} distribution $\mathcal{H}_0$ over texts, the probability of sampling a key whose FPR on $\mathcal{H}_0$ exceeds $\alpha$ is at most $\beta$.
\begin{equation}
	\label{eq:soundness_given_key}
	\forall \mathcal{H}_0 \in \Delta(\Sigma^*),\ \forall \alpha \in [0,1], \quad
	\Pr_{\xi}\big[ \Pr_{\Omega \sim \mathcal{H}_0} [P_{\xi}(\Omega) \le \alpha] > \alpha \big] \le \beta.
\end{equation}
For $\beta = 5\%$, if a watermarking scheme satisfies \cref{eq:soundness_given_key}, then for any fixed human text distribution, the probability of sampling a key for which the detector is not well calibrated (on that distribution) is at most 5\%.
Interestingly, as we show in \cref{app:proofs}, a scheme satisfying \cref{eq:soundness_given_key} can be derived from a scheme satisfying \cref{eq:soundness_over_key} at the cost of its power:
The same scheme with its p-value multiplied by $1/\beta$ also satisfies \cref{eq:soundness_given_key} (while still satisfying \cref{eq:soundness_over_key}).

\subsection{Automatic Verification}
\label{sec:method:verification}

Now we describe how we automatically test whether a proposed watermarking scheme violates our criteria~(\cref{sec:method:definition_validity}).
To this end, we design three statistical tests to test whether a scheme is distortion-free~(\cref{eq:distortion_free}), sound over the key~(\cref{eq:soundness_over_key}), and sound given a key~(\cref{eq:soundness_given_key}).
\begin{enumerate}[leftmargin=*]
	\item \textbf{Distortion-freeness}~(\cref{alg:distortion_free_test}): for given unwatermarked next-token distributions, we average the watermarked distributions over independently sampled private keys and test whether it significantly differs from the original.
	      We check both individual probability shifts and smaller shifts spread across tokens, using independent samples to identify and test a potential bias direction.
	      If either test detects a significant deviation after multiple testing corrections, the scheme is rejected.
	\item \textbf{Soundness over the key}~(\cref{alg:soundness_over_key_test}): for a given set of non-watermarked texts, the test verifies for independently sampled private keys whether the empirical false positive rate is abnormally high.
	      If it is, it means that the scheme likely does not satisfy~\cref{eq:soundness_over_key} and thus is rejected.
	\item \textbf{Soundness given a key}~(\cref{alg:soundness_given_key_test}): for independently sampled private keys, we sample a separate set of non-watermarked texts (human or LLM-generated) per key and screen for an abnormally high empirical false positive rate.
	      We then test whether too many keys screen positive, accounting for both the allowed fraction $\beta$ of bad keys and the probability that a sound key screens positive.
	      If they do, the scheme likely violates~\cref{eq:soundness_given_key} and is rejected.
\end{enumerate}

\subsection{A Framework for Watermark Autoresearch}
\label{sec:method:framework}

Next, we describe our harness designed around three components: a dedicated vLLM watermarking library, a coding agent environment, and the automatic verification \& evaluation pipeline~(\cref{sec:method:verification}).

\paragraph{Research harness}
We illustrate the autoresearch loop in \cref{fig:automark-overview}.
Our watermarking library is designed around four primary components: the context seeding, the watermark score, the logits transformation, and the detector, following the formalization described in~\cref{sec:method:preliminaries}.
A watermark is an implementation of all 4 components.
We defer additional design choices to~\cref{app:implementation_speed}.

The research environment is initialized with the watermarking library and the baseline evaluations.
The agent is tasked with proposing a watermarking algorithm that outperforms prior schemes in several key aspects: raw detectability and robustness to several attacks (\eg word substitution, word deletion, paraphrasing, and backtranslation) while remaining distortion-free and statistically sound.
We detail the initial prompt in \cref{app:prompt:agent}.
We also give the agent a public configuration specifying how to evaluate and verify the proposed scheme.
Once the agent completes its task, its proposed scheme is automatically evaluated (with a private configuration), and the evaluation results are added to the environment alongside the baselines.
Then, a new agent takes over with access to all prior results and code modifications.
The loop continues indefinitely.
Inside the environment, the agent has full access to the Internet (\eg it can download third-party libraries or read papers), a dedicated GPU, and API keys for third-party AI providers (primarily for experimenting with paraphrasing attacks).

To stimulate the agent even more, we run a two-phase setup.
The first agent is told it has 4 available submissions to the private evaluation before being shut down: this allows it to set up its environment, and adjust in case of unexpected failures.
Then, after the initial 4 submissions, subsequent agents are told they only have 1 submission left.
We found in~\cref{app:harness:environment} that this combined setup led to more creative solutions from the agents and less hyperparameter tuning.

\section{Evaluation}
\label{sec:eval}

In this section, we analyze the results of running our autoresearch framework across state-of-the-art agents.
In~\cref{sec:eval:trajectories} we analyze the trajectories of the agents during the autoresearch loop, while in~\cref{sec:eval:directions} we break down the schemes proposed by the agents into individual categories and assess which metrics they impact.
Lastly, in~\cref{sec:eval:best}, we select two watermarking algorithms from the autoresearch runs that outperform prior work on our evaluated metrics.

\paragraph{Experimental setup}
We compare against the most popular prior and deployed schemes: AAR, SynthID, and TextSeal.
Our private grader measures detectability and robustness as the TPR@1\%FPR on $1000$ \textsc{Llama-3.1-8B} $200$ to $300$ tokens long replies to \textsc{ELI5}~\citep{eli5} prompts, on clean text and under four attacks: word deletion, synonym substitution, back-translation, and paraphrasing.
For quality, we measure the distance to the unwatermarked model in perplexity, and in diversity (Self-BLEU) across $100$ replies to the same $10$ prompts, and, as these metrics are highly correlated, sometimes average them into a single quality metric (see \cref{app:full_results} for un-aggregated results).
Each metric is the worst over $5$ watermark keys to mitigate across-key variance.
We defer details to~\cref{app:setup}.
Visualizations of our evaluation results are available in our \href{https://www.sri.inf.ethz.ch/blog/automark}{blogpost}.

\subsection{Autoresearch Trajectories}
\label{sec:eval:trajectories}

We first focus on how the metrics evolved over the number of submissions to the evaluation harness.
\begin{figure}[t]
    \centering
    \vspace{-0.5em}
    \includegraphics[width=\linewidth]{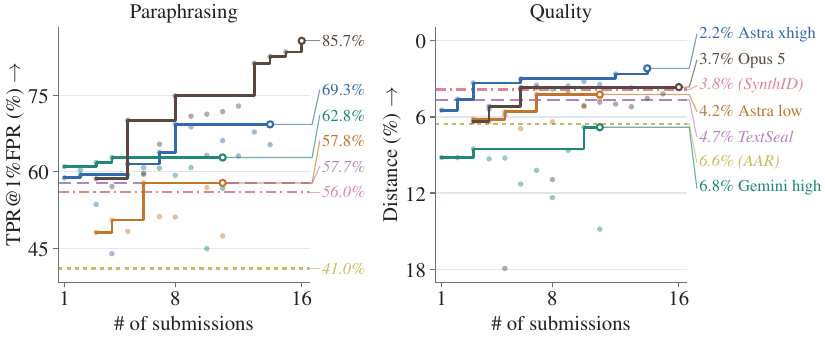}
    \vspace{-0.5em}
    \caption{Submission history for each agent along two evaluated axes: robustness to paraphrasing and quality. Each dot corresponds to a submission and the solid lines are the best value so far. The baselines are in italic, and a parenthesis means the baseline \emph{fails} our verification.}
    \label{fig:research-trajectories}
\end{figure}

\paragraph{Research agents and settings}
We use four different research agents: three frontier models \astra (xhigh), \astra (low), \opus, and a state-of-the-art efficient model \gemini (high).
We let each agent run until its submissionbudget was exhausted.
Using standard API prices, the research-agent costs are \$357.08 for \astra (xhigh), \$123.29 for \astra (low), \$435.96 for \opus, and \$13.81 for \gemini (high).

\paragraph{Main results}
\cref{fig:research-trajectories} shows the evolution of robustness to paraphrasing and of the watermark quality for different agents, with the baselines for comparison, for a total of $52$ watermark schemes being evaluated.
We find that all tested agents manage to improve their submissions over time, outperforming baselines in most cases.
For two baselines (AAR and SynthID), the evaluated configurations do not pass our soundness given a key check: as we show in~\cref{app:soundness_key} this means that their robustness is slightly inflated.
Moreover, only the submissions from \astra (both low and xhigh) are provably sound given a key.
In contrast, submissions from \opus, \gemini and the TextSeal baseline configuration satisfy the property only empirically (\ie they do not fail our check).
This means that the p-values of \astra are overly conservative, and, as we find in~\cref{app:soundness_key}, calibrating them to simply pass our empirical check increases the paraphrasing TPR of their best submission by $13$--$15$ percentage points, reaching robustness similar to that of \opus.

\paragraph{Research dynamics}
The research processes of \astra and \opus were similar, proposing a range of novel and creative watermarking ideas~(\cref{app:components}).
Their first submissions mostly consisted of tweaks to the TextSeal baseline and attempts to correlate both the private verification and evaluation results with their local evaluation set-up.
Both \astra agents then decided to submit only schemes that are provably sound given a key by multiplying their p-values by $1/\beta$~(\cref{sec:method:definition_validity}) whereas \opus opted for empirical calibration on its public verification evaluation, which in practice turned out to correlate well with the private verification.
Lastly, \gemini was significantly less creative. All its submissions consisted of hyperparameter tuning of the TextSeal watermark to improve the metrics and did not result in any new ideas.

\definecolor{summaryScalar}{HTML}{E3EFF8}
\definecolor{summaryGeometry}{HTML}{FBEAD7}
\definecolor{summarySkeleton}{HTML}{E4F1E5}
\definecolor{summarySkeletonEdge}{HTML}{7DB783}
\providecommand{\componentchipfont}{\fontsize{5.8pt}{6.8pt}\selectfont}
\begin{figure}[t]
\centering
\begin{tikzpicture}[
    stage/.style={shape=signal, signal to=east, signal pointer angle=130, fill=black!80, text=white,
        font=\scriptsize\bfseries, minimum height=0.33cm, inner xsep=1pt, inner ysep=0.5pt, text height=1.6ex, text depth=0.4ex},
    chip/.style={draw=black!30, font=\componentchipfont, inner xsep=1.5pt,
        minimum height=0.3cm, inner ysep=0.8pt, text height=1.5ex, text depth=0.35ex, rounded corners=1.2pt},
    shared/.style={chip, fill=white},
    fresh/.style={chip, fill=summaryScalar, draw=summaryScalar!60!black!50},
    fresk/.style={fresh, path picture={\fill[summarySkeletonEdge] (path picture bounding box.north west) rectangle ([xshift=1.3pt]path picture bounding box.south west);}},
    geom/.style={chip, fill=summaryGeometry, draw=summaryGeometry!60!black!50},
    skel/.style={chip, fill=summarySkeleton, draw=summarySkeleton!60!black!50},
]
\def\componentcolsep{2.5}   %
\def\componentstagesep{4}   %
\def\componentheadsep{13}   %
\def\componentrowsep{11}
\def\componentpad{3.1}      %
\gdef\componentstagex{0}
\newlength{\componentlabelwidth}

\newcommand{\stagecol}[4]{%
    \gdef\componentwa{0}\gdef\componentwb{0}
    \foreach \sty/\txt [count=\j] in {#4} {
        \pgfinterruptpicture\global\setbox1=\hbox{\componentchipfont\txt}\endpgfinterruptpicture
        \componentlabelwidth=\wd1
        \pgfmathsetmacro{\labelwidth}{\componentlabelwidth + \componentpad}
        \ifnum\j>#3
            \pgfmathsetmacro{\colwidth}{max(\componentwb, \labelwidth)}\xdef\componentwb{\colwidth}
        \else
            \pgfmathsetmacro{\colwidth}{max(\componentwa, \labelwidth)}\xdef\componentwa{\colwidth}
        \fi
    }
    \pgfmathsetmacro{\stagewidth}{\componentwb > 0 ? \componentwa + \componentcolsep + \componentwb : \componentwa}
    \node[stage, #2, minimum width=\stagewidth pt, anchor=north west] at (\componentstagex pt, 0) {#1};
    \foreach \sty/\txt [count=\j] in {#4} {
        \pgfmathsetmacro{\chipx}{\j > #3 ? \componentstagex + \componentwa + \componentcolsep : \componentstagex}
        \pgfmathsetmacro{\chiptextwidth}{(\j > #3 ? \componentwb : \componentwa) - 3}
        \pgfmathsetmacro{\chipy}{-\componentheadsep - (\j > #3 ? \j - #3 - 1 : \j - 1) * \componentrowsep}
        \node[\sty, anchor=north west] at (\chipx pt, \chipy pt) {\makebox[\chiptextwidth pt][l]{\txt}};
    }
    \pgfmathsetmacro{\nextx}{\componentstagex + \stagewidth + \componentstagesep}
    \xdef\componentstagex{\nextx}
}
\stagecol{Unit}{signal from=nowhere}{2}{shared/{\textit{Token}$^\ast$}, shared/Lexical group}
\stagecol{Context seeding}{signal from=west}{5}{shared/{\textit{Fixed-length}$^\ast$}, shared/Unordered,
    shared/Adaptive length, shared/Global + local, shared/Global first use,
    shared/Anchored, shared/Suffix cascade, shared/Tally \& ladder,
    skel/Last important, skel/Interleaved}
\stagecol{Score}{signal from=west}{5}{shared/Prelude, fresk/Gaussian copula, fresk/Tail indicator,
    fresk/Private sampling, fresk/Channels, fresk/Channel variants,
    geom/Gaussian direction, geom/Spherical, geom/Wheel}
\stagecol{Logits transform}{signal from=west}{5}{shared/{\textit{Gumbel race}$^\ast$}, shared/BandRace, shared/PoolRace,
    shared/MomentBridge, shared/QuantileRace, shared/TransportRace, shared/QuotaTransport,
    shared/ResidualRace, skel/Content gate, skel/Interleaved gate}
\stagecol{Detection}{signal from=west, signal to=nowhere}{5}{fresk/{Gamma/Gaussian}, fresk/{Clustering/excess},
    fresk/Quantile, fresk/Weighted Gamma, fresk/Position evidence,
    geom/Gaussian $\chi^2$, geom/Resultant + caps,
    geom/Anchor distance, geom/Predictive Gamma,
    geom/Circle search, geom/Quantile + wheel}
\end{tikzpicture}
\vspace{-0.2in}
\caption{Atomic components extracted from the agents' watermarking schemes. Colors indicate the family (\colorbox{summaryScalar}{Fresh race}, \colorbox{summaryGeometry}{Private geometry}, \colorbox{summarySkeleton}{Skeleton}) where components of the same family are interchangeable, and white components are compatible with every family. Skeleton schemes reuse the fresh-race scores and detectors (green edge). Components marked $^\ast$ come from prior work.}
\label{fig:component-pipeline}
\end{figure}

\subsection{Breakdown of the Agent Research Directions}
\label{sec:eval:directions}

Next, we manually go through the schemes proposed by all agents, decompose them into five components (the four components from \cref{sec:method:preliminaries} and the \emph{watermark unit}, a generalization introduced by the agents), and evaluate the impact of each component on watermark properties.
We summarize the decomposition in \cref{fig:component-pipeline}, and provide a detailed explanation of each component in \cref{app:components}.

\paragraph{Watermark unit}
Most watermarks operate at the token level: each token is pseudo-randomly scored and the watermark modifies the logits according to both the score and the model probability distribution (\cref{sec:method:preliminaries}).
We find that the agents proposed to generalize this idea: let $\mathcal{G}: \Sigma \rightarrow \mathbb{N}$ be a mapping from the token space to a set of equivalence classes.
In any watermark algorithm, one can replace a token $v$ with its class $\mathcal{G}(v)$ and, given the next-token probability distribution $p_t \in \Delta(\Sigma)$, define the probability of each class $g \in \mathcal{G}(\Sigma)$ as $\Pr[g] = \sum_{v \in \Sigma} p_t(v) \indicator\{\mathcal{G}(v) = g\}$.
If the watermark decides to sample a class $g$, the actual token is then sampled according to the residual probability
\begin{equation} \label{eq:watermark-unit-sampling}
	\forall v \in \Sigma, \Pr[v \mid g] = \indicator\{\mathcal{G}(v) = g\}\frac{p_t(v)}{\sum_{v'\in\Sigma} p_t(v') \indicator\{\mathcal{G}(v') = g\}}.
\end{equation}
For instance, using a \texttt{Lexical group} as $\mathcal{G}$ naturally increases the scheme robustness but slighly lowers quality: if tokens from the same lexical group are swapped, the context remains unchanged.

\begin{wrapfigure}{r}{0.42\textwidth}
    \centering
    \includegraphics[width=\linewidth]{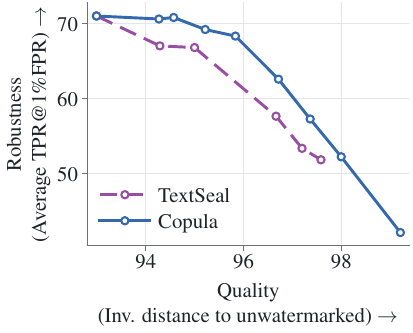}
    \vspace{-2em}
    \caption{Robustness-quality for TextSeal dual key routing and \texttt{Gaussian copula}. Robustness is the average of all 4 attacks' TPR@1\%FPR.}
    \label{fig:textseal-copula}
    \vspace{-0.5em}
\end{wrapfigure}

\paragraph{Context seeding}
For context seeding, agents explore the trade-off alluded to in~\citet{kgw2}: increasing the context size lowers the robustness but improves the quality.
Also, for small contexts, because of deduplication during detection (to guarantee soundness over the key~(\cref{sec:method:preliminaries})), lowering the context size may also hurt detectability.
Historically, prior works have converged on hashing the tuple of the candidate token and the $k$ preceding tokens (\texttt{Fixed-length} in \cref{fig:component-pipeline}).
Agents iterated upon this idea along three main axes.
The first idea is that tokens should not be treated equally: less frequent tokens should be dynamically assigned a lower $k$.
This increases robustness without impacting quality significantly (because the distortion-free guarantee ensures that quality degrades only when contexts are repeated across texts).
The second idea is to dynamically increase the context size: the first time we use a tuple as a context, it is of size $1$. The next time we are about to use the same tuple, we increase the context size until the context has not been seen yet.
This keeps the robustness benefits of short contexts without lowering the detectability.
The third axis is more creative: assign tokens to categories, and when sampling a token from a given category, use a fixed-window context over the previous tokens from that category only.
This increases robustness: edits in one category do not change the context for other categories.

\begin{table}[t]
\centering
\vspace{-1em}
\definecolor{selectedBest}{HTML}{E2F0D9}
\caption{Evaluation of selected agent schemes. Each cell reports English / Chinese results. 
We \textbf{bold} all results that outperform the best baseline, and \colorbox{selectedBest}{highlight} the best result in each column.}
\label{tab:selected-schemes}
\scriptsize
\setlength{\tabcolsep}{2pt}
\resizebox{\linewidth}{!}{%
\begin{tabular}{@{}lrrrrrrrr@{}}
\toprule
 & & \multicolumn{4}{c}{Robustness TPR@1  $\uparrow$} & \multicolumn{3}{c}{Quality deviation (\%) $\downarrow$} \\
\cmidrule(lr){3-6}\cmidrule(l){7-9}
Scheme & Clean TPR@1 $\uparrow$ & Deletion & Substitution & Back-trans. & Paraphrase & $\Delta$PPL & $\Delta$SB-2 & $\Delta$SB-3 \\
\midrule
\multicolumn{9}{@{}l}{\textbf{Llama-3.1-8B}} \\
AAR & 93 / 89 & 7 / 16 & 8 / 22 & 66 / 17 & 57 / 42 & 3.1 / 3.0 & 8.3 / 6.5 & 20 / 15 \\
SynthID & 95 / 90 & 12 / 25 & 10 / 31 & 84 / 24 & 71 / 53 & 3.0 / 1.0 & 3.5 / 3.3 & 9.3 / 8.0 \\
TextSeal (high) & 95 / 91 & 15 / 31 & 12 / 38 & 78 / 28 & 66 / 56 & 2.9 / 4.4 & 3.8 / 3.6 & 9.3 / 8.7 \\
\addlinespace[2pt]
ConcordMark (low) & 94 / 89 & \textbf{30} / \textbf{54} & \textbf{42} / \textbf{62} & \textbf{86} / \textbf{37} & \textbf{75} / \textbf{65} & \textbf{0.8} / 1.8 & \textbf{2.2} / \textbf{2.7} & \textbf{5.1} / \textbf{5.7} \\
ConcordMark (high) & \cellcolor{selectedBest}\textbf{96} / \textbf{91} & \textbf{34} / \textbf{61} & \cellcolor{selectedBest}\textbf{48} / \textbf{69} & \cellcolor{selectedBest}\textbf{89} / \textbf{55} & \cellcolor{selectedBest}\textbf{82} / \textbf{77} & \textbf{0.8} / 1.8 & \textbf{2.2} / \textbf{2.7} & \textbf{5.1} / \textbf{5.7} \\
RotationFirst & 91 / 89 & \cellcolor{selectedBest}\textbf{53} / \textbf{60} & \textbf{40} / \textbf{62} & 79 / \textbf{37} & 69 / \textbf{61} & \cellcolor{selectedBest}\textbf{0.5} / \textbf{0.6} & \cellcolor{selectedBest}\textbf{0.9} / \textbf{1.2} & \cellcolor{selectedBest}\textbf{2.1} / \textbf{2.6} \\
\midrule\midrule
\multicolumn{9}{@{}l}{\textbf{Qwen3.8-27B}} \\
AAR & 95 / 84 & 4 / 6 & 5 / 11 & 50 / 15 & 61 / 65 & 2.1 / 1.6 & 8.7 / 7.4 & 18 / 16 \\
SynthID & 97 / 88 & 6 / 10 & 7 / 16 & 60 / 21 & 68 / 72 & 5.9 / 1.6 & 3.7 / 4.0 & 8.5 / 9.0 \\
TextSeal (high) & 97 / 88 & 7 / 10 & 8 / 17 & 65 / 26 & 70 / 74 & 2.7 / 1.2 & 3.3 / 3.4 & 7.5 / 7.5 \\
\addlinespace[2pt]
ConcordMark (low) & 96 / 85 & \textbf{16} / \textbf{23} & \textbf{21} / \textbf{34} & \textbf{71} / \textbf{29} & \textbf{74} / 72 & \cellcolor{selectedBest}2.2 / \textbf{0.8} & \textbf{1.5} / \textbf{2.0} & \textbf{3.3} / \textbf{4.1} \\
ConcordMark (high) & \cellcolor{selectedBest}\textbf{98} / \textbf{90} & \cellcolor{selectedBest}\textbf{20} / \textbf{31} & \cellcolor{selectedBest}\textbf{25} / \textbf{43} & \cellcolor{selectedBest}\textbf{81} / \textbf{47} & \cellcolor{selectedBest}\textbf{83} / \textbf{82} & \cellcolor{selectedBest}2.2 / \textbf{0.8} & \textbf{1.5} / \textbf{2.0} & \textbf{3.3} / \textbf{4.1} \\
RotationFirst & 86 / 81 & \textbf{25} / \textbf{24} & \textbf{16} / \textbf{35} & 54 / 22 & 57 / 64 & 2.2 / \textbf{0.9} & \cellcolor{selectedBest}\textbf{0.3} / \textbf{0.7} & \cellcolor{selectedBest}\textbf{0.7} / \textbf{1.6} \\
\bottomrule
\end{tabular}%
}\vspace{-1.2em}

\end{table}

\paragraph{Score}
The score refers to the pseudo-random variable used by the logits transformation.
For instance, with Red-Green watermarks, the scores are the colors of the tokens.
Without loss of generality, as in \cref{sec:method:preliminaries}, we refer to the \emph{seeded} watermark score as $U_t \sim \mathcal{U}([0,1])^{|\Sigma|}$ and the final score used in the logits transformation as $\tilde{U}_t$.
With different score designs, the agents tried to increase the quality while minimizing the impact on detectability.
To do that, they consistently add a small amount of true randomness to the watermark.
This is conceptually similar to the approaches by \citet{textseal,synthid} which also add true randomness to the watermark sampling to improve quality.
Yet, we find that agents explore concepts that substantially differ from prior works.
With the \texttt{Gaussian copula}, they propose sampling random Gaussian noise $\eta_t \sim \mathcal{N}(0,I_{|\Sigma|})$ and applying the following transformation
\begin{equation} \label{eq:gaussian_copula}
	\tilde{U}_t = \Phi(\rho \Phi^{-1}(U_t) + \sqrt{1-\rho^2} \eta_t),
\end{equation}
where $\rho \in [0,1]$ is a scheme parameter that trades off detectability for quality (in the limit, $\rho = 0$ has no watermark).
In~\cref{fig:textseal-copula}, we compare this approach to TextSeal's random key switching (keeping all other watermarking components fixed) and find that it is strictly better: it outperforms TextSeal along all evaluated dimensions.
With \texttt{Channels}, we use a finite pool of private keys and for each request randomly decide which key to use.
Detection is run with every key and uses the Bonferroni correction.
\texttt{Private geometries} are more complex, and we explain them further in \cref{app:components}.

\begin{wrapfigure}{r}{0.42\textwidth}
    \centering
    \vspace{-0.25in}
    \includegraphics[width=\linewidth]{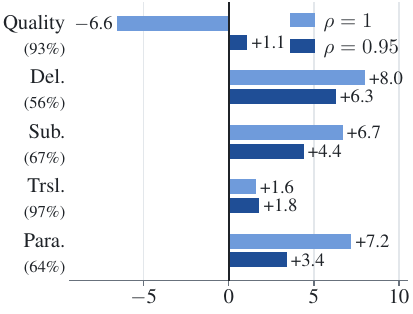}
    \vspace{-0.25in}
    \caption{\texttt{ResidualRace} standalone or with \texttt{Gaussian copula} $\rho=0.95$ to match its quality with the \texttt{Gumbel race}.}
    \vspace{-0.35in}
    \label{fig:residual-clock}
\end{wrapfigure}

\paragraph{Logits transformation}
For logits transformations, most proposed schemes settle on Gumbel-max sampling (\texttt{Gumbel race}) from~\citet{aar}: as proven
in~\citet{unified_wm_framework}, it maximizes watermark detectability.
Yet, we find that the agents commonly propose a Gumbel-max sampling extension via residual clocks (\texttt{ResidualRace}). 
Given a context seed $a_t$, let $U_{a_t}$ be the uniform watermark score.
The first time $a_t$ is visited, we initialize residual clocks $R_{a_t} = -\log U_{a_t}$; otherwise, we use previously computed scores.
Then we sample the next token according to the exponential race:
\begin{equation}
	\label{eq:exponential_race}
	\begin{gathered}
		\tau_t=\min_u\frac{R_{a_t}(u)}{p_t(u)},\qquad
		u_* = \arg\min_u\frac{R_{a_t}(u)}{p_t(u)},\\
		R_{a_t}(u)\leftarrow R_{a_t}(u)-\tau_t p_t(u).
	\end{gathered}
\end{equation}
and replace the winning token's score with a fresh exponential random variable: $R_{a_t}(u_*) \sim \operatorname{Exp}(1)$.
In~\cref{fig:residual-clock}, we compare \texttt{ResidualRace} against the \texttt{Gumbel race}, keeping all other components fixed.
\texttt{ResidualRace} improves robustness against every attack at a slight cost to quality.
Yet, adding \texttt{Gaussian copula} noise to the initial clocks recovers the quality of the \texttt{Gumbel race} while retaining most of the robustness gains.
We find that most other logits transforms proposed by agents are designed to only increase the quality of the schemes at the expense of detectability.
As we show in \cref{app:components} this generally offer worse trade-offs than \eg using the \texttt{Gaussian copula}.

\begin{wrapfigure}{r}{0.5\textwidth}
    \centering
    \vspace{-1.3em}
    \includegraphics[width=\linewidth]{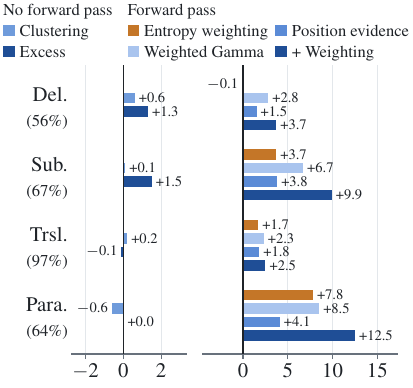}
    \vspace{-1em}
    \caption{Comparison of detector variants without forward pass (left) and with forward pass (right).}
    \label{fig:detectors}
    \vspace{-3em}
\end{wrapfigure}

\paragraph{Detection}
Most agents initially converged independently on the \texttt{Gamma} detector.
Let $\omega \in \Sigma^*$ be a sequence of tokens and $U_t(\omega_t)$ the corresponding uniform watermark score at position $t$. Then we have
\begin{equation}
	S_{\text{Gamma}}(\omega) = \sum_{t=1}^{|\omega|} -\log(1-U_t(\omega_t))
\end{equation}
which follows a $\text{Gamma}(|\omega|,1)$ distribution under the null.
As we show in~\cref{fig:detectors} (left), they further improve it with \texttt{Clustering} and \texttt{Excess} thresholding.
Lastly, using a model forward pass during detection to estimate the model probabilities can further improve detectability by a significant margin over \texttt{entropy weighting}, a prior method that also relies on a forward pass during detection (\cref{fig:detectors}, right).
For \texttt{Private geometries}, which use a different class of detectors, we defer their full evaluation to \cref{app:components}.

\subsection{Independent Evaluation of the Best Directions}
\label{sec:eval:best}

Next, we highlight two state-of-the-art schemes proposed by the agents that outperform prior works: ConcordMark as a robust watermarking scheme, and RotationFirst as a scheme with a strong quality focus.
We give an in-depth description of both schemes and provide additional evaluation in~\cref{app:selected_schemes}.

\paragraph{Experimental setup}
To mitigate overfitting due to the autoresearch process, we use a different evaluation protocol detailed in~\cref{app:setup:heldout}.
Notably, we generate $1000$ replies to English prompts, and $1000$ replies to Chinese prompts from WildChat~\citep{wildchat}, use \textsc{DeepSeek-v4 Flash} for backtranslation and paraphrasing, and evaluate replies generated by both \textsc{Llama-3.1-8B} and \textsc{Qwen3.8-27B}.
Lastly, we pool the results over $5$ watermarking seeds and report the worst result.
For consistency with the baselines, we do not apply the soundness given a key correction~(\cref{sec:method:definition_validity}).

\paragraph{Main results}
\cref{tab:selected-schemes} compares the two selected agent schemes with the three baselines.
We bold the results from the agent schemes that are better than the best baseline, and highlight the best results of each column in green.
Each cell reports the metrics on English text (left) and Chinese text (right).
For each scheme, high means the detector requires a model forward pass.
We see that ConcordMark (high) is Pareto better than all the baselines: on almost all metrics, it outperforms the best baseline.
RotationFirst significantly outperforms all schemes on all quality metrics: its diversity is very close to the unwatermarked model.
At the same time, it keeps relatively strong detection and robustness metrics making this scheme very suitable for model providers that desire strong watermarking with the smallest impact on quality.
Importantly, for both agent schemes, the results from the autoresearch evaluation~(\cref{fig:research-trajectories}) transfer to new settings with different prompts and attackers.
This suggests that autoresearch, with adequate verification, is a reliable method for improving watermarking schemes.

\vspace{-0.2em}
\section{Conclusion and Limitations}
\label{sec:conclusion}
\vspace{-0.2em}

In this work, we enable for the first time autoresearch for LLM watermarking with rigorous validity criteria~(\cref{sec:method:definition_validity}), statistical tests to verify them~(\cref{sec:method:verification}), and a tailored harness~(\cref{sec:method:framework}).
We show that frontier agents go beyond hyperparameter tuning to introduce novel ideas~(\cref{sec:eval:directions}), with two of their schemes, ConcordMark and RotationFirst, outperforming prior work on held-out data~(\cref{sec:eval:best}).

\paragraph{Limitations}
In autoresearch, the key challenge is to build powerful verification that agents can't exploit, and that faithfully transcribes the desired property (\ie in our case a watermark that is robust and with controlled false positives given a key).
For watermarking, there are many desirable properties we did not include in this work \eg security, radioactivity, applicability to open-weights setting, or extension to the multibit setting.
We believe that, future research, should not focus on establishing new schemes but rather on how to build verifier and evaluation for those properties.

\message{^^JLASTBODYPAGE \thepage^^J}

\subsection*{LLM Usage}

The goal of this work is to enable coding agents to discover novel and better watermarking schemes: hence, all the new watermark proposals in this work have been designed by fully autonomous coding agents.
Yet, organizing the schemes into individual components and understanding their impact on watermark properties have required significant manual efforts.
Notably, all ablation evaluations of the components have been designed by humans and executed with AI-assistance.
The autoresearch harness used in this work has been coded with AI-assistance.
The main paper has been written by hand, with AI-assistance for polishing the writing and the visuals.
The proofs have been LLM-generated, then proofread and edited by humans.

\subsection*{Ethics Statement}

This paper enables autonomous agents to discover new LLM watermarking schemes, which we believe has a positive societal impact: better watermarks support the provenance, auditing, and mitigation of large-scale misuse of AI-generated text.
While our verification tests screen for schemes with unsound false positive rates~(\cref{sec:method:verification}), LLM watermarks should not be used as sole evidence in high-stakes decisions.
We also acknowledge that publicly disclosing new schemes, along with a detailled component-level analysis, may facilitate targeted attacks on them.
Yet we argue that the benefits of advancing research in this area clearly outweigh the risks.

\subsection*{Reproducibility Statement}

Our autoresearch harness is described in~\cref{sec:method:framework}, and all details of the verification, the evaluation, the research agents (including their versions and reasoning efforts), and the baselines are given in~\cref{app:setup}, with the full agent prompt in~\cref{app:prompt}.
We fully specify the agent-generated components in~\cref{app:components} and the two selected schemes, ConcordMark and RotationFirst, in~\cref{app:selected_schemes}.
We also release our code, including the harness, the watermarking library, each agent final environment, a standalone implementation of all described components, and a standalone optimized implementation for ConcordMark and RotationFirst.

\clearpage
\bibliographystyle{iclr2027_conference}
\bibliography{references}

\begin{thebibliography}{34}
\providecommand{\natexlab}[1]{#1}
\providecommand{\url}[1]{\texttt{#1}}
\expandafter\ifx\csname urlstyle\endcsname\relax
  \providecommand{\doi}[1]{doi: #1}\else
  \providecommand{\doi}{doi: \begingroup \urlstyle{rm}\Url}\fi

\bibitem[Aaronson(2023)]{aar}
Scott Aaronson.
\newblock Watermarking of large language models.
\newblock Talk at the Simons Institute for the Theory of Computing, August 2023.
\newblock URL \url{https://simons.berkeley.edu/talks/scott-aaronson-ut-austin-openai-2023-08-17}.

\bibitem[{Anthropic}(2026)]{anthropic}
{Anthropic}.
\newblock How {Claude}'s text watermark works.
\newblock Anthropic News, August 2026.
\newblock URL \url{https://www.anthropic.com/news/claude-text-watermark}.

\bibitem[Bahri \& Wieting(2026)Bahri and Wieting]{blackbox_wm}
Dara Bahri and John Wieting.
\newblock A watermark for black-box language models.
\newblock \emph{Transactions on Machine Learning Research}, 2026.
\newblock URL \url{https://arxiv.org/abs/2410.02099}.

\bibitem[Chan et~al.(2025)Chan, Chowdhury, Jaffe, Aung, Sherburn, Mays, Starace, Liu, Maksin, Patwardhan, Weng, and M{\k{a}}dry]{mlebench}
Jun~Shern Chan, Neil Chowdhury, Oliver Jaffe, James Aung, Dane Sherburn, Evan Mays, Giulio Starace, Kevin Liu, Leon Maksin, Tejal Patwardhan, Lilian Weng, and Aleksander M{\k{a}}dry.
\newblock {MLE}-bench: Evaluating machine learning agents on machine learning engineering.
\newblock In \emph{The Thirteenth International Conference on Learning Representations}, 2025.
\newblock URL \url{https://openreview.net/forum?id=6s5uXNWGIh}.

\bibitem[Chen et~al.(2025)Chen, Wu, Guo, and Huang]{mcmark}
Ruibo Chen, Yihan Wu, Junfeng Guo, and Heng Huang.
\newblock Improved unbiased watermark for large language models.
\newblock In \emph{Proceedings of the 63rd Annual Meeting of the Association for Computational Linguistics (Volume 1: Long Papers)}, pp.\  20587--20601, Vienna, Austria, July 2025. Association for Computational Linguistics.
\newblock \doi{10.18653/v1/2025.acl-long.1005}.
\newblock URL \url{https://aclanthology.org/2025.acl-long.1005/}.

\bibitem[Dathathri et~al.(2024)Dathathri, See, Ghaisas, Huang, McAdam, Welbl, Bachani, Kaskasoli, Stanforth, Matejovicova, Hayes, Vyas, Merey, Brown-Cohen, Bunel, Balle, Cemgil, Ahmed, Stacpoole, Shumailov, Baetu, Gowal, Hassabis, and Kohli]{synthid}
Sumanth Dathathri, Abigail See, Sumedh Ghaisas, Po-Sen Huang, Rob McAdam, Johannes Welbl, Vandana Bachani, Alex Kaskasoli, Robert Stanforth, Tatiana Matejovicova, Jamie Hayes, Nidhi Vyas, Majd~Al Merey, Jonah Brown-Cohen, Rudy Bunel, Borja Balle, Taylan Cemgil, Zahra Ahmed, Kitty Stacpoole, Ilia Shumailov, Ciprian Baetu, Sven Gowal, Demis Hassabis, and Pushmeet Kohli.
\newblock Scalable watermarking for identifying large language model outputs.
\newblock \emph{Nature}, 634\penalty0 (8035):\penalty0 818--823, October 2024.
\newblock \doi{10.1038/s41586-024-08025-4}.
\newblock URL \url{https://doi.org/10.1038/s41586-024-08025-4}.

\bibitem[Davies(1987)]{davies1987}
Robert~B. Davies.
\newblock Hypothesis testing when a nuisance parameter is present only under the alternative.
\newblock \emph{Biometrika}, 74\penalty0 (1):\penalty0 33--43, 1987.
\newblock \doi{10.1093/biomet/74.1.33}.

\bibitem[{European Parliament and Council of the European Union}(2024)]{euiact}
{European Parliament and Council of the European Union}.
\newblock Regulation ({EU}) 2024/1689 of the european parliament and of the council of 13 june 2024 laying down harmonised rules on artificial intelligence and amending regulations ({EC}) no 300/2008, ({EU}) no 167/2013, ({EU}) no 168/2013, ({EU}) 2018/858, ({EU}) 2018/1139 and ({EU}) 2019/2144 and directives 2014/90/{EU}, ({EU}) 2016/797 and ({EU}) 2020/1828 ({Artificial Intelligence Act}).
\newblock Official Journal of the European Union, June 2024.
\newblock URL \url{https://eur-lex.europa.eu/eli/reg/2024/1689/oj/eng}.

\bibitem[Fan et~al.(2019)Fan, Jernite, Perez, Grangier, Weston, and Auli]{eli5}
Angela Fan, Yacine Jernite, Ethan Perez, David Grangier, Jason Weston, and Michael Auli.
\newblock {ELI5}: Long form question answering.
\newblock In \emph{Proceedings of the 57th Annual Meeting of the Association for Computational Linguistics}, pp.\  3558--3567, Florence, Italy, July 2019. Association for Computational Linguistics.
\newblock \doi{10.18653/v1/P19-1346}.
\newblock URL \url{https://aclanthology.org/P19-1346/}.

\bibitem[Fernandez et~al.(2023)Fernandez, Chaffin, Tit, Chappelier, and Furon]{threebricks}
Pierre Fernandez, Antoine Chaffin, Karim Tit, Vivien Chappelier, and Teddy Furon.
\newblock Three bricks to consolidate watermarks for large language models.
\newblock In \emph{2023 IEEE International Workshop on Information Forensics and Security (WIFS)}, pp.\  1--6. IEEE, December 2023.
\newblock \doi{10.1109/WIFS58808.2023.10374576}.
\newblock URL \url{https://doi.org/10.1109/WIFS58808.2023.10374576}.

\bibitem[Gloaguen et~al.(2026)Gloaguen, Staab, Jovanovi{\'c}, and Vechev]{unified_wm_framework}
Thibaud Gloaguen, Robin Staab, Nikola Jovanovi{\'c}, and Martin Vechev.
\newblock A unified framework for {LLM} watermarks.
\newblock \emph{arXiv preprint arXiv:2602.06754}, February 2026.
\newblock \doi{10.48550/arXiv.2602.06754}.
\newblock URL \url{https://arxiv.org/abs/2602.06754}.

\bibitem[Jovanovi{\'c} et~al.(2024)Jovanovi{\'c}, Staab, and Vechev]{ws}
Nikola Jovanovi{\'c}, Robin Staab, and Martin Vechev.
\newblock Watermark stealing in large language models.
\newblock In \emph{Proceedings of the 41st International Conference on Machine Learning}, volume 235 of \emph{Proceedings of Machine Learning Research}, pp.\  22570--22593. PMLR, 2024.
\newblock URL \url{https://proceedings.mlr.press/v235/jovanovic24a.html}.

\bibitem[Karpathy(2026)]{autoresearch}
Andrej Karpathy.
\newblock autoresearch: {AI} agents running research on single-{GPU} nanochat training automatically.
\newblock GitHub repository, March 2026.
\newblock URL \url{https://github.com/karpathy/autoresearch}.

\bibitem[Kirchenbauer et~al.(2023)Kirchenbauer, Geiping, Wen, Katz, Miers, and Goldstein]{kgw}
John Kirchenbauer, Jonas Geiping, Yuxin Wen, Jonathan Katz, Ian Miers, and Tom Goldstein.
\newblock A watermark for large language models.
\newblock In \emph{Proceedings of the 40th International Conference on Machine Learning}, volume 202 of \emph{Proceedings of Machine Learning Research}, pp.\  17061--17084. PMLR, 2023.
\newblock URL \url{https://proceedings.mlr.press/v202/kirchenbauer23a.html}.

\bibitem[Kirchenbauer et~al.(2024)Kirchenbauer, Geiping, Wen, Shu, Saifullah, Kong, Fernando, Saha, Goldblum, and Goldstein]{kgw2}
John Kirchenbauer, Jonas Geiping, Yuxin Wen, Manli Shu, Khalid Saifullah, Kezhi Kong, Kasun Fernando, Aniruddha Saha, Micah Goldblum, and Tom Goldstein.
\newblock On the reliability of watermarks for large language models.
\newblock In \emph{The Twelfth International Conference on Learning Representations}, pp.\  49660--49704, 2024.
\newblock URL \url{https://openreview.net/forum?id=DEJIDCmWOz}.

\bibitem[Kuditipudi et~al.(2024)Kuditipudi, Thickstun, Hashimoto, and Liang]{kth}
Rohith Kuditipudi, John Thickstun, Tatsunori Hashimoto, and Percy Liang.
\newblock Robust distortion-free watermarks for language models.
\newblock \emph{Transactions on Machine Learning Research}, 2024.
\newblock URL \url{https://openreview.net/forum?id=FpaCL1MO2C}.

\bibitem[Lange et~al.(2026)Lange, Imajuku, and Cetin]{shinkaevolve}
Robert~Tjarko Lange, Yuki Imajuku, and Edoardo Cetin.
\newblock {ShinkaEvolve}: Towards open-ended and sample-efficient program evolution.
\newblock In \emph{The Fourteenth International Conference on Learning Representations}, 2026.
\newblock URL \url{https://arxiv.org/abs/2509.19349}.

\bibitem[Lee et~al.(2024)Lee, Hong, Ahn, Hong, Lee, Yun, Shin, and Kim]{sweet}
Taehyun Lee, Seokhee Hong, Jaewoo Ahn, Ilgee Hong, Hwaran Lee, Sangdoo Yun, Jamin Shin, and Gunhee Kim.
\newblock Who wrote this code? watermarking for code generation.
\newblock In \emph{Proceedings of the 62nd Annual Meeting of the Association for Computational Linguistics (Volume 1: Long Papers)}, pp.\  4890--4911, Bangkok, Thailand, August 2024. Association for Computational Linguistics.
\newblock \doi{10.18653/v1/2024.acl-long.268}.
\newblock URL \url{https://aclanthology.org/2024.acl-long.268/}.

\bibitem[Lu et~al.(2024)Lu, Holt, Fanconi, Chan, Foerster, van~der Schaar, and Lange]{discopop}
Chris Lu, Samuel Holt, Claudio Fanconi, Alex~J. Chan, Jakob Foerster, Mihaela van~der Schaar, and Robert~Tjarko Lange.
\newblock Discovering preference optimization algorithms with and for large language models.
\newblock In \emph{Advances in Neural Information Processing Systems}, volume~37, 2024.
\newblock URL \url{https://arxiv.org/abs/2406.08414}.

\bibitem[Lu et~al.(2026)Lu, Lu, Lange, Yamada, Hu, Foerster, Ha, and Clune]{ai_scientist}
Chris Lu, Cong Lu, Robert~Tjarko Lange, Yutaro Yamada, Shengran Hu, Jakob Foerster, David Ha, and Jeff Clune.
\newblock Towards end-to-end automation of {AI} research.
\newblock \emph{Nature}, 651:\penalty0 914--919, March 2026.
\newblock \doi{10.1038/s41586-026-10265-5}.
\newblock URL \url{https://doi.org/10.1038/s41586-026-10265-5}.

\bibitem[Novikov et~al.(2025)Novikov, V{\~u}, Eisenberger, Dupont, Huang, Wagner, Shirobokov, Kozlovskii, Ruiz, Mehrabian, Kumar, See, Chaudhuri, Holland, Davies, Nowozin, Kohli, and Balog]{alphaevolve}
Alexander Novikov, Ng{\^a}n V{\~u}, Marvin Eisenberger, Emilien Dupont, Po-Sen Huang, Adam~Zsolt Wagner, Sergey Shirobokov, Borislav Kozlovskii, Francisco J.~R. Ruiz, Abbas Mehrabian, M.~Pawan Kumar, Abigail See, Swarat Chaudhuri, George Holland, Alex Davies, Sebastian Nowozin, Pushmeet Kohli, and Matej Balog.
\newblock Alphaevolve: A coding agent for scientific and algorithmic discovery.
\newblock \emph{arXiv preprint arXiv:2506.13131}, 2025.

\bibitem[{OpenAI}(2026)]{openai2026euler}
{OpenAI}.
\newblock Finite time blowup for the euler equation, September 2026.
\newblock URL \url{https://cdn.openai.com/pdf/315b36cd-ec98-4023-8342-93345194ece1/euler.pdf}.
\newblock Accessed: 2026-09-22.

\bibitem[Pan et~al.(2024)Pan, Liu, He, Gao, Zhao, Lu, Zhou, Liu, Hu, Wen, King, and Yu]{markllm}
Leyi Pan, Aiwei Liu, Zhiwei He, Zitian Gao, Xuandong Zhao, Yijian Lu, Binglin Zhou, Shuliang Liu, Xuming Hu, Lijie Wen, Irwin King, and Philip~S. Yu.
\newblock {MarkLLM}: An open-source toolkit for {LLM} watermarking.
\newblock In \emph{Proceedings of the 2024 Conference on Empirical Methods in Natural Language Processing: System Demonstrations}, pp.\  61--71, Miami, Florida, USA, November 2024. Association for Computational Linguistics.
\newblock \doi{10.18653/v1/2024.emnlp-demo.7}.
\newblock URL \url{https://aclanthology.org/2024.emnlp-demo.7/}.

\bibitem[Piet et~al.(2025)Piet, Sitawarin, Fang, Mu, and Wagner]{markmywords}
Julien Piet, Chawin Sitawarin, Vivian Fang, Norman Mu, and David Wagner.
\newblock Mark my words: Analyzing and evaluating language model watermarks.
\newblock In \emph{2025 IEEE Conference on Secure and Trustworthy Machine Learning (SaTML)}, pp.\  68--91. IEEE, April 2025.
\newblock \doi{10.1109/SaTML64287.2025.00012}.
\newblock URL \url{https://doi.org/10.1109/SaTML64287.2025.00012}.

\bibitem[Raffel et~al.(2020)Raffel, Shazeer, Roberts, Lee, Narang, Matena, Zhou, Li, and Liu]{c4}
Colin Raffel, Noam Shazeer, Adam Roberts, Katherine Lee, Sharan Narang, Michael Matena, Yanqi Zhou, Wei Li, and Peter~J. Liu.
\newblock Exploring the limits of transfer learning with a unified text-to-text transformer.
\newblock \emph{Journal of Machine Learning Research}, 21\penalty0 (140):\penalty0 1--67, 2020.

\bibitem[Romera-Paredes et~al.(2024)Romera-Paredes, Barekatain, Novikov, Balog, Kumar, Dupont, Ruiz, Ellenberg, Wang, Fawzi, Kohli, and Fawzi]{funsearch}
Bernardino Romera-Paredes, Mohammadamin Barekatain, Alexander Novikov, Matej Balog, M.~Pawan Kumar, Emilien Dupont, Francisco J.~R. Ruiz, Jordan~S. Ellenberg, Pengming Wang, Omar Fawzi, Pushmeet Kohli, and Alhussein Fawzi.
\newblock Mathematical discoveries from program search with large language models.
\newblock \emph{Nature}, 625\penalty0 (7995):\penalty0 468--475, 2024.
\newblock \doi{10.1038/s41586-023-06924-6}.
\newblock URL \url{https://doi.org/10.1038/s41586-023-06924-6}.

\bibitem[Salmon et~al.(2011)Salmon, Moraes, Dror, and Shaw]{random123}
John~K. Salmon, Mark~A. Moraes, Ron~O. Dror, and David~E. Shaw.
\newblock Parallel random numbers: As easy as 1, 2, 3.
\newblock In \emph{Proceedings of the 2011 International Conference for High Performance Computing, Networking, Storage and Analysis}, SC '11, pp.\  1--12. ACM, November 2011.
\newblock \doi{10.1145/2063384.2063405}.
\newblock URL \url{https://doi.org/10.1145/2063384.2063405}.

\bibitem[Sander et~al.(2026)Sander, Chang, Sou{\v{c}}ek, Tran, Lacatusu, Rebuffi, Mourachko, Parimi, Ropers, Moritz, Stark, Elsahar, and Fernandez]{textseal}
Tom Sander, Hongyan Chang, Tom{\'a}{\v{s}} Sou{\v{c}}ek, Tuan Tran, Valeriu Lacatusu, Sylvestre-Alvise Rebuffi, Alexandre Mourachko, Surya Parimi, Christophe Ropers, Rashel Moritz, Vanessa Stark, Hady Elsahar, and Pierre Fernandez.
\newblock {TextSeal}: A localized {LLM} watermark for provenance \& distillation protection.
\newblock \emph{arXiv preprint arXiv:2605.12456}, May 2026.
\newblock \doi{10.48550/arXiv.2605.12456}.
\newblock URL \url{https://arxiv.org/abs/2605.12456}.

\bibitem[Tu et~al.(2024)Tu, Sun, Bai, Yu, Hou, and Li]{waterbench}
Shangqing Tu, Yuliang Sun, Yushi Bai, Jifan Yu, Lei Hou, and Juanzi Li.
\newblock {WaterBench}: Towards holistic evaluation of watermarks for large language models.
\newblock In \emph{Proceedings of the 62nd Annual Meeting of the Association for Computational Linguistics (Volume 1: Long Papers)}, pp.\  1517--1542, Bangkok, Thailand, August 2024. Association for Computational Linguistics.
\newblock \doi{10.18653/v1/2024.acl-long.83}.
\newblock URL \url{https://aclanthology.org/2024.acl-long.83/}.

\bibitem[Wijk et~al.(2025)Wijk, Lin, Becker, Jawhar, Parikh, Broadley, Chan, Chen, Clymer, Dhyani, Ericheva, Garcia, Goodrich, Jurkovic, Kinniment, Lajko, Nix, Sato, Saunders, Taran, West, and Barnes]{rebench}
Hjalmar Wijk, Tao~Roa Lin, Joel Becker, Sami Jawhar, Neev Parikh, Thomas Broadley, Lawrence Chan, Michael Chen, Joshua~M. Clymer, Jai Dhyani, Elena Ericheva, Katharyn Garcia, Brian Goodrich, Nikola Jurkovic, Megan Kinniment, Aron Lajko, Seraphina Nix, Lucas Jun~Koba Sato, William Saunders, Maksym Taran, Ben West, and Elizabeth Barnes.
\newblock {RE}-bench: Evaluating frontier {AI} {R\&D} capabilities of language model agents against human experts.
\newblock In \emph{Proceedings of the 42nd International Conference on Machine Learning}, volume 267 of \emph{Proceedings of Machine Learning Research}, pp.\  66772--66832. PMLR, 2025.
\newblock URL \url{https://proceedings.mlr.press/v267/wijk25a.html}.

\bibitem[Wu et~al.(2024)Wu, Hu, Guo, Zhang, and Huang]{dipmark}
Yihan Wu, Zhengmian Hu, Junfeng Guo, Hongyang Zhang, and Heng Huang.
\newblock A resilient and accessible distribution-preserving watermark for large language models.
\newblock In \emph{Proceedings of the 41st International Conference on Machine Learning}, volume 235 of \emph{Proceedings of Machine Learning Research}, pp.\  53443--53470. PMLR, 2024.
\newblock URL \url{https://proceedings.mlr.press/v235/wu24h.html}.

\bibitem[Zhao et~al.(2024{\natexlab{a}})Zhao, Ren, Hessel, Cardie, Choi, and Deng]{wildchat}
Wenting Zhao, Xiang Ren, Jack Hessel, Claire Cardie, Yejin Choi, and Yuntian Deng.
\newblock Wildchat: 1m chat{GPT} interaction logs in the wild.
\newblock In \emph{The Twelfth International Conference on Learning Representations}, 2024{\natexlab{a}}.
\newblock URL \url{https://openreview.net/forum?id=Bl8u7ZRlbM}.

\bibitem[Zhao et~al.(2024{\natexlab{b}})Zhao, Ananth, Li, and Wang]{unigram}
Xuandong Zhao, Prabhanjan Ananth, Lei Li, and Yu-Xiang Wang.
\newblock Provable robust watermarking for {AI}-generated text.
\newblock In \emph{International Conference on Learning Representations}, volume 2024, pp.\  43738--43772, 2024{\natexlab{b}}.
\newblock URL \url{https://proceedings.iclr.cc/paper_files/paper/2024/file/beae9ed5316bcc48e616754c06c11875-Paper-Conference.pdf}.

\bibitem[Zhong et~al.(2026)Zhong, Raghunathan, and Carlini]{impossiblebench}
Ziqian Zhong, Aditi Raghunathan, and Nicholas Carlini.
\newblock {ImpossibleBench}: Measuring {LLMs}' propensity of exploiting test cases.
\newblock In \emph{The Fourteenth International Conference on Learning Representations}, 2026.
\newblock URL \url{https://arxiv.org/abs/2510.20270}.

\end{thebibliography}

\message{^^JLASTREFERENCESPAGE \thepage^^J}

\ifincludeappendixx
  \clearpage
  \appendix
  \crefalias{section}{appendix}
  \crefalias{subsection}{appendix}
  \section{Breakdown of All Generated Components}
\label{app:components}

In this section, we explain each of the watermarking components proposed by the research agents and visualized in~\cref{fig:component-pipeline}.
In \cref{app:components:prelim}, we explain in depth the decomposition of watermarks in 5 components (watermark unit, context seeding, scores, logits transformation, and detectors), and detail the notation used throughout.
In the remaining subsections, we follow the architecture of~\cref{fig:component-pipeline} and go through each component one by one: the watermark unit~(\cref{app:components:unit}), context seeding~(\cref{app:components:context}), the watermark score~(\cref{app:components:score}), the logits transformation~(\cref{app:components:transform}), and detection~(\cref{app:components:detection}).

\subsection{Preliminaries and Notation}
\label{app:components:prelim}

\paragraph{Four components}
We decompose a watermarking scheme into the four components introduced in~\cref{sec:method:preliminaries}, with the additional concept of watermark unit.
\begin{enumerate}[leftmargin=*]
	\item \emph{Watermark unit} is the component on which the watermark is applied: in prior works, it is usually tokens, but agents also proposed to apply watermarks on equivalence classes of tokens.
	\item \emph{Context seeding} hashes the previously generated tokens into a seed that is used, along with the watermark private key, to pseudo-randomly sample the watermarking scores.
	\item The \emph{watermark score} is the pseudo-random score that can be combined with true randomness before being used by the logits transformation. At detection, only the pseudo-random score can be recovered and not the true randomness.
	\item The \emph{logits transformation} combines the score with the model's next-token distribution to pick the next token, in a way that is distortion-free~(\cref{eq:distortion_free}).
	\item The \emph{detector} recomputes the pseudo-random scores of a candidate text, aggregates them into a statistic, and returns a p-value.
\end{enumerate}

\paragraph{Notation}
Let $\Sigma$ be the vocabulary and $\omega \in \Sigma^*$ a text, with $\omega_t$ its \nth{t} token and $p_t \in \Delta(\Sigma)$ the model's next-token distribution given $\omega_{<t}$.
A grouping $\mathcal{G}: \Sigma \rightarrow \mathbb{N}$ maps tokens to units. We write $u = \mathcal{G}(v)$ for the group of a candidate token $v$ and
\begin{equation}
	\label{eq:group_mass}
	M_t(u) \mathrel{:=} \sum_{v \in \Sigma} p_t(v) \indicator\{\mathcal{G}(v) = u\}
\end{equation}
for its mass under $p_t$.
Context seeding maps the context $c_t$ (a function of $\omega_{<t}$) and the private key $\xi$ to a seed $a_t$.
The seed is then used to pseudo-randomly sample a multivariate uniform random variable $U_{a_t,u} \sim \mathcal{U}([0,1])$ for every token unit $u$.
Because the watermark score may add true randomness that is not computable at detection, we use the notation $\tilde{U}_{a_t,u}$ for the watermark score mixed with true randomness (and thus used by the logits transformation).
The logits transformation selects a group $u_*$ from $(\tilde{U}_{a_t}, p_t)$ and then samples the emitted token inside that group using the residual probability $p_t(v)/M_t(u_*)$.
At detection, the scheme reconstructs the set $D$ of retained (seed, unit) pairs of $\omega$ (\eg after de-duplication of repeated $(a_t,u)$ tuples as in~\citet{threebricks}), with $N \mathrel{:=} \lvert D \rvert$, aggregates them into a statistic $S(\omega)$, and reports a p-value $P(\omega)$.
We write $Q(N,s) \mathrel{:=} \Pr[\operatorname{Gamma}(N,1) \ge s]$ for the Gamma upper tail.

\paragraph{Experimental setup}
We ablate every component against a simple baseline AAR scheme: identity grouping, a two-token context, no private randomness in the score ($\tilde{U} = U$), the \texttt{Gumbel race}, and the scalar \texttt{Gamma} detector of~\cref{sec:eval:directions}.
For each ablation, we swap one component (\eg the \texttt{Gumbel race} with the \texttt{ResidualRace}) and report the difference in percentage points for our evaluation metrics~(\cref{sec:eval}).
Quality is one minus the average relative distance to the unwatermarked model in perplexity and both Self-BLEU scores.
In the bar plots, positive values indicate an improvement, and values in parentheses give the reference quality or TPR@1\%FPR.
For ablating some component hyperparameters, we may sometime use a different reference (instead of the simple AAR baseline) and specify in the caption when the reference differs from the baseline.

\subsection{Watermark Unit}
\label{app:components:unit}

Prior works~\citep{aar,kgw,synthid,textseal} used a single token as the watermark: $\mathcal{G}$ is the identity and $M_t = p_t$. 

\begin{figure}[t]
    \centering
    \begin{minipage}[t]{0.47\textwidth}
        \vspace{0pt}\centering
        \input{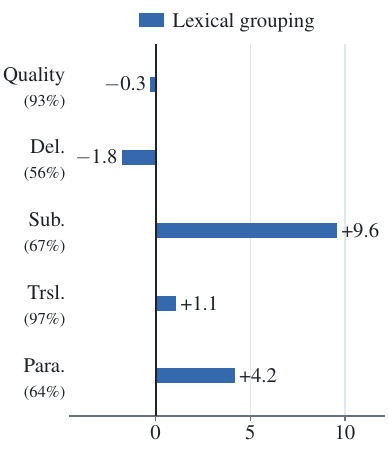}
    \end{minipage}\hfill
    \begin{minipage}[t]{0.47\textwidth}
        \vspace{0pt}\centering
        \input{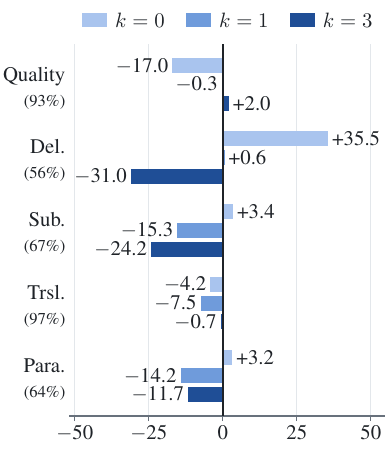}
    \end{minipage}
\end{figure}

\paragraph{\texttt{Lexical group}}
Agents replace tokens by classes of a public map $\mathcal{G}$ built from the tokenizer and WordNet~3.0: a token is decoded, stripped and lower-cased, and if what remains is an ASCII word of at least three characters it is mapped to the representative of its first WordNet synset.
Hence synonyms share a class, which is what makes it robust to substitution: a replacement inside a class leaves both the seed and the observed "token" unchanged.
\cref{fig:comp:grouping} confirms this intuition: \texttt{Lexical group} indeed significantly improves robustness to substitution, and slightly improves backtranslation and paraphrasing robustness as well. However, it significantly lowers deletion robustness, and slightly lowers quality.

\subsection{Context Seeding}
\label{app:components:context}

For context seeding, unlike some prior works~\citep{kgw}, we assume the hashing function is without collision: the hash of a context is unique.
Because the set of possible contexts is finite, this assumption is achievable in practice.
We defer the details of the hashing functions, as well as an analysis of the different hashing functions, to \cref{app:hash_function}.
We also use a per-request watermark cache: if a seed $a_t$ hashed from the context is repeated, the next token is sampled without watermarking.
This is required to guarantee that each individual request is distortion-free.

\paragraph{\texttt{Fixed-length}}
Prior work hashes the tuple of the $k$ preceding tokens~\citep{kgw,textseal}. In particular, $k=0$ is the unigram scheme of~\citet{unigram}.
As we show in \cref{fig:comp:context-fixed}, the context size $k$ controls a robustness-quality trade-off: shorter contexts survive more edits but repeat more often which leads to correlation across a corpus of generated texts.
Furthermore, at low context, \ie $k=1$ or $k=0$, because we need to de-duplicate (seed, token) pairs at detection, the lack of context diversity lowers the detectability and therefore can counteract the gain in robustness.

\begin{figure}[t]
    \centering
    \begin{minipage}[t]{0.47\textwidth}
        \vspace{0pt}\centering
        \input{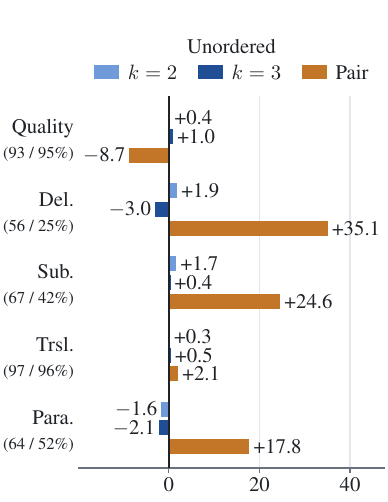}
    \end{minipage}\hfill
    \begin{minipage}[t]{0.47\textwidth}
        \vspace{0pt}\centering
        \input{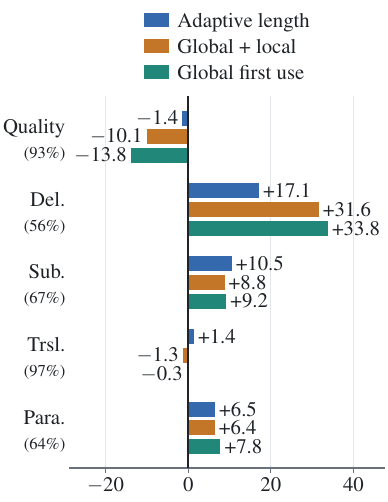}
    \end{minipage}
\end{figure}

\paragraph{\texttt{Unordered}}
It is the same as \texttt{Fixed-length} context seeding, except that it hashes the {set} of the $k$ preceding tokens rather than their tuple (the candidate $v$ is kept separate), so that a transposition of preceding tokens does not change the context.
For $k \le 1$, it is identical to \texttt{Fixed-length}.
As we show in~\cref{fig:comp:context-undirected}, for $k=2$ and $k=3$ this change barely affects robustness and slightly improves quality.
The agents also proposed an \emph{unordered pair} variant (\emph{Pair} in~\cref{tab:full-components-context}) that hashes the set $\{\omega_{t-1}, v\}$, which includes the candidate token itself, rather than the tuple $(\omega_{t-1}, v)$.
Compared to the core scheme, it significantly increases robustness but significantly lowers quality~(\cref{fig:comp:context-undirected}).

\paragraph{Adaptive and global length}
Rather than fixing $k$, these variants shorten the context for the tokens that are unlikely to repeat anyway, so that only the positions that need a long context pay for one.
\begin{enumerate}[leftmargin=*]
	\item \texttt{Adaptive length}: given a classifier $\mathrm{rare}: \Sigma \rightarrow \{0,1\}$ (here that a token frequency in a fixed corpus is below $0.1\%$), if the previous token is rare, hash only that token ($k=1$), otherwise hash the full two-token window.
	In~\cref{tab:full-components-context}, we also evaluate it with the \texttt{Excess} detector of~\cref{app:components:detection}.
	\item \texttt{Global + local}: fix a public set $E \subset \Sigma$ of eligible tokens (the agents used as $E$ the tokens that have a WordNet entry), use unigram (\ie $k=0$) if the candidate is in $E$, and a $k$-token context otherwise, with $k=2$ unless stated otherwise.
	\item \texttt{Global first use} refines it by applying unigram to a candidate in $E$ only the first time it appears in the completion, and then uses a two-token context if it appears again.
\end{enumerate}
We show in \cref{fig:comp:context-adaptive} that all those approaches significantly increase the robustness compared to the baseline, but both global variants also significantly degrade quality.

\begin{figure}[t]
    \centering
    \begin{minipage}[t]{0.47\textwidth}
        \vspace{0pt}\centering
        \input{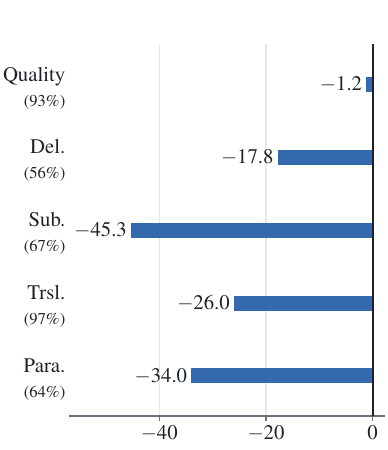}
    \end{minipage}\hfill
    \begin{minipage}[t]{0.47\textwidth}
        \vspace{0pt}\centering
        \input{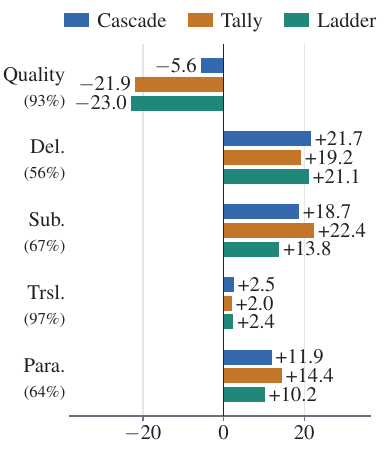}
    \end{minipage}
\end{figure}

\paragraph{\texttt{Anchored}}
Let $E$ again be a public token set and let $s$ be the distance from the current token back to the most recent token of $E$.
If $s \le S_{\max} = 4$, use as context the tuple of that anchor token and the offset $s$. Otherwise, if $s > S_{\max}$ skip the watermark.
For the detector, at every token $\omega_t$ it finds the closest previous anchor and sweeps all permitted distances $s\in \{0,\dots,S_{\max}\}$ adjusting the p-values with Bonferroni correction.
Yet, as we show in \cref{fig:comp:context-anchored}, while intuitively the slot search should buy robustness (editing tokens in between anchor tokens does not change the context), the search is too expensive and this scheme is worse than the baseline on every direction.

\paragraph{\texttt{Suffix cascade} and \texttt{Tally \& ladder}}
These approaches share the same goal: when a context has already been used in the completion, derive a new seed for it instead of skipping the position, so that a short context keeps its robustness without losing detectability.
They only differ in how that new seed is built.
\begin{enumerate}[leftmargin=*]
	\item \texttt{Suffix cascade}: try suffix lengths in increasing order and use the shortest one that has not already been used in this completion. If every allowed length has been used, skip the watermark.
	\item \texttt{Tally}: keep a one-token suffix and hash the previous token, and the number of its earlier occurrences, so that a repeated context yields a fresh seed.
	\item \texttt{Ladder} generalizes both: count every suffix width $1 \le w \le K$, select the shortest width whose count is below a threshold, and hash the width, the suffix, and the count. 
	Counts are capped at $8192$, beyond which the position falls back to ordinary sampling due to the specifics of the hash function~(\cref{app:hash_function}).
	Optionally, a \emph{floor} forces a width of at least $2$ when the previous token ID is below $200$ (as in ConcordMark,~\cref{app:concordmark}), as for most tokenizers low ID tokens tend to be very common (\eg letters, numers, or punctuations).
	Unless stated otherwise, we use $K=4$, a threshold of $3$, and the floor, and we report all four combinations of $K \in \{3,4\}$ with and without the floor in~\cref{tab:full-components-context}.
	Ladder makes almost every position usable at a one- or two-token context, which is why the \texttt{Ladder} is used in  ConcordMark~(\cref{sec:eval:best}).
\end{enumerate}
As we find in \cref{fig:comp:context-cascade}, the \texttt{Suffix cascade} indeed keeps most of the robustness of a short context, with a small quality cost, and unlike $k=1$ it improves on all robustness metrics.

\begin{wrapfigure}[14]{R}{0.45\textwidth}
    \centering
    \vspace{-0.1in}
    \includegraphics[width=\linewidth]{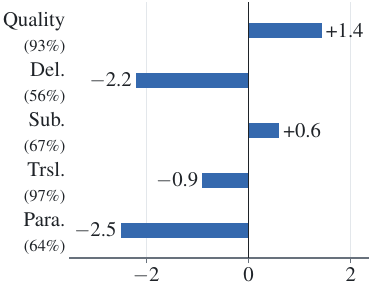}
    \vspace{-0.2in}
    \caption{Impact of the entropy \texttt{Prelude}.}
    \vspace{-0.1in}
    \label{fig:comp:prelude}
\end{wrapfigure}

\paragraph{\texttt{Last important} and \texttt{Interleaved}}
Those two schemes belong to the \texttt{skeleton family}~(\cref{fig:component-pipeline}).
The \texttt{skeleton family} splits the vocabulary using a fixed important/filler tokens classifier and gives the two streams separate histories.
\begin{enumerate}[leftmargin=*]
	\item \texttt{Last important} ignores filler tokens when computing the context and uses the last important token as context, so that edits to filler tokens do not move the context of the next content word.
	\item \texttt{Interleaved} uses when sampling an important token the last important token as context, and when sampling a filler token uses the last two previous filler tokens as context.
\end{enumerate}
Both context seedings must be changed together with their logits transformation which is why we defer the evaluation to~\cref{app:components:transform}.

\begin{figure}[!tb]
    \centering
    \includegraphics[width=\linewidth]{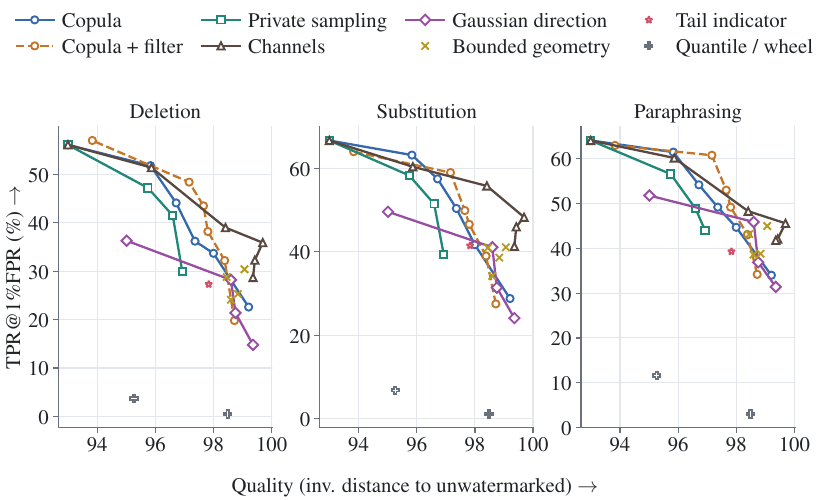}
    \caption{Quality-robustness trade-offs for the different watermark scores.}
    \label{fig:comp:score-global}
\end{figure}

\subsection{Watermark Score}
\label{app:components:score}

\paragraph{\texttt{Prelude}}
\texttt{Prelude} is the only component that can be composed with any other.
It consists of sampling ordinarily at the start of the completion until the accumulated entropy $H_2(p_t) = -\log\sum_v p_t(v)^2$ of the first generated tokens reaches $h$ nats, or after $T$ tokens, whichever comes first.
In practice, agents use $h=14$ and $T=24$.
For detection, we keep all tokens, including those of the prelude, as the detector cannot recompute the entropy without the model.
We compare \texttt{Prelude} to the no-Prelude baseline and find in~\cref{fig:comp:prelude} that it indeed improves quality with a mild robustness cost.

\paragraph{\texttt{Gaussian copula}}
\texttt{Gaussian copula} draws a fresh private normal $\eta_{t,u} \sim \mathcal{N}(0,1)$ at every position and sets $\tilde{U}_{t,u} = \Phi(\rho \Phi^{-1}(U_{a_t,u}) + \sqrt{1-\rho^2}\,\eta_{t,u})$.
The goal is to inject true randomness into the watermark score to increase quality at the cost of detectability.
At $\rho = 0$ the scheme is unwatermarked, at $\rho = 1$ it uses only the watermarked pseudo-random scores.
As we show in~\cref{sec:eval:directions}, it is a more effective way of improving quality via true randomness than prior works.

\begin{wrapfigure}{R}{0.45\textwidth}
    \centering
    \vspace{-0.1in}
    \includegraphics[width=\linewidth]{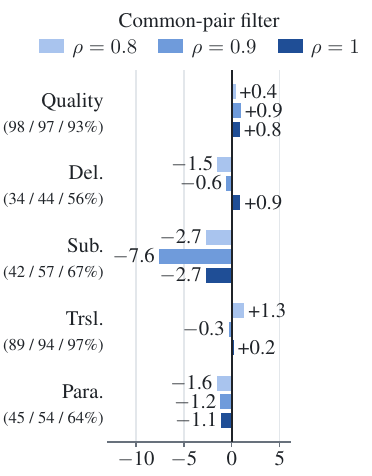}
    \vspace{-0.2in}
    \caption{Impact of common-pair filtering compared to the unfiltered \texttt{Gaussian copula} at the same correlation $\rho$. Parentheses give the unfiltered reference percentages for $\rho=0.8$ / $\rho=0.9$ / $\rho=1$, respectively.}
    \vspace{-0.6in}
    \label{fig:comp:score-filter}
\end{wrapfigure}

The \texttt{Gaussian copula} can be further refined with a common pair filter: we replace $\rho$ with a function of the current context, $\rho_{t,u} = \rho\,\indicator\{(a_t,u) \notin \mathcal{F}\}$, where $\mathcal{F}$ is the set of the $10{,}000$ most frequent (seed, token) pairs of a public corpus (in this case, the first $20{,}000$ texts of the public \textsc{ELI5} corpus, counting each pair at most once per text).
The goal is to further increase the quality with a better detectability trade-off.
Yet, as we show in \cref{fig:comp:score-filter}, using the simpler $\rho \in [0,1]$ is better in practice.

\paragraph{\texttt{Tail indicator}}
For each unit $u$, set $B_u = \max(1, \lfloor 1/\max(M_t(u), 2^{-40})\rfloor)$ and $b_u = 1 - 1/B_u$, draw a private $V_u \sim \mathcal{U}([0,1])$, and let $\tilde{U}_{t,u} = b_u + V_u/B_u$ if $U_{a_t,u} \ge b_u$ and $b_uV_u$ otherwise.
The construction keeps $\tilde{U}$ uniform, so the scheme stays distortion-free, and the detector still reads the full $U$.
However, as we show in \cref{fig:comp:score-global}, it has a significantly worse quality-robustness trade-off than \eg the \texttt{Gaussian copula}.

\paragraph{\texttt{Private sampling}}
This is the naive idea: with probability $r$, ignore the watermark at this position and sample from $p_t$ and otherwise use the keyed sampler.
However, as we show in \cref{fig:comp:score-global}, it has a significantly worse quality-robustness trade-off than \eg the \texttt{Gaussian copula}.

\paragraph{\texttt{Channels}}
\texttt{Channels} maintain $L$ independent keyed fields $U_{j,a,u}$ and draw a single channel $J \sim \mathcal{U}\{0,\ldots,L-1\}$ per request, using $\tilde{U}_{t,u} = U_{J,a_t,u}$ throughout the completion.
The goal is the same as the \texttt{Gaussian copula}, spending private randomness on quality, except that here the randomness is discrete: two requests that share a prompt and a context take different fields.
The cost is paid at detection, where the detector runs all $L$ fields and corrects its p-value by a Bonferroni factor $L$.
However, as we show in \cref{fig:comp:score-global}, that factor dominates quickly: quality saturates around $L = 32$ while the robustness to every attack keeps falling, so only small $L$ are useful.

For adjusting the p-values to account for the $L$ tests, we can either use Bonferroni or \v{S}id\'ak 
Indeed, the $L$ per-channel p-values are independent under the null, so the \v{S}id\'ak correction $1 - (1 - \min_j P_j)^L$ is valid and strictly tighter than the Bonferroni factor $L$, and the submissions use both.
Yet, we observe no difference in practice~(\cref{tab:full-components-detection}).

\begin{figure}[t]
    \centering
    \includegraphics[width=\linewidth]{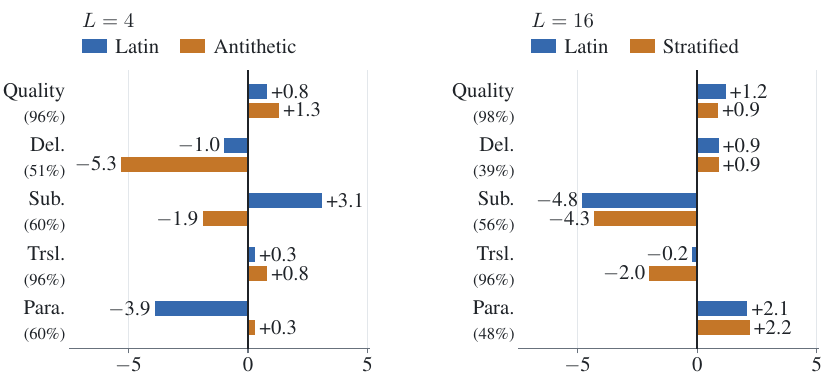}
    \caption{Impact of correlated \texttt{Channel variants} compared to $L$ independent \texttt{Channels}, keeping the channel count and Bonferroni correction fixed. Parentheses give the independent-channel reference percentages in each panel; bars show differences in percentage points.}
    \label{fig:comp:score-channel-variants}
\end{figure}

\texttt{Channels} can be further refined by correlating the $L$ fields, so that the request-level choice costs less detectability.
\begin{enumerate}[leftmargin=*]
	\item \texttt{Latin} channels derive them from a single base uniform by rotation, $U_{j,a,u} = (U_{a,u} + j/L) \bmod 1$.
	\item \texttt{Antithetic} channels work only for schemes that use normal scores $X=\Phi^{-1}(U_{j,a,u})$, and pair each keyed normal $X$ with $-X$.
	\item \texttt{Stratified} channels use a keyed permutation $\sigma_{a,u}$ and keyed jitters, $U_{j,a,u} = (\sigma_{a,u}(j) + \upsilon_{j,a,u})/L$.
\end{enumerate}
All three keep the same detector and use only Bonferroni correction (\v{S}id\'ak is no longer valid due to the dependence between tests), so~\cref{fig:comp:score-channel-variants} compares each against independent channels at the same $L$.
Yet, none of the three is a reliable improvement: the differences are within a few points everywhere and change sign across attacks.
We therefore suggest using the simpler version with independent keys, which also benefits from the more powerful \v{S}id\'ak correction.

\begin{figure}[t]
    \centering
    \begin{minipage}[t]{0.47\textwidth}
        \vspace{0pt}\centering
        \input{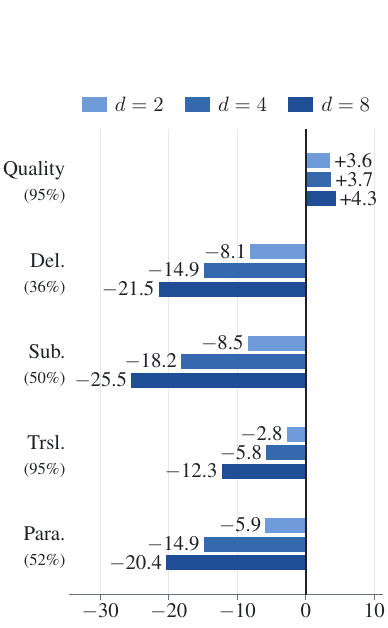}
    \end{minipage}\hfill
    \begin{minipage}[t]{0.47\textwidth}
        \vspace{0pt}\centering
        \input{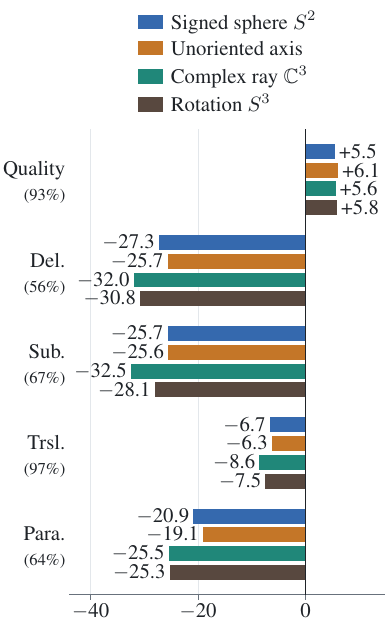}
    \end{minipage}
\end{figure}

\paragraph{\texttt{Gaussian direction}}
This is the first of the private-geometry scores~(\cref{fig:component-pipeline}), and it generalizes \texttt{Channels} from a discrete choice to a continuous one.
We use the context seed and watermark private key to sample a multivariate Gaussian $Z_{a_t,u} \sim \mathcal{N}(0, I_d)$, and \emph{per request} we draw a private unit direction $A \in S^{d-1}$ once, used for the whole completion:
\begin{equation}
	\label{eq:gaussian_direction}
	\tilde{U}_{a_t,u} = \Phi(A \cdot Z_{a_t,u}).
\end{equation}
Hence, two different requests (even with the same prompt) will have different directions, hence different watermark scores and a higher diversity.
As we show in \cref{fig:comp:score-gaussian-direction}, increasing the dimension $d$ significantly improves the quality at the cost of detectability, and thus robustness.

\paragraph{\texttt{Spherical}}
This is a similar idea to the \texttt{Gaussian direction}, but on compact spaces.
\begin{enumerate}[leftmargin=*]
	\item \texttt{Signed sphere} draws $Z_{a_t,u} \sim \mathcal{U}(S^2)$ and $A \in S^2$, with 
	$
		\tilde{U}_{a_t,u} = (1 + A\cdot Z_{a_t,u})/2.
	$
	\item \texttt{Unoriented axis} draws $Z_{a_t,u} \sim \mathcal{U}(S^2)$ and $A \in S^2$, with
	$
		\tilde{U}_{a_t,u} = \lvert A \cdot Z_{a_t,u} \rvert.
	$
	\item \texttt{Complex ray} draws $Z_{a_t,u} \sim \mathcal{U}(S^3)$ and $A$ unit vectors of $\mathbb{C}^3$ modulo phase with 
	$
		\tilde{U}_{a_t,u} = \lvert A^*Z_{a_t,u}\rvert^2(2-\lvert A^*Z_{a_t,u}\rvert^2).
	$
	\item \texttt{Rotation} draws $Z_{a_t,u} \sim \mathcal{U}(S^3)$ and $A$ is a unit quaternion with 
	$
		\tilde{U}_{a_t,u} = F(\lvert A\cdot Z_{a_t,u}\rvert),
	$
	where $F(r) = \frac{2}{\pi}[\arcsin r + r\sqrt{1-r^2}]$.
\end{enumerate}
\Cref{fig:comp:score-geometries} evaluates all four (using their corresponding detectors) compared to the simple baseline.
We find that they all trade off quality for detectability and that \texttt{Rotation} has the most favorable trade-offs.
Because the \texttt{Gaussian direction} and the spherical geometries only change the score $\tilde{U}$, they can be combined with any logits transformation: in~\cref{tab:full-components-score}, we evaluate each of them both with the \texttt{Gumbel race} and with the \texttt{ResidualRace}~(\cref{app:components:transform}).

\paragraph{\texttt{Wheel}}
Draw a keyed offset $U_{a_t}$ and a per-request random rotation $\varphi \sim \mathcal{U}([0,1])$, and feed $\tilde{U}_{a_t,u} = (U_{a_t,u} + \varphi) \bmod 1$ to the \texttt{QuantileRace} transformation of~\cref{app:components:transform}.
It is the continuous limit of the \texttt{Latin} channels of~\cref{fig:comp:score-channel-variants}: the rotation $\varphi$ is drawn from the continuous circle rather than from $L$ evenly spaced angles.
It requires the rotation-invariant detector of~\cref{app:components:detection}.

\subsection{Logits Transformation}
\label{app:components:transform}

\begin{figure}[t]
    \centering
    \includegraphics[width=\linewidth]{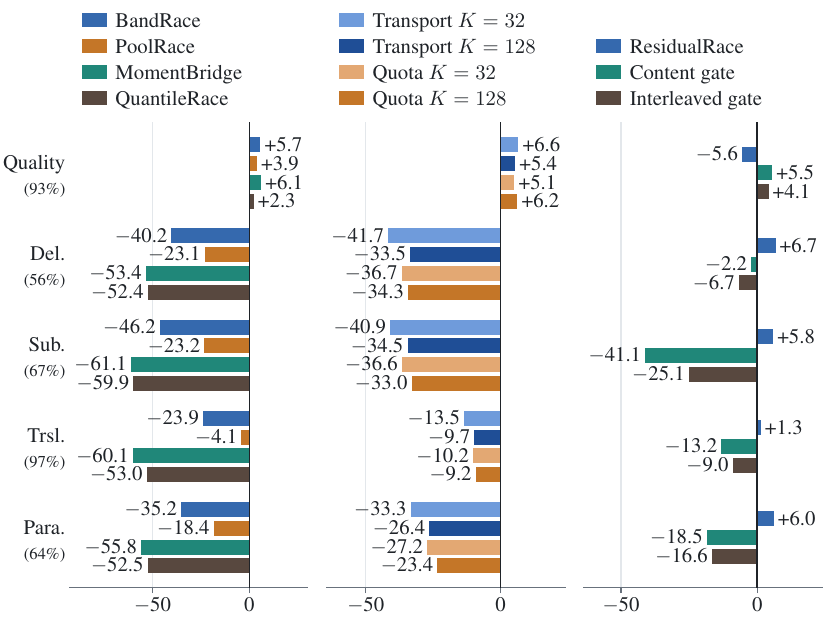}
    \caption{Impact of logits transformations compared to the \texttt{Gumbel race} baseline.}
    \label{fig:comp:transform-global}
\end{figure}

\paragraph{\texttt{Gumbel race}}
This is the transformation of~\citet{aar}, which we rewrite here with the group notation for clarity:
\begin{equation}
	\label{eq:gumbel_race}
	G_{t,u} = -\log(-\log \tilde{U}_{t,u}), \qquad u_* = \arg\max_u\left[\log M_t(u) + G_{t,u}\right],
\end{equation}
after which the token is drawn inside $u_*$ with the residual probability $p_t(v)/M_t(u_*)$.
By the Gumbel-max trick, $u_*$ has distribution $M_t$ in expectation over the key, so the scheme is distortion-free.
Equivalently, it is an exponential race: $u_*$ minimizes $-\log \tilde{U}_{t,u} / M_t(u)$.

\paragraph{\texttt{BandRace}}
Partition the groups into bands whose masses differ by less than a factor of four, 
$$
	\mathcal{B}_{t,b} = \{u : 4^b \le M_{t,\max}/M_t(u) < 4^{b+1}\},
$$
pick a band randomly with probability equal to its total mass, and run the \texttt{Gumbel race} inside it.
This increases quality as the choice of the band is truly random, but lowers the detectability~(\cref{fig:comp:transform-global}).

\paragraph{\texttt{PoolRace}}
It is the black-box equivalent of the \texttt{Gumbel race} (\ie it is an approximation of the \texttt{Gumbel race} that does not require access to the logits), similar to what has been proposed in~\citet{blackbox_wm}: draw $K = 16$ tokens $V_i$ independently from $p_t$, let $N_u$ count the proposals in group $u$, and pick
\begin{equation}
	u_* = \arg\min_{u: N_u > 0} \frac{-\log \tilde{U}_{t,u}}{N_u}.
\end{equation}
As we see in \cref{fig:comp:transform-global}, this variant slightly improves quality because the candidates are drawn with true randomness, but hurts detectability.

\paragraph{\texttt{MomentBridge}}
The closest equivalent to this scheme in the literature is SynthID~\citep{synthid}, where the watermark distribution is refined iteratively using the watermark keyed scores. 
Let $K \in \mathbb{N}$ be a scheme parameter.
Start with the weights $w_u^{(0)}=M_t(u)$ and the fixed costs $c_u = -\log M_t(u)$.
Let $X_{t,u} = \Phi^{-1}(\tilde{U}_{t,u})$ be the normal watermark score.
For $i=1,\dots,K$, sample random Gaussians $\eta_{i,u}$ and compute:
\begin{equation*}
	z_{i,u}=\frac{X_{t,u}}{\sqrt{K}}+\eta_{i,u}-\frac1{K}\sum_{\ell=1}^{K}\eta_{\ell,u}.
\end{equation*}
Let
\begin{equation*}
	  \bar c_i=\sum_uw_u^{(i-1)}c_u,\qquad
		\bar z_i=\sum_uw_u^{(i-1)}z_{i,u},
\end{equation*}
we update the weights iteratively with
\begin{equation*}
	    w_u^{(i)}=w_u^{(i-1)}\left(1+\frac{r_{i,u}}{A_i}\right),
\end{equation*}
where $A_i=\max_{u:w_u^{(i-1)}>0}|r_{i,u}|$,
\begin{equation*}
	    b_i=
	\frac{\sum_uw_u^{(i-1)}(c_u-\bar c_i)(z_{i,u}-\bar z_i)}
		{\sum_uw_u^{(i-1)}(c_u-\bar c_i)^2},
	\text{ and }
	r_{i,u}=z_{i,u}-\bar z_i-b_i(c_u-\bar c_i).
\end{equation*}
The next group is sampled randomly according to the $w^{(K)}$ distribution.
Because of the true randomness in $\eta_{i,u}$, this logits transformation trades off quality for detectability.
However, as we show in \cref{fig:comp:transform-global}, it achieves a relatively poor trade-off.

\paragraph{\texttt{QuantileRace}}
It replaces the \texttt{Gumbel race} with inverse-CDF selection: order each group (or token) in the vocabulary according to a keyed pseudo-random permutation $R_{a_t,u}$ with an interval of length $M_t(u)$, and emit the group whose interval contains the watermark score $\tilde{U}_{a_t}$.
See \cref{fig:comp:transform-global}.

\paragraph{\texttt{TransportRace} and \texttt{QuotaTransport}}
This is the second black-box family, and the only one that couples the transformation to the \texttt{Channels}.
Draw $L$ tokens from the model, compute the keyed score $E_{j,a_t,\mathcal{G}(V_i)}$ of every (channel, proposal) pair, and choose the permutation $\sigma_*$ maximizing $\sum_j E_{j,a_t,\mathcal{G}(V_{\sigma(j)})}$.
\texttt{QuotaTransport} stratifies the proposal pool before the same assignment, drawing one proposal from each of the $L$ mass intervals $(i+W)/L$ for a fresh $W \sim \mathcal{U}([0,1])$, so that the pool covers the distribution rather than concentrating on its mode.
Both were evaluated at $L \in \{32, 128\}$ against the \texttt{Gamma} detector with Bonferroni correction.
As we see in~\cref{fig:comp:transform-global}, all four configurations substantially improve quality, but they are less robust than the \texttt{Gumbel race}.

\paragraph{\texttt{ResidualRace}}
This is the extension of the \texttt{Gumbel race} described in~\cref{sec:eval:directions}.
We keep one residual clock $R_{a_t,u}$ per (context, group) pair, initialized at $-\log \tilde{U}_{a_t,u}$ on first use. 
At each step, we sample the next token according to the exponential race:
\begin{equation*}
	\tau_t=\min_u\frac{R_{a_t,u}}{M_t(u)},\qquad
	u_* = \arg\min_u\frac{R_{a_t,u}}{M_t(u)},\qquad
	R_{a_t,u}\leftarrow R_{a_t,u}-\tau_t M_t(u),
\end{equation*}
and replace the winning token's score with a fresh exponential random variable: $R_{a_t,u_*} \sim \operatorname{Exp}(1)$.
Because of the memorylessness property of the exponential, all losing residual exponentials are independent given the past, so the next winning probability is still $M_t(u)$ even when a context recurs and its masses have changed.
This is what makes it sound to re-use previous clocks in the exponential race when a context re-occurs~(\cref{prop:residual_race}).
As we show in \cref{fig:comp:transform-global}, it increases detectability at the cost of quality.
Yet, as we have shown in \cref{fig:residual-clock}, adding \texttt{Gaussian copula} noise to the initial clocks recovers the quality cost while keeping the detectability gains.

\paragraph{\texttt{Content gate} and \texttt{Interleaved gate}}
This is the skeleton family logits transformation, which relies on either the \texttt{Last important} or \texttt{Interleaved} context seeding.
Let $\mathcal{C} $ be the important groups, and $r_t = \sum_{u \in \mathcal{C}} M_t(u)$.
At each step, draw a fresh $B_t \sim \operatorname{Bernoulli}(r_t)$.
\begin{enumerate}[leftmargin=*]
	\item The \texttt{Content gate} uses the \texttt{Gumbel race} within the important groups when $B_t = 1$ and samples conditionally from the fillers otherwise.
	\item The \texttt{Interleaved gate} uses a \texttt{Gumbel race} inside whichever group the gate selects (\eg either the important or filler group), seeding the \texttt{Gumbel race} with the corresponding group context.
\end{enumerate}
As we see in \cref{fig:comp:transform-global}, the \texttt{Interleaved gate} strictly outperforms the \texttt{Content gate}, though both show limited robustness.

\subsection{Detection}
\label{app:components:detection}

\begin{figure}[t]
    \centering
    \includegraphics[width=\linewidth]{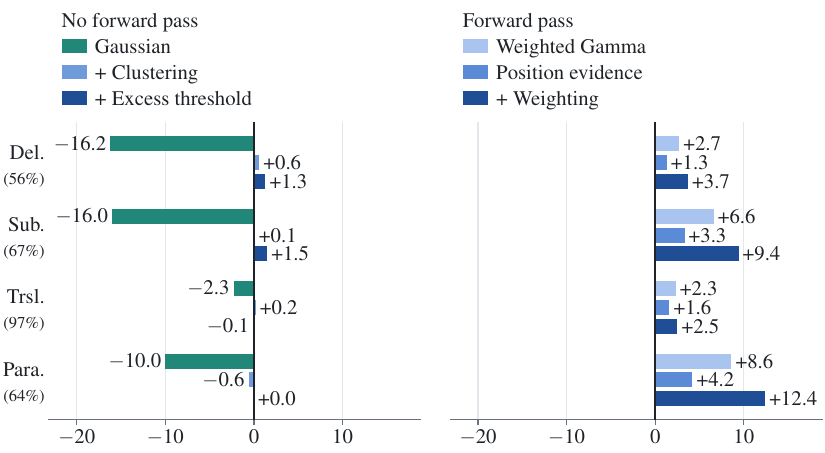}
    \caption{Comparison of scalar detectors without a model forward pass (left) and with a forward pass (right). Quality is omitted because changing the detector does not change generation.}
    \label{fig:comp:detector-global}
\end{figure}

\paragraph{Scalar \texttt{Gamma} and \texttt{Gaussian}}
The \texttt{Gamma} detector sums exponential scores, 
$$
	S(\omega) = \sum_{t \in D} -\log(1 - U_{a_t,\omega_t}) \sim \operatorname{Gamma}(N,1).
$$ 
The \texttt{Gaussian} detector sums the normal scores instead, 
$$
	S(\omega) = N^{-1/2}\sum_{t \in D} \Phi^{-1}(U_{a_t,\omega_t}) \sim \mathcal{N}(0,1).
$$
Both detectors weight the evidence differently: unlike the \texttt{Gaussian} detector, the \texttt{Gamma} detector gives more weight to the tail of the distribution, \ie to the tokens with very high watermarking scores.
As we see in \cref{fig:comp:detector-global}, the global \texttt{Gamma} detector significantly outperforms the \texttt{Gaussian} detector without changing the generation algorithm.

\paragraph{Gamma \texttt{Clustering} and \texttt{Excess} threshold}
Both detectors are refinements over the \texttt{Gamma} detector that group the evidence by context before aggregating.
Let $\mathcal{H}(\omega)$ be the set of unique contexts in $\omega$ and $D_a(\omega)$ the distinct successors of context $a$, with $m_a = \lvert D_a(\omega) \rvert$.
\begin{enumerate}[leftmargin=*]
	\item \texttt{Clustering} calibrates each context separately, and then aggregates the calibrated values,
	\begin{align*}
		S_a(\omega) &= \sum_{v \in D_a(\omega)} -\log(1-U_{a,v}), \\
		S(\omega) &= \sum_{a \in \mathcal{H}(\omega)} -\log Q(m_a, S_a(\omega)) \sim \operatorname{Gamma}(\lvert\mathcal{H}(\omega)\rvert, 1),
	\end{align*}
	so that a context with many successors cannot dominate the statistic on its own.
	Indeed, during generation, all successors of a given context take part in the same race, which means that two observations of the same race (\ie context) carry less information than two observations of distinct races.
	Hence, even though the detection is sound as long as no (context, candidate) pair is repeated, it should account for this loss of information whenever a context is repeated.
		\item \texttt{Excess} detection keeps only the part of each $-\log Q(m_a, S_a(\omega))$ above the $\operatorname{Exp}(1)$ median, 
	$$
		S(\omega) = \sum_a \max(-\log Q(m_a, S_a(\omega)) - \log 2, 0),
	$$ 
	which is zero with probability one half under the null and $\operatorname{Exp}(1)$ otherwise, giving the exact tail $R_H(s) = \sum_{b=1}^{H}\binom{H}{b}2^{-H}Q(b,s)$ for $H = \lvert\mathcal{H}\rvert$.
	This pushes the idea of the \texttt{Gamma} detector even further: only the tail of the distribution carries reliable watermarking signal.
\end{enumerate}	
As we see in \cref{fig:comp:detector-global}, both approaches indeed slightly increase detection power without changing the generation algorithm.

\paragraph{\texttt{Weighted Gamma} and \texttt{Position evidence}}
Those two detectors are additions on top of the \texttt{Gamma} detector that require a model forward pass at detection, and hence are considerably more costly.
Let $\hat{p}_t$ be the estimated next-token distribution given $\omega_{<t}$.
\begin{enumerate}[leftmargin=*]
	\item \texttt{Weighted Gamma} weights each exponential score by the observed token's surprisal quantized to quarter-unit, 
	$$
		\eta_t=1+\frac14\operatorname{round}\left(4\min(\max(-\log \hat{p}_t(\omega_t),0),2.5)\right)
	$$ 
	and uses as a detection statistic
	$$
		S(\omega) = \sum_{t \in D} -\eta_t\log(1 - U_{a_t,\omega_t}).
	$$
	The intuition is that a high-entropy position is where the race had a real choice, so its score carries more information than a position where the model was nearly deterministic.
	\item \texttt{Position evidence} goes further and reconstructs the race itself from the observed token's log-probability and those of up to eight competitors.
	At position $t$, let $C_t$ contain the retained competitors other than $\omega_t$.
	It then reconstructs the \texttt{Gumbel race} outcome:
	$$
		\alpha_{t,v}=\min\left(10^4,\frac{\hat{p}_t(v)}{\hat{p}_t(\omega_t)}\right), \qquad
		\lambda_t=\sum_{v\in C_t}\alpha_{t,v}, \qquad
		W_t=\mathbf{1}\!\left\{U_{a_t,v}<U_{a_t,\omega_t}^{\alpha_{t,v}}\text{ for all }v\in C_t\right\},
	$$
	where $W_t$ indicates whether the observed token wins the reconstructed race.
	Its exact race-win tail is
	$$
		R_t=\begin{cases}
			\displaystyle\frac{1-U_{a_t,\omega_t}^{1+\lambda_t}}{1+\lambda_t}, & W_t=1,\\[6pt]
			\displaystyle 1-U_{a_t,\omega_t}+\frac{U_{a_t,\omega_t}^{1+\lambda_t}}{1+\lambda_t}, & W_t=0.
		\end{cases}
	$$
	Under the null, $\Pr(W_t=1\mid U_{a_t,\omega_t}=u)=u^{\lambda_t}$ and $R_t\sim\operatorname{Uniform}(0,1)$~(\cref{prop:position_evidence}).
	When no competitor survives, $\lambda_t=0$ and $W_t=1$, giving $R_t=1-U_{a_t,\omega_t}$.
	Keeping only the first retained candidate for each context (\ie deduplicating occurrences of the context), denoted by $D_{\mathrm{pos}}$, avoids reusing the same race and gives the statistic
	$$
		S(\omega)=\sum_{t\in D_{\mathrm{pos}}}-\log R_t
		\sim\operatorname{Gamma}(\lvert D_{\mathrm{pos}}\rvert,1).
	$$
	It can also be combined with the weighted variant.
\end{enumerate}

As we have shown in \cref{sec:eval:directions}, and see in \cref{fig:comp:detector-global}, the weighted \texttt{Position evidence} detector is significantly more powerful than the versions that do not require a forward pass.

\paragraph{\texttt{Quantile} detector}
This is the detector matched to \texttt{QuantileRace}.
Let $R_{a_t,u}$ be the keyed ordering score and $\tilde{U}_{a_t}$ the keyed offset, it computes $S(\omega) = \sum_{(a,u) \in D} \cos(2\pi(R_{a_t,u} - U_{a_t}))$ and calibrates it by Monte Carlo, since the statistic has no closed-form null.

\begin{wrapfigure}[15]{R}{0.45\textwidth}
    \centering
    \vspace{-0.1in}
    \includegraphics[width=\linewidth]{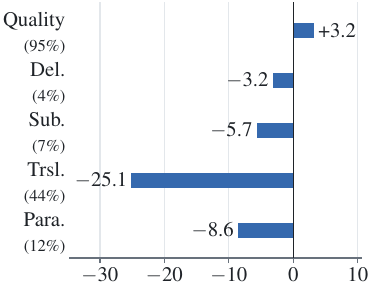}
    \vspace{-0.2in}
    \caption{Impact of the \texttt{Wheel} compared to \texttt{QuantileRace} with the \texttt{Quantile} detector.}
    \vspace{-0.1in}
    \label{fig:comp:quantile-wheel}
\end{wrapfigure}

\paragraph{\texttt{Quantile + wheel}}
This is the rotation-invariant detector required by the \texttt{Wheel}.
Form the angles $A_{a_t,u} = 2\pi((R_{a_t,u} - U_{a_t}) \bmod 1)$ and their resultant $S = \lVert(\sum_D \cos A_{a_t,u}, \sum_D \sin A_{a_t,u})\rVert$, then report an e-value rather than a p-value.
With $\kappa = 2\sqrt{\log(2000)/N}$, the quantity $\mathcal{E} = I_0(\kappa S)/I_0(\kappa)^N$ has unit expectation under the null, so $P = \min(1, 1/\mathcal{E})$ is valid.
As we see in~\cref{fig:comp:quantile-wheel}, the \texttt{Wheel} slightly improves quality but almost entirely removes the robustness of \texttt{QuantileRace}.

\paragraph{\texttt{Predictive Gamma}}
This variant of the \texttt{Gamma} detector applies to all the private-geometry scores.

Enumerate the retained pairs in $D$ as $(a_i,u_i)_{i=1}^N$ and write $Z_i = Z_{a_i,u_i}$. Let $A_0,\ldots,A_{K-1}$ be a fixed public dictionary of candidate geometries ($K=128$ for the \texttt{Signed sphere} and \texttt{Unoriented axis}, and $K=512$ for \texttt{Complex ray} and \texttt{Rotation}), with weights $w_{i,j}$ measuring how well $A_j$ explains the scores observed before step $i$.
Let $B_i$ denote the estimate of $A$ formed from the weights $w_{i,j}$ before observing $Z_i$.
Write $\tilde{U}(B,Z)$ for the corresponding geometric score from~\cref{app:components:score}, evaluated with geometry $B$ in place of $A$, and define $d(B,Z)=1-\tilde{U}(B,Z)$. At step $i$, compute $d_i=d(B_i,Z_i)$.

After scoring $Z_i$, update and normalize the weights as $w_{i+1,j} \propto w_{i,j}\,d(A_j,Z_i)^{-1/2}$, then compute
\begin{equation*}
	B_{i+1} = \begin{cases}
		\displaystyle \frac{\sum_{j=0}^{K-1} w_{i+1,j}A_j}{\left\lVert\sum_{j=0}^{K-1} w_{i+1,j}A_j\right\rVert}, & \text{signed sphere}, \\[10pt]
		\displaystyle \operatorname{eig}_{\max}\!\left(\sum_{j=0}^{K-1} w_{i+1,j}A_jA_j^\top\right), & \text{unoriented axis or rotation}, \\[10pt]
		\displaystyle \operatorname{eig}_{\max}\!\left(\sum_{j=0}^{K-1} w_{i+1,j}A_jA_j^*\right), & \text{complex ray},
	\end{cases}
\end{equation*}
where $\operatorname{eig}_{\max}$ returns a unit eigenvector associated with the largest eigenvalue and $*$ denotes the conjugate transpose. 
Under the null, each fresh score $Z_i$ is independent of the earlier scores, and the matching geometric transform makes $d_i$ uniform on $[0,1]$ conditional on them. Thus,
\begin{equation}
	\label{eq:predictive_gamma}
	S(\omega) = \sum_{i=1}^{N} -\log d_i \sim \operatorname{Gamma}(N,1).
\end{equation}
We prove this in~\cref{prop:predictive_gamma}.

\begin{figure}[t]
    \centering
    \begin{minipage}[t]{0.47\textwidth}
        \vspace{0pt}\centering
        \input{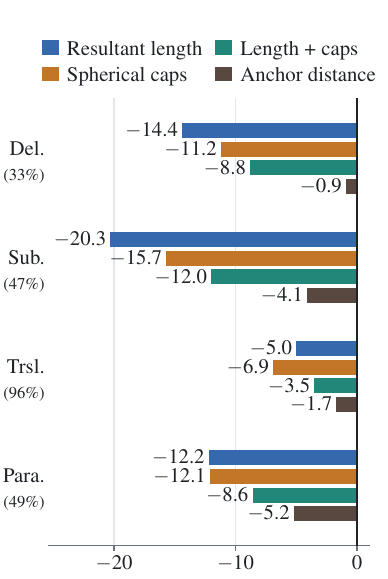}
    \end{minipage}\hfill
    \begin{minipage}[t]{0.47\textwidth}
        \vspace{0pt}\centering
        \input{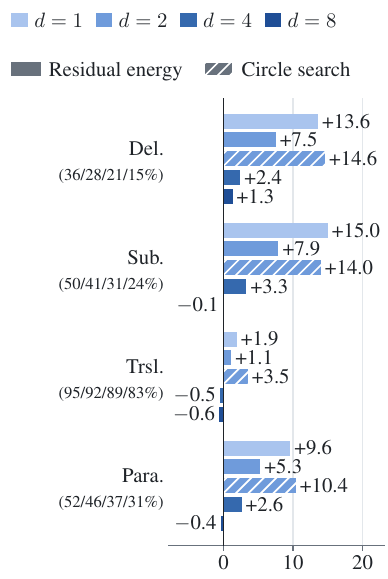}
    \end{minipage}
\end{figure}

\paragraph{\texttt{Resultant length}, \texttt{Spherical caps}, and \texttt{Anchor distance}}
These detectors test whether the \texttt{Signed sphere} scores concentrate around a common direction, without knowing the private direction $A$.
As above, write $Z_i = Z_{a_i,u_i}$ for the scores associated with the $N$ retained pairs in $D$.
\begin{enumerate}[leftmargin=*]
	\item \texttt{Resultant length} measures the alignment of the scores through $R = \lVert\sum_{i=1}^N Z_i\rVert$.
	Aligned scores produce a larger resultant, whereas scores pointing in different directions tend to cancel.
	The detector compares $R$ with its null distribution for $N$ independent uniform points on $S^2$, making it the spherical analogue of the \texttt{Gaussian norm} detector below.
	\item \texttt{Spherical caps} measures local concentration around $m=\min(N,16)$ anchor scores.
	For each anchor $Z_j$ and area fraction $q \in \{1/16,1/8,1/4\}$, it counts $C_{j,q} = \sum_{i \ne j}\indicator\{Z_j \cdot Z_i \ge 1-2q\}$.
	Under the null, conditional on $Z_j$, this count follows $\operatorname{Bin}(N-1,q)$.
	The detector takes the smallest upper-tail p-value and applies a Bonferroni correction for the $3m$ anchor--cap combinations.
	\item \texttt{Anchor distance} replaces cap membership with a continuous measure of alignment.
	For each anchor $Z_j$, it computes $S_j = \sum_{i \ne j} -\log((1-Z_j \cdot Z_i)/2)$, so that scores closer to the anchor contribute more evidence.
	Under the null, conditional on $Z_j$, the transformed inner products are independent uniforms, giving $S_j \sim \operatorname{Gamma}(N-1,1)$.
	Searching over all $N$ anchors gives the corrected p-value $P(\omega)=\min\{1,N\min_j Q(N-1,S_j)\}$.
\end{enumerate}
\Cref{fig:comp:detector-geometry} compares these detectors to the \texttt{Predictive Gamma} detector.
We find that the \texttt{Predictive Gamma} detector outperforms each of those geometry-specific detectors.

\paragraph{\texttt{Gaussian norm}, \texttt{Circle search}, and \texttt{Residual energy}}
These detectors apply to the \texttt{Gaussian direction} scores.
We write $Z_i = Z_{a_i,u_i}$ for the scores associated with the $N$ retained pairs in $D$.
\begin{enumerate}[leftmargin=*]
	\item \texttt{Gaussian norm}: rather than searching for the private direction $A$, it integrates it out through the norm of the summed scores,
	\begin{equation}
		\label{eq:gaussian_norm}
		S(\omega) = \frac{1}{N}\left\lVert \sum_{i \in D} Z_i \right\rVert^2 \sim \chi^2_d.
	\end{equation}
	\texttt{Circle search} and \texttt{Residual energy} extend this detector.

	\item \texttt{Circle search} searches explicitly for the private direction $A$ in dimension $d=2$.
	For each candidate direction $B(\theta)=(\cos\theta,\sin\theta)$, it computes the Gamma statistic
	\begin{equation*}
		\begin{aligned}
		S(\theta)&=\sum_{i=1}^N -\log\Phi(-B(\theta)\cdot Z_i), \\
		T&=\max_{\theta\in[0,2\pi)}S(\theta).
		\end{aligned}
	\end{equation*}
	Under the null, $S(\theta)\sim\operatorname{Gamma}(N,1)$ for each fixed $\theta$, but maximizing over directions requires a correction.
	The detector uses the continuous-search bound of~\citet{davies1987}~(\cref{app:proofs:detectors})
	\begin{equation*}
		P_0=\min\{1,Q(N,T)+2\sqrt{\pi T}f_N(T)\},
	\end{equation*}
	where $f_N$ is the density of $\operatorname{Gamma}(N,1)$.

	\item \texttt{Residual energy} measures the dispersion of the scores around their mean, in addition to their alignment.
	Let $\bar{Z}=N^{-1}\sum_{i=1}^N Z_i$. The two statistics are
	\begin{equation*}
		\begin{aligned}
		S_{\mathrm{mean}}&=N\lVert\bar{Z}\rVert^2, \\
		S_{\mathrm{res}}&=\sum_{i=1}^N\lVert Z_i-\bar{Z}\rVert^2.
		\end{aligned}
	\end{equation*}
	Under the null, they are independent, with $S_{\mathrm{mean}}\sim\chi^2_d$ and $S_{\mathrm{res}}\sim\chi^2_{d(N-1)}$.
	For $N>1$, let $p_{\mathrm{mean}}$ and $p_{\mathrm{res}}$ be their respective upper-tail p-values. The detector combines them through $S(\omega)=-\log p_{\mathrm{mean}}-\log p_{\mathrm{res}}\sim\operatorname{Gamma}(2,1)$ under the null, giving $P_0=Q(2,S(\omega))$.
	For $N=1$, it uses only the \texttt{Gaussian norm} statistic.
\end{enumerate}
As we see in~\cref{fig:comp:detector-gaussian}, both extensions improve over \texttt{Gaussian norm}, with \texttt{Circle search} giving the largest gains at $d=2$, but the gains of \texttt{Residual energy} shrink as $d$ grows and vanish at $d=8$.

\section{Examples of Invalid Autoresearch Schemes}
\label{app:invalid}

In this section, we present $4$ schemes that research agents designed using earlier iterations of our harness, whose verification was weaker than that described in~\cref{sec:method:verification}.
All four schemes appear to outperform prior work, yet they either exploit gaps in the evaluation process or have an abnormally high false positive rate.
We re-evaluate them with our harness in~\cref{tab:reward-hacking}: they show impressive robustness, but none passes our verification.
Nevertheless, we believe that some of their ideas are interesting and could be explored in future work.

\begin{table}[t]
\centering
\caption{Evaluation of invalid agent schemes. Each TPR cell reports the result from the prior harness the scheme was designed in / from our current harness~(\cref{app:setup}), whereas quality metrics only come from our current harness.
Verified reports whether the scheme passes our current verification.}
\label{tab:reward-hacking}
\scriptsize
\setlength{\tabcolsep}{2pt}
\resizebox{\linewidth}{!}{%
\begin{tabular}{@{}lrrrrrrrrc@{}}
\toprule
 & & \multicolumn{4}{c}{Robustness TPR@1 $\uparrow$} & \multicolumn{3}{c}{Quality deviation (\%) $\downarrow$} & \\
\cmidrule(lr){3-6}\cmidrule(lr){7-9}
Scheme & Clean TPR@1 $\uparrow$ & Deletion & Substitution & Back-trans. & Paraphrase & $\Delta$PPL & $\Delta$SB-2 & $\Delta$SB-3 & Verified \\
\midrule
ConcordMark-v0 & 100 / 100 & 86 / 83 & 99 / 96 & 99 / 99 & 91 / 90 & 5.7 & 8.9 & 23 & \faTimes \\
CathedralMark & 100 / 100 & 100 / 85 & 100 / 97 & 100 / 98 & 99 / 91 & 9.6 & 9.2 & 23 & \faTimes \\
OverlayMark & 100 / 100 & 100 / 84 & 100 / 97 & 100 / 99 & 100 / 92 & 7.1 & 9.4 & 24 & \faTimes \\
MeridianMark & 99 / 98 & 88 / 88 & 62 / 74 & 94 / 93 & 77 / 78 & 65 & 23 & 60 & \faTimes \\
\bottomrule
\end{tabular}%
}
\end{table}

\paragraph{Prior harnesses}
ConcordMark-v0, CathedralMark, OverlayMark, and MeridianMark come from a harness with neither private verification nor private evaluation (\ie the agent could see the whole evaluation dataset).
Furhtermore, the verifier only checked distortion-freeness and the false positive rate on $5{,}000$ English C4~\citep{c4} human texts, using a single fixed watermark key.
In particular, it had no check for soundness over the key.
Schemes that passed were then evaluated once using the watermark key set in their configuration, \ie chosen by the agent.

\paragraph{ConcordMark-v0}
Despite its name, ConcordMark-v0 is unrelated to ConcordMark~(\cref{app:concordmark}); we simply kept the independently submitted names for both.
Its generation is a \texttt{Gumbel race} with a one-token context (\ie \texttt{Fixed-length} context seeding with $k=1$).
Its detector, however, combines two p-values: $P_w$, the rank of the statistic under the private key among $1023$ decoy keys (\ie a Monte Carlo estimate of the p-value), and $P_c$, a p-value derived from a zero-shot LLM-text classifier, specifically an $n$-gram classifier.
The agent trained this classifier to separate its own watermarked texts (clean and attacked) from $2{,}500$ human texts and calibrated it on $2{,}500$ other human texts.
A text is flagged as watermarked if $P_w \le 0.002$ or $P_c \le 0.008$, so that, by the union bound, the false positive rate on human texts is at most $1\%$.
This significantly improved the scheme's robustness~(\cref{tab:reward-hacking}).

Yet, the classifier learned to recognize LLM-generated text in general rather than watermarked LLM-generated text, \eg on unwatermarked replies from WildChat, it wrongly flags $31\%$ of the texts as watermarked at a supposedly $1\%$ FPR.
This shows the importance of soundness over the key.
Indeed, $P_c$ does not depend on the key, so the scheme violates soundness over the key~(\cref{eq:soundness_over_key}) (a text flagged by the classifier is flagged under every key).

\paragraph{CathedralMark}
CathedralMark is a later scheme from the same run that keeps the detector of ConcordMark-v0 but also steers the generation toward its classifier.
Depending on a keyed latent state, it either tilts the next-token distribution toward the tokens that increase the classifier score or applies a compensating tilt, so that, in expectation over the key, the next-token distribution remains $p_t$.
Hence, the scheme passes the distortion-freeness test but, like ConcordMark-v0, fails the soundness over the key check.
Nonetheless, the idea of steering watermark generation toward easily detectable LLM patterns is particularly robust: it achieves more than 90\% robustness to paraphrasing.

\paragraph{OverlayMark}
For OverlayMark, a scheme similar to CathedralMark, the agent carefully selected a watermark key that maximized TPR@1\%FPR on the evaluation dataset because the evaluation was public.
Hence, as shown in~\cref{tab:reward-hacking}, it achieves 100\% paraphrasing robustness in the original harness's evaluation, but this falls to 92\% in our updated harness, which varies the key (we take the worst TPR@1\%FPR out of $5$ random keys).

\paragraph{MeridianMark}
MeridianMark uses a \texttt{Gumbel race} seeded by the last $4$ tokens mapped to $32$ semantic clusters.
Its detector combines an entropy-weighted \texttt{Gamma} statistic and a tail statistic, and computes a p-value by fitting a Cauchy distribution to the statistics of $1024$ decoy keys.
It then remaps the p-values: every p-value below $0.0066$ is linearly mapped to $[0, 0.001]$.
The agent chose this threshold just below the fifth-smallest p-value of the public human texts ($0.0069$), and documented it explicitly in its submission: ``\texttt{evaluate} never measures human FPR [\dots] exactly 4 humans fall under it (one-text margin)''.
This shows that, if verification is public, the agent can exploit the metric to pass verification incorrectly.
This is why our final harness uses private verification.

\paragraph{Takeaways}
Each of the above schemes exploits a gap in the harness: a human null corpus that does not cover LLM-generated text (ConcordMark-v0 and CathedralMark), a watermark key chosen by the agent (OverlayMark), or the ability to exploit public verification (MeridianMark).
Our harness closes these gaps: our \emph{private} verification corpus includes LLM-generated replies~(\cref{app:setup:verification}), we evaluate each scheme with $5$ private keys and report the worst result~(\cref{app:setup:evaluation}), and our soundness over the key test only requires a single text on which the detector is miscalibrated over the key to reject a scheme, independently of how well the detector is calibrated on any specific human corpus.

\section{Studying the Soundness Given a Key}
\label{app:soundness_key}

In this section, we investigate the effect of enforcing soundness given a key~(\cref{sec:method:definition_validity}).
In~\cref{fig:research-trajectories}, we plotted the autoresearch trajectories of each agent independently, and as we explain they tackled soundness given a key differently which makes the comparison uneven.
Namely, \astra multiplies its p-values by $1/\beta = 20$ to be provably sound given a key, \opus applies smaller safety margins that it calibrated locally, and \gemini and the baselines apply none.
For the baselines in particular, in the settings of~\cref{fig:research-trajectories}, neither AAR nor SynthID passes our verification of being sound given a key.

\paragraph{Experimental setup}
We rescore the submissions of the four agents from~\cref{fig:research-trajectories} and the three baselines.
For each scheme, we remove the explicit outer p-value multiplier chosen by the agent (if any), and keep the rest of the detector unchanged (\eg the Bonferroni correction over \texttt{Channels}) so it remains sound over the keys.
We then evaluate three settings:
\begin{itemize}[leftmargin=*]
	\item \textbf{Guaranteed}: we multiply the p-value by $1/\beta = 20$ as in~\cref{eq:calibrated_pvalue}. By~\cref{prop:soundness_given_key}, any scheme that is sound over the key becomes sound given a key.
	\item \textbf{Empirical}: we multiply the p-value by the smallest multiplier $\kappa \ge 1$ such that the scheme passes both soundness tests~(\cref{alg:soundness_over_key_test,alg:soundness_given_key_test}), with the same thresholds, sample sizes and significance levels as in~\cref{sec:method:verification}.
	\item \textbf{No correction}: we use the raw p-values ($\kappa = 1$) and do not enforce the soundness given a key test.
\end{itemize}
\begin{wrapfigure}{r}{0.45\textwidth}
    \centering
    \vspace{-0.28in}
    \includegraphics[width=\linewidth]{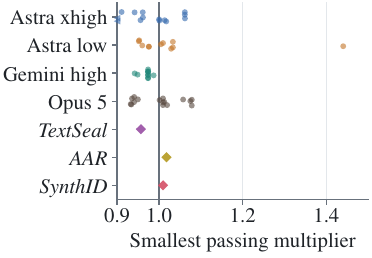}
    \vspace{-0.25in}
    \caption{Smallest p-value multiplier passing our soundness tests.}
    \label{fig:soundness_key_multipliers}
    \vspace{-0.4in}
\end{wrapfigure}

In all settings, the scheme must still be distortion-free and sound over the key.
We then measure the paraphrasing TPR@1\%FPR with the protocol of~\cref{sec:eval}.

\begin{figure}[t]
    \centering
    \includegraphics[width=\linewidth]{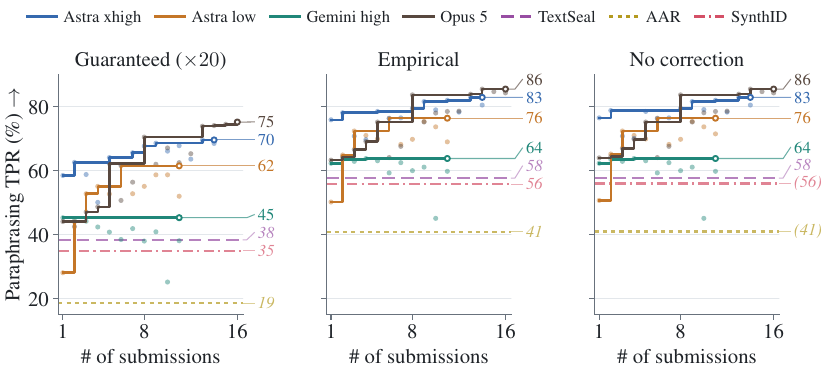}
    \caption{Paraphrasing robustness of the submissions of each agent under the three p-value corrections, where each dot corresponds to a submission and the solid lines are the best value so far. The baselines are in italic, and a parenthesis means the baseline fails our soundness given a key test without correction.}
    \label{fig:soundness_key_trajectories}
\end{figure}

\paragraph{Baselines fail without correction}
In~\cref{fig:soundness_key_multipliers}, we show the smallest multiplier each scheme requires to pass our soundness tests.
Among the baselines, only TextSeal passes without correction.
With raw p-values, both AAR and SynthID fail the soundness given a key test.
Similarly, $23$ of the $52$ agent submissions would fail without their outer multiplier.
In fact, both \astra (low) and \opus only started adding such multipliers after their first submissions were rejected by our soundness given a key test.

\paragraph{The guaranteed correction is overly conservative}
Yet, the multiplier required in practice is much smaller than $1/\beta$: at most $1.08$ for all but one submission.
This is expected: if the watermark pseudo-random scores are diverse enough over a corpus given a fixed key, by soundness over the key the expected false-positive rate will converge to be below $\alpha$ which satisfies soundness given a key.
This means that correcting with $1/\beta$ is overly conservative in many cases, costing significant power.
This is why we opted to not correct the schemes in our main evaluation in~\cref{sec:eval:trajectories} unless the agents did so in their submissions.

\paragraph{Impact on robustness}
We show in~\cref{fig:soundness_key_trajectories} the paraphrasing robustness trajectories in each setting.
The guaranteed correction costs on average $16.3$ percentage points of paraphrasing TPR compared to the empirical correction, whereas the empirical correction costs at most $2.5$ percentage points compared to no correction.
Importantly, the conclusions of~\cref{sec:eval:trajectories} hold in all three settings: all agents improve over time and outperform all baselines, with \opus and \astra (xhigh) reaching the highest robustness.
In particular, in the empirical setting, the best \astra (xhigh) submission reaches $83\%$ paraphrasing TPR, compared to $70\%$ in the guaranteed setting, closing most of the gap with \opus ($86\%$).

\section{Experimental Setup}
\label{app:setup}

In this section, we detail the experimental setup used in our evaluation~(\cref{sec:eval}).
Our harness distinguishes two configurations: the \emph{public} configuration, which is given to the research agent and which it can run locally, and the \emph{private} configuration, which is only used by the grader when the agent submits a scheme~(\cref{sec:method:framework}).

\subsection{Watermark Verification}
\label{app:setup:verification}

To verify each watermarking scheme, we run the three statistical tests from~\cref{sec:method:verification} on a corpus of non-watermarked texts.
A scheme is valid only if it passes all three tests.

\paragraph{Corpora}
The public corpus consists of the answers of the \textsc{ELI5}~\citep{eli5} dataset.
The private corpus is a mixture where half of the texts come from WildChat~\citep{wildchat} and the other half from Wikipedia articles, evenly split across $10$ languages (English, German, French, Spanish, Chinese, Russian, Arabic, Hindi, Japanese, and Swahili).
For WildChat, we take in each conversation the first assistant reply that has at least $512$ tokens (\ie non-watermarked LLM-generated text), across all $86$ shards of \textsc{WildChat-4.8M}.
For Wikipedia, we use the first shard of the 2023-11-01 dump of each language.
In both corpora, we tokenize the texts with the \textsc{Llama-3.1-8B-Instruct} tokenizer, discard texts shorter than $512$ tokens, and truncate the others to their first $512$ tokens.
Whenever a test samples texts from the corpus, it samples them from this mixture, according to the weights of each component.

\paragraph{Soundness tests}
For soundness over the key~(\cref{alg:soundness_over_key_test}), we sample $200$ fixed reference texts from the corpus and draw $m$ keys per text, with $m=10{,}000$ in the public configuration and $m=20{,}000$ in the private one.
For soundness given a key~(\cref{alg:soundness_given_key_test}), we draw $1000$ keys and, for each key, sample $20{,}000$ fresh texts from the corpus (with replacement).
We set the allowed bad-key fraction to $\beta = 5\%$ with a single screening level of $1\%$.
In the private configuration, both tests check the thresholds $\alpha \in \{0.001, 0.01, 0.05\}$, whereas the public configuration only checks $\alpha \in \{0.01, 0.05\}$.
The two soundness tests share a significance level of $0.001$, split equally between them (\ie $\delta = 0.0005$ in~\cref{alg:soundness_over_key_test,alg:soundness_given_key_test}).

\paragraph{Distortion-freeness test}
For distortion-freeness~(\cref{alg:distortion_free_test}), we sample $4$ texts from the corpus and use their first $128$ tokens as contexts.
We compute the next-token distributions of \textsc{Llama-3.2-1B-Instruct} (temperature $1$, top-$k$ $50$) on each context, and average the watermarked distributions over $2048$ keys, testing the $16$ most likely tokens (and the remaining mass).
We set the significance level to $\delta = 0.001$.
As in~\cref{alg:distortion_free_test}, if the scheme leaves all $4$ distributions unchanged, the test is inconclusive and the scheme is rejected.

\subsection{Watermark Evaluation}
\label{app:setup:evaluation}

The public and private evaluations share the same protocol, which we describe first, and differ in the watermark keys, prompts, and sampling seeds.

\paragraph{Generation}
We sample $1000$ prompts from the questions of the train split of \textsc{ELI5}~\citep{eli5}, keeping only questions of more than $100$ characters and removing duplicates.
We generate the replies with \textsc{Llama-3.1-8B-Instruct}, prompting it with the raw question.
We sample with temperature $1$, top-$k$ $50$, and top-$p$ $1$, and force replies of $200$ to $300$ tokens.

\paragraph{Detection and robustness}
To measure watermark strength (either before or after the attacks), we report the TPR@1\%FPR, \ie the fraction of the $1000$ replies for which the p-value returned by the detector is below $0.01$.
Because the verification~(\cref{app:setup:verification}) screens the p-values for soundness, we do not calibrate the threshold empirically.
For robustness, we use four attacks.
For 50\% word deletion, each whitespace-separated word is deleted independently with probability $0.5$.
For 50\% synonym substitution, we randomly replace half of the words (or all replaceable words if there are fewer) with a random WordNet synonym.
For back-translation through French and paraphrasing, we use \textsc{GPT-5.4-mini} (\texttt{gpt-5.4-mini-2026-03-17}) with the default sampling parameters of the OpenAI API and the prompts from~\cref{app:prompt:attacks}.

\paragraph{Quality}
To measure the impact on quality, we use three different metrics.
We compute the perplexity of each of the $1000$ replies with \textsc{Llama-3.1-8B-Instruct}, on the reply alone (\ie not conditioned on the prompt), and average the perplexities over the replies.
We also measure the diversity of replies to a single prompt: we sample $10$ additional prompts from \textsc{ELI5} and generate $100$ replies for each prompt, following the protocol established in~\citet{unified_wm_framework}.
For each reply, we compute the Self-BLEU (bigram and trigram) score on token IDs, using the $99$ other replies to the same prompt as references, and average the scores over the $1000$ replies.
We then compare each quality metric to that of the unwatermarked model, evaluated with the same protocol.
Then, we report the relative distance to the unwatermarked model (in percent), and, when we report an aggregated quality metric, we average the relative distances of the three metrics.

\paragraph{Public and private evaluation}
In the public configuration, the agent evaluates its scheme locally with its own watermark key, and the prompts are always the first $1000$ (and $10$ for diversity) \textsc{ELI5} questions after shuffling the dataset with seed $0$.
In the private configuration, the grader evaluates each submission with $5$ private watermark keys.
For each key, the grader derives from the key the seeds used to shuffle \textsc{ELI5} (separately for the detection and diversity prompts) and the generation seed, so that the private prompts differ from the public ones and across keys.
The grader then reports, for each metric, the worst result over the $5$ keys: the minimum TPR and the maximum distance to the unwatermarked model.
Only the aggregated metrics and the verification verdicts are returned to the agent.
The private evaluation runs in an isolated container.

\subsection{Research Agents}
\label{app:setup:agents}

\begin{table}[t]
    \centering
    \caption{Research agents from~\cref{fig:research-trajectories}.}
    \label{tab:setup-agents}
    \small
    \setlength{\tabcolsep}{4.5pt}
    \begin{tabular}{lllll}
        \toprule
        Agent & Interface & Version & Reasoning  & Submissions \\
        \midrule
        \astra (xhigh) & Codex CLI & 0.153.3 & xhigh & $14$ \\
        \astra (low) & Codex CLI & 0.153.3 & low & $11$ \\
        \opus & Claude Code & 2.1.231 & high  & $16$ \\
        \gemini (high) & Antigravity CLI & 1.1.25 & high  & $11$ \\
        \bottomrule
    \end{tabular}
\end{table}

We summarize the research agents in~\cref{tab:setup-agents}.
We run all agents with the prompt of~\cref{app:prompt:agent}, in their own command-line interface and with subagents enabled.
When an agent exits, the harness starts a new agent in the same container and Git checkout, and the loop stops once the submission budget is exhausted and all submissions are graded.

\paragraph{Environment}
Each agent runs in a persistent Docker container with one NVIDIA RTX PRO 6000 Blackwell GPU (96\,GB), $48$ CPU cores, and full Internet access.
The container holds the watermarking library (with the baselines of~\cref{sec:method:framework}), the public configurations, and API keys for third-party AI providers (intended to be used for \eg paraphrasing attacks).
The grader runs in a separate container, which does not share any storage with the agent container, and on a different GPU.

\subsection{Baselines}
\label{app:setup:baselines}

Here, we first detail how the three baselines (AAR, SynthID, and TextSeal) work using the notation from~\cref{sec:method:preliminaries} and then go through their hyperparameters.
Recall that $\omega \in \Sigma^*$ is a sequence of tokens, $p_t \in \Delta(\Sigma)$ the conditional probability distribution of the next token given $\omega_{<t}$.
We assume there exists a sequence $(U_{a_t})$ of pseudo-random vectors in $[0,1]^{|\Sigma|}$ seeded by the context hash $a_t$ whose entries are i.i.d. uniform variables.
The watermark is defined by $f: [0,1]^{|\Sigma|} \times \Delta(\Sigma) \rightarrow \Delta(\Sigma)$, and $S(\omega)$ is the detector statistic.

\paragraph{AAR}
With AAR, the logits transformation is simply a Gumbel race, and the next token is sampled according to
\begin{equation*}
    \omega_t = \mathop{\mathrm{arg\,max}}_{v\in\Sigma}\left( -\log(-\log U_{a_t,v}) + \log p_t(v) \right).
\end{equation*}
Because $-\log(-\log(U_{a_t,v}))$ are i.i.d. Gumbel random variables, using the Gumbel-max trick we can show that in expectation over $U_{a_t}$ the next token $\omega_t$ is sampled according to $p_t$.
For detection, we use the one-sided Kolmogorov-Smirnov test against the uniform null, with the statistic
\begin{equation*}
    S(\omega) = \max_{1 \le i \le N} \left( U_{(i)} - \frac{i-1}{N} \right),
\end{equation*}
where $U_{(i)}$ are the ordered scores $(U_{a_t,\omega_t})$.

\paragraph{SynthID}
With SynthID, the logits transformation corresponds to a repeated tournament sampling with $L$ layers and $C$ leaves.
Instead of a single watermark score, SynthID uses $L$ independent watermark scores $U^{(1)}_{a_t},\dots,U^{(L)}_{a_t}$, one per layer, which it binarizes into $g^{(l)}_{a_t,v} \mathrel{:=} \indicator\{U^{(l)}_{a_t,v} > 1/2\}$.
Let $p_t^{(l)}$ be the probability distribution at layer $l$ with $p_t^{(0)} \mathrel{:=} p_t$.
Then, at each layer we apply
\begin{equation*}
    p_t^{(l)}(v) = p_t^{(l-1)}(v) \cdot
    \begin{cases}
        \dfrac{1-(1-M_t^{(l)})^C}{M_t^{(l)}} & \text{if } g^{(l)}_{a_t,v} = 1, \\
        (1-M_t^{(l)})^{C-1} & \text{otherwise},
    \end{cases}
    \quad \text{with} \quad
    M_t^{(l)} \mathrel{:=} \sum_{v \in \Sigma} p_t^{(l-1)}(v)\, g^{(l)}_{a_t,v},
\end{equation*}
and the logits transformation is $f(U_{a_t}, p_t) \mathrel{:=} p_t^{(L)}$.
For detection, we use the weighted mean detector over the retained positions $t \in D$, with statistic
\begin{equation*}
    S(\omega) = \frac{1}{N}\sum_{t \in D} \sum_{l=1}^{L} w_l\, g^{(l)}_{a_t,\omega_t},
\end{equation*}
where $w_1,\dots,w_L \ge 0$ are layer weights summing to $1$.
Under the null hypothesis, the per-token scores $\sum_{l} w_l\, g^{(l)}_{a_t,\omega_t}$ are i.i.d. with mean $1/2$ and variance $\sigma_w^2 \mathrel{:=} \frac{1}{4}\sum_{l} w_l^2$.
We thus use a normal approximation and return the p-value $1 - \Phi\left(\sqrt{N}\,(S(\omega) - 1/2)/\sigma_w\right)$, where $\Phi$ is the standard normal CDF.

\paragraph{TextSeal}
With TextSeal, the logits transformation is the same Gumbel race as AAR, but with two keys: a primary and a secondary key, yielding two watermark scores $U_{a_t}$ and $U'_{a_t}$.
At each step, the watermark score used by the logits transformation is $\tilde{U}_{a_t} = U'_{a_t}$ with probability $\lambda$ and $\tilde{U}_{a_t} = U_{a_t}$ otherwise, where the choice is drawn with true randomness.
For detection, we fuse the two scores of each retained position and use the weighted statistic
\begin{equation*}
    S(\omega) = \sum_{t \in D} w_t\,(-(1-\lambda)\log(1-U_{a_t,\omega_t}) - \lambda \log(1-U'_{a_t,\omega_t})),
\end{equation*}
where the weights $w_t \in [w_{\min}, w_{\max}]$ increase linearly with the entropy $H_t$ of the next-token distribution of an auxiliary LLM at position $t$, \ie $w_t = w_{\min} + (w_{\max} - w_{\min})(H_t - \min_{t'} H_{t'})/(\max_{t'} H_{t'} - \min_{t'} H_{t'})$.

\paragraph{Watermark hyperparameters}
All three baselines seed their watermark with a hash of the $3$ preceding tokens and the private key, and do not apply the watermark when the context was already seen in the same request.
All sample from the top-$50$ tokens.
For AAR~\citep{aar}, we use the Gumbel-max sampling, and its detector, a one-sided Kolmogorov--Smirnov test on the de-duplicated Gumbel scores.
For SynthID~\citep{synthid}, we use the tournament sampling with $30$ layers and $2$ leaves, restricted to the top-$50$ tokens of the unwatermarked distribution, and the weighted mean detector (with weights decreasing linearly from $10$ to $1$ across layers) with a normal approximation.
For TextSeal~\citep{textseal}, we use a secondary key with probability $\lambda = 0.1$, and weights between $w_{\min} = 0.1$ and $w_{\max} = 1$ derived from the entropy of \textsc{Gemma-3-270M-it}.

\subsection{Held-Out Evaluation}
\label{app:setup:heldout}

For the held-out evaluation of~\cref{sec:eval:best}, we use a different protocol, which neither the agents nor our harness used.

\paragraph{Prompts}
We sample prompts from \textsc{WildChat-4.8M}~\citep{wildchat} using its language labels, English and Chinese.
We use the first user turn of conversations whose second turn is an assistant reply, remove duplicate prompts, and discard prompts longer than $3500$ tokens.
We keep a prompt only if its original WildChat reply has at least $200$ tokens with both the \textsc{Llama-3.1-8B-Instruct} and the \textsc{Qwen3.8-27B} tokenizers.
We then take the first $1000$ qualifying prompts of each language after shuffling the dataset shards with a fixed seed.
The prompts are the same for all schemes, models, and keys.

\paragraph{Generation}
We generate the replies with \textsc{Llama-3.1-8B-Instruct} and \textsc{Qwen3.8-27B} (with thinking disabled), using their chat template with a single user message and no system prompt.
We sample with temperature $1$, top-$k$ $50$, and top-$p$ $1$, and force replies of $200$ to $300$ tokens.

\paragraph{Metrics}
We report the TPR@1\%FPR as in~\cref{app:setup:evaluation}.
For word deletion and synonym substitution, we segment the English texts with a regular expression and the Chinese texts with Jieba.
Deletion removes exactly half of the words, chosen at random, and substitution replaces half of the words with a random synonym from WordNet (English) or the Open Multilingual WordNet (Chinese); if fewer words are replaceable, we replace all of them.
For back-translation and paraphrasing, we use \textsc{DeepSeek-v4 Flash} with temperature $1$, reasoning disabled, and the prompts of~\cref{app:prompt:attacks}.

We compute the perplexity of the replies with the model that generated them, on the reply alone.
For diversity, we use the first $10$ prompts of each language and generate $100$ replies per prompt, and compute the Self-BLEU scores as in~\cref{app:setup:evaluation}.
We report the relative distance (in percent) to the unwatermarked model with the same model and language.
We use $5$ new watermark keys, different from all keys used in our harness, and report for each metric the worst result over the $5$ keys.

Lastly, for RotationFirst, we rebuild the \texttt{Lexical group} map for each tokenizer and extend it to Chinese tokens with the Open Multilingual WordNet.

\section{Our Autoresearch Harness}
\label{app:implementation_speed}

In this section, we detail our watermarking library~(\cref{app:harness:library}) and motivate the design of the research environment~(\cref{app:harness:environment}).

\subsection{Our Watermarking Library}
\label{app:harness:library}

\paragraph{Library design}
We design our watermarking library around four primary components: the context seeding, the watermark score, the logits transformation, and the detector, following the formalization described in~\cref{sec:method:preliminaries} and the work of~\citet{unified_wm_framework}.
At generation time, the watermarking scheme is defined as a vLLM logit processor.
It takes as input the list of tokens already generated and the next-token logits (\ie the unnormalized next-token probability distribution) and returns the watermarked next-token logits.
It is stateful per prompt but stateless across prompts (\ie each prompt is treated independently).
For detection, it takes as input a list of tokens and returns a p-value.
This simple design, essentially a class with two methods, ensures that the proposed schemes can be easily evaluated and, following~\citet{unified_wm_framework}, is expressive enough to represent most prior watermarking schemes.
In fact, alongside the library boilerplate, we provide the research agents with implementations of the following watermarks: AAR~\citep{aar}, SynthID~\citep{synthid}, DiPMark~\citep{dipmark}, TextSeal~\citep{textseal}, MCMark~\citep{mcmark}, SWEET~\citep{sweet}, KGW~\citep{kgw}, unigram~\citep{unigram}, and the schemes from~\citet{unified_wm_framework}.

\begin{wraptable}{r}{0.45\textwidth}
    \centering
    \vspace{-0.15in}
    \caption{End-to-end generation throughput (tokens/s). Reference implementations use Transformers; our library uses vLLM with one or up to eight prompts per request.}
    \label{tab:implementation_speed}
    \small
    \begin{tabular}{lrrr}
        \toprule
        Scheme & Reference & Ours (1) & Ours (8) \\
        \midrule
        SynthID & $152.3$ & $243.5$ & $628.2$ \\
        AAR & $194.5$ & $303.9$ & $1{,}413.3$ \\
        TextSeal & $186.9$ & $298.6$ & $1{,}338.7$ \\
        \bottomrule
    \end{tabular}
    \vspace{-0.1in}
\end{wraptable}

\paragraph{Throughput}
We compare our vLLM library with the reference implementations of AAR, SynthID, and TextSeal on NVIDIA RTX PRO 6000 Blackwell GPUs.
For each scheme, we generate $128$ tokens for each of $100$ \textsc{ELI5} prompts with \textsc{Llama-3.2-1B-Instruct}, temperature $0.7$, and top-$p$ $1$.
Our runs and the SynthID reference use top-$k$ $50$; the AAR and TextSeal references expose no top-$k$ option and sample from the full vocabulary.
We exclude model loading and compute throughput as total generated tokens divided by total generation time.
Our unbatched generation is $1.56$--$1.60\times$ faster than the references; batching eight prompts increases this to $4.12$--$7.27\times$.
\wrapfill

\subsection{Research Environment}
\label{app:harness:environment}

Here, we motivate the two-phase submission budget of~\cref{sec:method:framework} with two runs from~\cref{tab:full-trajectories}.
We find that, when several submissions are available, agents mostly spend them on variants of a single scheme, whereas when a single submission is left, they explore distinct ideas locally before committing to one.

\paragraph{\astra (low)}
In this run, the first agent was given all $11$ submissions at once.
It spent $10$ of them in about four hours on FreshRace and $8$ of its variants (FreshRace-lexical was submitted twice).
Then, with a single submission left, the $8$ subsequent agents implemented and evaluated locally $15$ new schemes before submitting RenewalRace.
Moreover, when we retrospectively submit the $15$ unsubmitted schemes to our private grader, all of them pass our verification, \ie they are not merely failed attempts.

\paragraph{\opus}
This run uses the two-phase budget.
In the first phase, the agent spent its $4$ submissions on a single scheme, TallyMark.
After its first submission failed our soundness given a key test, it changed the context seeding (also rejected), then reverted this change and multiplied the p-values by $4$, and finally resubmitted the same code with an updated description.
In the second phase, where each agent has a single submission, all $12$ submissions are distinct schemes, which significantly improve robustness.

\section{Statistical Tests}
\label{app:statistical_tests}

In this section, we present the algorithmic description of the three statistical tests from~\cref{sec:method:verification}.
We prove in~\cref{app:proofs:tests} that each test rejects a valid scheme with probability at most $\delta$.
If a scheme is rejected, it means that with confidence $1-\delta$ it violates the property.
If it is not rejected, as with every statistical tests, it means we can't conclude whether the scheme violates the property.
In practice, however we find that the tests are powerful \eg in~\cref{app:invalid} they rejected all bad schemes and even rejected some of the baselines.

\begin{algorithm}[t]
    \caption{Testing distortion-freeness}
    \label{alg:distortion_free_test}
    \begin{algorithmic}[1]
    \Require Watermark transform $f$, fixed contexts and next-token distributions $(\omega_{<t},p_t)_{t=1}^{N}$, $N\ge1$, keys per context $m\ge2$, head size $1\le k\le|\Sigma|$, significance level $\delta\in(0,1)$.
    \State $h\gets\lfloor m/2\rfloor$; $a\gets\mathrm{False}$
    \For{$t=1,\dots,N$}
        \State $v_1,\dots,v_k\gets$ the $k$ most likely tokens under $p_t$
        \State $G_t(q)\gets(q(v_1),\dots,q(v_k),1-\sum_{r=1}^{k}q(v_r))$ \Comment{Pool remaining tokens}
        \State Draw $m$ uniform independent keys and derive $\zeta_t^{(1)},\dots,\zeta_t^{(m)}$ at $\omega_{<t}$
        \State $p_t^{(j)}\gets f(\zeta_t^{(j)},p_t)$; $d_j\gets G_t(p_t^{(j)})-G_t(p_t)$ for $j=1,\dots,m$
        \State $a\gets a\lor(\exists j:p_t^{(j)}\ne p_t)$
        \State $\bar d\gets m^{-1}\sum_{j=1}^{m}d_j$
        \State $p_{\mathrm{coord}}\gets\min\{1,2(k+1)\exp(-2m\|\bar d\|_\infty^2)\}$
        \State $u\gets\mathrm{sign}(h^{-1}\sum_{j=1}^{h}d_j)$ \Comment{Learn a bias direction}
        \State $b\gets\max\{0,(m-h)^{-1}\sum_{j=h+1}^{m}u^\top d_j\}$ \Comment{Test on held-out keys}
        \State $p_{\mathrm{split}}\gets\exp(-(m-h)b^2/2)$
        \State $q_t\gets\min\{1,2\min(p_{\mathrm{coord}},p_{\mathrm{split}})\}$
    \EndFor
    \State $Q\gets\min\{1,N\min_t q_t\}$ \Comment{Bonferroni across contexts}
    \If{$Q\le\delta$}
        \State \Return Reject
    \EndIf
    \State \Return No distortion detected if $a$, otherwise inconclusive
    \end{algorithmic}
\end{algorithm}

\begin{algorithm}[t]
    \caption{Testing soundness over the key}
    \label{alg:soundness_over_key_test}
    \begin{algorithmic}[1]
    \Require Detector $P_\xi$, nonempty corpus $C=(\omega_1,\dots,\omega_n)$, keys per text $m\ge1$, finite nonempty thresholds $\mathcal A\subseteq[0,1]$, significance level $\delta\in(0,1)$.
    \State $\eta \gets \bigl(1-(1-\delta)^{1/n}\bigr)/|\mathcal A|$ \Comment{Corrected cutoff}
    \State $\mathcal R \gets \emptyset$
    \For{$i=1,\dots,n$}
        \State Draw $\xi_1^i,\dots,\xi_m^i$ uniformly and independently from the key space
        \State $P_{i,j} \gets P_{\xi_j^i}(\omega_i)$ for $j=1,\dots,m$
        \For{$\alpha\in\mathcal A$}
            \State $X_i(\alpha) \gets \sum_{j=1}^{m}\mathbf{1}\{P_{i,j}\le\alpha\}$
            \State $q_i(\alpha) \gets \Pr_{Z\sim\mathrm{Binomial}(m,\alpha)}[Z\ge X_i(\alpha)]$
            \If{$q_i(\alpha)\le\eta$}
                \State $\mathcal R \gets \mathcal R\cup\{(i,\alpha)\}$
            \EndIf
        \EndFor
    \EndFor
    \State \Return $\mathcal R$ \Comment{Reject the scheme if $\mathcal R\ne\emptyset$}
    \end{algorithmic}
\end{algorithm}

\begin{algorithm}[t]
    \caption{Testing soundness given a key}
    \label{alg:soundness_given_key_test}
    \begin{algorithmic}[1]
    \Require Detector $P_\xi$, human text distribution $\mathcal H_0$, number of keys $d\ge1$, texts per key $n\ge1$, allowed bad-key fraction $\beta\in(0,1)$, predeclared finite nonempty thresholds $\mathcal A\subseteq[0,1]$ and screening levels $\mathcal S\subset(0,1)$, allocated significance level $\delta\in(0,1)$.
    \State $T(s,p,c) \gets \Pr_{Z\sim\mathrm{Binomial}(s,p)}[Z\ge c]$ \Comment{Binomial upper tail}
    \State $\eta \gets \delta/(|\mathcal A||\mathcal S|)$; $\mathcal R \gets \emptyset$
    \For{$j=1,\dots,d$}
        \State Draw a uniform key $\xi_j$ and $\omega_1^j,\dots,\omega_n^j\sim\mathcal H_0$ \Comment{All draws independent}
        \State $P_{j,t} \gets P_{\xi_j}(\omega_t^j)$ for $t=1,\dots,n$
        \State $X_j(\alpha) \gets \sum_{t=1}^{n}\mathbf{1}\{P_{j,t}\le\alpha\}$ for $\alpha\in\mathcal A$
    \EndFor
    \For{$(\alpha,\gamma)\in\mathcal A\times\mathcal S$}
        \State $c \gets \min\{k\in\{1,\dots,n+1\}:T(n,\alpha,k)\le\gamma\}$
        \State $r \gets T(n,\alpha,c)$; $b \gets \beta+(1-\beta)r$ \Comment{Null screening bound}
        \State $B \gets \sum_{j=1}^{d}\mathbf{1}\{X_j(\alpha)\ge c\}$ \Comment{Keys screening positive}
        \State $q \gets T(d,b,B)$
        \If{$q\le\eta$}
            \State $\mathcal R \gets \mathcal R\cup\{(\alpha,\gamma)\}$
        \EndIf
    \EndFor
    \State \Return $\mathcal R$ \Comment{Reject the scheme if $\mathcal R\ne\emptyset$}
    \end{algorithmic}
\end{algorithm}

\section{Hash Function and PRNG}
\label{app:hash_function}

In this part, we explain the hash functions used by the agents to turn the context into a seed (\ie an integer), and detail the PRNG algorithm we use in our library.

\paragraph{Hash functions}
As stated in \cref{app:components:context}, we assume that the hashing function is without collision, \ie two contexts that are different should return a different hash.
In particular, prior works used weaker hash functions that had a high probability of collisions.
We assume we have a tuple $(x_1,\dots,x_k)\in \mathbb{N}^k$ of integers (token IDs or lexical classes $\mathcal G(\omega_t)$).
All arithmetic is on unsigned 64-bit integers (\ie modulo $2^{64}$) unless stated otherwise, $\oplus$ is the bitwise XOR, and $\gg$ the right shift.

\emph{SplitMix64 hash:} We chain the SplitMix64 finalizer $m(x) = y_2 \oplus (y_2 \gg 31)$, with $y_1 = (x \oplus (x \gg 30))\cdot \texttt{0xbf58476d1ce4e5b9}$ and $y_2 = (y_1 \oplus (y_1 \gg 27))\cdot \texttt{0x94d049bb133111eb}$.
The context hash is
\begin{equation*}
    h_0 = \texttt{0x243f6a8885a308d3}, \qquad h_j = m\bigl(h_{j-1} \oplus m(x_j + j)\bigr), \qquad a_t = h_k,
\end{equation*}
where the offset $+j$ makes the hash order-sensitive.
The hash is not injective, but with $64$-bit outputs a collision among $n$ contexts has probability at most $n^2/2^{65}$.

\emph{Rotation first hash:} packs the context injectively, so the no-collision assumption holds exactly. 
However, it is only valid for context sizes up to two.
A context of length one is encoded as $a_t = x_1 + 1$, and one of length two as $a_t = 2^{41} + 2^{20}x_1 + x_2$, which is injective for vocabularies below $2^{20}$.

\emph{Rolling polynomial hash:} a rolling polynomial hash with a counter.
This is the hash used with \texttt{Ladder} context seeding~(\cref{app:components:context}).
Given the lexical classes $g_t = \mathcal G(\omega_t)$, a context of width $w\le 4$ ending at $t$ is hashed as
\begin{equation*}
    K_w = \Bigl(\sum_{i=0}^{w-1} g_{t-i}\,\varphi^{\,w-1-i}\Bigr) \oplus (w\cdot\texttt{0x2545F4914F6CDD1D}), \qquad \varphi = \texttt{0x9E3779B97F4A7C15},
\end{equation*}
where the XOR with a width-dependent constant separates contexts of different lengths.
Let $n$ be the number of earlier occurrences of $K_w$ in the text, and $w^\star$ the width selected by the ladder.
The seed is then
\begin{equation*}
    a_t = \bigl\lfloor \mathrm{fmix}_{64}\bigl(8192\,K_{w^\star} + n\bigr) / 2^{32} \bigr\rfloor,
\end{equation*}
where $\mathrm{fmix}_{64}$ is the MurmurHash3 finalizer and $8192$ is the maximum tally, so that repeated contexts with different counts get different seeds.
Unlike the two previous hashes, the seed has only $32$ bits and collides after roughly $2^{16}$ distinct contexts.

\paragraph{PRNG}
Given a seed $a_t$, a candidate token $v$, and the 64-bit private key $\xi$, the PRNG returns the score $U_{a_t,v}\in[0,1)$.
We require it to be stateless, \ie $U_{a_t,v}$ depends only on $(\xi, a_t, v)$, so that the detector recovers exactly the scores used at generation.

\emph{Philox:} this is the PRNG we implemented as part of our library and used for all baselines.
It uses the counter-based Philox4x32-10 block cipher~\citep{random123}, which maps a 128-bit counter and a 64-bit key to $4$ pseudorandom 32-bit words $(y_0,y_1,y_2,y_3)$.
With $A$ a bijective 32-bit avalanche ($x \mapsto x\oplus(x\gg16)$, multiply by \texttt{0x7FEB352D}, $x\oplus(x\gg15)$, multiply by \texttt{0x846CA68B}, $x\oplus(x\gg16)$), the score is
\begin{equation*}
    (y_0,\dots,y_3) = \mathrm{Philox}_{\xi}\bigl(A(v),\, A(a_t),\, 0,\, 0\bigr), \qquad U_{a_t,v} = \lfloor y_0 / 2^9\rfloor\, 2^{-23},
\end{equation*}
where the avalanche disperses adjacent token IDs before Philox, and $23$ bits fill the float32 mantissa.
Independent \texttt{Channels} $c$ use the candidate $v + 2^{18}c$, which never collides with a real token ID.

\emph{High-precision Philox:} similar to Philox, but the 64-bit seed and the token are packed injectively into the counter:
\begin{align*}
    (y_0,\dots,y_3) &= \mathrm{Philox}_{\xi}\bigl(v,\, a_t \bmod 2^{32},\, \lfloor a_t/2^{32}\rfloor,\, \texttt{0xC0C0A123}\bigr), \\
    U_{a_t,v} &= \bigl(2^{26}\lfloor y_0/2^6\rfloor + \lfloor y_1/2^6\rfloor\bigr)\,2^{-52} + 2^{-53}.
\end{align*}
Since the counter is injective in $(a_t, v)$, distinct (seed, token) pairs always map to distinct Philox counters, hence to pseudorandomly independent scores.
When several scores are needed per (seed, token) pair (\eg the three coordinates for RotationFirst), the seed is replaced by $2^6 a_t + d$, reserving the low $6$ bits for the index $d$.

\emph{SplitMix64 PRF:} reuses the SplitMix64 finalizer $m$ from the \emph{SplitMix64 hash}, with the domain-separation constants $c_1 = \texttt{0xa4093822299f31d0}$ and $c_2 = \texttt{0x082efa98ec4e6c89}$:
\begin{equation*}
    s = m\bigl(a_t \oplus m(\xi \oplus c_1)\bigr), \qquad s' = m\bigl(s \oplus m(v \oplus c_2)\bigr), \qquad U_{a_t,v} = \bigl(\lfloor s'/2^{12}\rfloor + \tfrac12\bigr)\,2^{-52},
\end{equation*}
where $52$ bits fill the float64 mantissa, and the midpoint $+\tfrac12$ keeps $U_{a_t,v}\in(0,1)$ so that $\Phi^{-1}(U_{a_t,v})$ is always finite.
Independent \texttt{Channels} $c$ use the key $\xi \oplus (c\cdot\texttt{0x9e3779b97f4a7c15})$.

\section{ConcordMark and RotationFirst}
\label{app:selected_schemes}
\raggedbottom

We describe the schemes used in~\cref{sec:eval:best}.
For low and high variants, the generation is shared and only the detector differs.

\subsection{ConcordMark}
\label{app:concordmark}

\paragraph{Key components}
We describe precisely the ConcordMark generation algorithm in~\cref{alg:concordmark_generation}, and both detector variants in~\cref{alg:concordmark_low_detection,alg:concordmark_high_detection}.
ConcordMark uses \texttt{Lexical group} only for context seeding.
Let $\mathcal G: \Sigma \rightarrow \mathbb{N}$ be the corresponding lexical map.
Given a word, its spelling is first normalized and then gets attributed a synonym class of at most eight members (using WordNet).
It then uses \texttt{Ladder} context seeding starting at width $1$, or width $2$ when the last mapped token ID is below $200$.
For the hash function, it uses the \emph{rolling polynomial hash}~(\cref{app:hash_function}).

For the watermark score, it uses both \texttt{Prelude} with a cap of $14$ nats of entropy or $24$ tokens and $6$ independent \texttt{Channels}. 
It uses the \texttt{Gumbel race} as a logits transformation.
For detection, the low variant uses the \texttt{Excess} Gamma detection and the high variant the weighted \texttt{Position evidence}.

\begin{algorithm}[H]
    \caption{ConcordMark generation}
    \label{alg:concordmark_generation}
    \footnotesize
    \begin{algorithmic}[1]
    \Require Keys $\xi_0,\ldots,\xi_5$, model distributions $p_t$, lexical map $\mathcal G$, completion length $T$.
    \State Sample $J\sim\mathcal U\{0,\ldots,5\}$ \Comment{Use channel $J$, keyed by $\xi_J$, throughout}
    \State $\omega\gets()$; $B\gets0$; $\mathcal U\gets\emptyset$
    \For{$t=1,\ldots,T$}
        \If{$t\le24$ and $B<14$}
            \State Sample $\omega_t\sim p_t$ using private randomness
            \State $B\gets B-\log\sum_v p_t(v)^2$ \Comment{Prelude collision entropy}
        \Else
            \State $z_{1:t-1}\gets(\mathcal G(\omega_1),\ldots,\mathcal G(\omega_{t-1}))$; $m\gets\min(4,t-1)$
            \State $\ell\gets1$; if $z_{t-1}<200$, set $\ell\gets\min(2,m)$
            \For{$w=\ell,\ldots,m$}
                \State $W\gets z_{t-w:t-1}$ \Comment{Suffix of width $w$}
                \State $n_t\gets\left|\{i:w\le i<t-1,\ z_{i-w+1:i}=W\}\right|$ \Comment{Count earlier occurrences}
                \State If $n_t<4$ or $w=m$, set $c_t\gets(w,W,n_t)$, $a_t\gets h(c_t)$ and \textbf{break} \Comment{Context and seed}
            \EndFor
            \If{$n_t\ge8192$ or $a_t\in\mathcal U$}
                \State Sample $\omega_t\sim p_t$ using private randomness
            \Else
                \State $\mathcal U\gets\mathcal U\cup\{a_t\}$
                \State $\tilde U_{t,v}\gets U_{J,a_t,v}$ for each $v$ with $p_t(v)>0$
                \State $\omega_t\gets\arg\min_{v:p_t(v)>0}\{-\log\tilde U_{t,v}/p_t(v)\}$ \Comment{Token-level exponential race}
            \EndIf
        \EndIf
        \State Append $\omega_t$ to $\omega$
    \EndFor
    \State \Return $\omega$
    \end{algorithmic}
\end{algorithm}

\begin{algorithm}[H]
    \caption{ConcordMark (low) detection}
    \label{alg:concordmark_low_detection}
    \footnotesize
    \begin{algorithmic}[1]
    \Require Keys $\xi_0,\ldots,\xi_5$, completion $\omega$, lexical map $\mathcal G$.
    \State $D\gets\emptyset$; $\mathcal U\gets\emptyset$
    \For{$t=2,\ldots,|\omega|$}
        \State Reconstruct the ladder seed $a_t$ and count $n_t$ from $\mathcal G(\omega_{<t})$ as in~\cref{alg:concordmark_generation}
        \State If $n_t<8192$ and $(a_t,\omega_t)\notin\mathcal U$, set $D\gets D\cup\{t\}$ and $\mathcal U\gets\mathcal U\cup\{(a_t,\omega_t)\}$
    \EndFor
    \State If $D=\emptyset$, \Return $1$
    \State $\mathcal H(\omega)\gets\{a_t:t\in D\}$; $H\gets|\mathcal H(\omega)|$
    \State $D_a(\omega)\gets\{\omega_t:t\in D,\ a_t=a\}$; $m_a\gets|D_a(\omega)|$ for each $a\in\mathcal H(\omega)$
    \For{$j=0,\ldots,5$}
        \For{$a\in\mathcal H(\omega)$}
            \State $S_{j,a}(\omega)\gets\sum_{v\in D_a(\omega)}-\log(1-U_{j,a,v})$
            \State $Z_{j,a}\gets-\log Q(m_a,S_{j,a}(\omega))$ \Comment{Calibrate each context with its Gamma tail}
        \EndFor
        \State $S_j(\omega)\gets\sum_{a\in\mathcal H(\omega)}\max\{Z_{j,a}-\log2,0\}$ \Comment{Keep evidence above the median}
        \State $P_j\gets1$
        \If{$S_j(\omega)>H/2$}
            \State $\theta\gets(3-\sqrt{1+4H/S_j(\omega)})/2$
            \State $P_j\gets\exp\{H\log[(1-\theta/2)/(1-\theta)]-\theta S_j(\omega)\}$ \Comment{Chernoff bound}
            \State If $P_j<0.1$, set $P_j\gets\sum_{b=1}^{H}\binom Hb2^{-H}Q(b,S_j(\omega))$ \Comment{Exact excess tail $R_H$}
        \EndIf
    \EndFor
    \State \Return $P(\omega)=1-(1-\min_jP_j)^6$ \Comment{\v{S}id\'ak correction over channels}
    \end{algorithmic}
\end{algorithm}

\begin{algorithm}[H]
    \caption{ConcordMark (high) detection}
    \label{alg:concordmark_high_detection}
    \footnotesize
    \begin{algorithmic}[1]
    \Require Keys $\xi_0,\ldots,\xi_5$, completion $\omega$, lexical map $\mathcal G$, auxiliary model distributions $\hat p_t$.
    \State Reconstruct $a_t$ and the retained positions $D$ as in~\cref{alg:concordmark_low_detection}
    \State If $D=\emptyset$, \Return $1$
    \For{$t\in D$}
        \State $\eta_t\gets1+\frac14\operatorname{round}\!\left(4\min\{\max\{-\log\hat p_t(\omega_t),0\},2.5\}\right)$
        \State $C_t\gets\{\text{eight most probable tokens under }\hat p_t\}\setminus\{\omega_t\}$
        \State $\alpha_{t,v}\gets\min\{10^4,\hat p_t(v)/\hat p_t(\omega_t)\}$ for each $v\in C_t$
        \State $C_t\gets\{v\in C_t:\alpha_{t,v}\ge10^{-3}\}$; $\lambda_t\gets\sum_{v\in C_t}\alpha_{t,v}$
        \State If $\lambda_t>8$, set $C_t\gets\emptyset$ and $\lambda_t\gets0$ \Comment{Discard an unreliable reconstructed race}
    \EndFor
    \For{$j=0,\ldots,5$}
        \For{$t\in D$}
            \State $W_{j,t}\gets\mathbf1\{U_{j,a_t,v}<U_{j,a_t,\omega_t}^{\alpha_{t,v}}\text{ for all }v\in C_t\}$
            \If{$W_{j,t}=1$} \Comment{The observed token wins the reconstructed race}
                \State $R_{j,t}\gets(1-U_{j,a_t,\omega_t}^{1+\lambda_t})/(1+\lambda_t)$
            \Else
                \State $R_{j,t}\gets1-U_{j,a_t,\omega_t}+U_{j,a_t,\omega_t}^{1+\lambda_t}/(1+\lambda_t)$
            \EndIf
        \EndFor
        \State $S_j(\omega)\gets\sum_{t\in D}-\eta_t\log R_{j,t}$ \Comment{Weighted position evidence}
    \EndFor
    \State $s\gets\max_j S_j(\omega)$
    \If{all weights equal some $\eta$}
        \State $q\gets Q(|D|,s/\eta)$ \Comment{Exact Gamma upper tail}
    \Else
        \State $q\gets$ Lugannani--Rice approximation to $\Pr[\sum_{t\in D}\eta_t E_t\ge s]$, for independent $E_t\sim\operatorname{Exp}(1)$
    \EndIf
    \State \Return $P(\omega)=\min\{1,1.02[1-(1-q)^6]\}$ \Comment{Channel correction and tail safety factor}
    \end{algorithmic}
\end{algorithm}

\subsection{RotationFirst}
\label{app:rotationfirst}

\paragraph{Key components}
We describe precisely the RotationFirst generation algorithm and its detector in~\cref{alg:rotationfirst_generation,alg:rotationfirst_detection}.
RotationFirst uses \texttt{Lexical group} throughout the scheme.
Let $\mathcal G: \Sigma \rightarrow \mathbb{N}$ be the corresponding lexical map.
Each token is decoded, stripped of surrounding whitespace, and lowercased.
If the resulting word appears in WordNet, it is assigned to the synonym class given by its first WordNet synset. 
Let $\mathcal{E}$ be the set of the resulting synonym classes.
Tokens sharing this synset are grouped together, whereas all other tokens remain in singleton classes.
It then uses \texttt{Global first use} context seeding with the \emph{rotation first hash}~(\cref{app:hash_function}).
For a class $c$, it uses address $(0,c)$ if it has not been used before and $c\in\mathcal E$, and otherwise uses $(\mathcal G(\omega_{t-1})+1,c)$.

For the watermark score, it uses a \emph{private geometry} approach, namely the \texttt{Rotation} spherical geometry~(\texttt{Spherical}).
For each request, sample randomly a uniform quaternion $A$.
Then, at each step draw a uniform keyed $Z_{a_t,u}\in S^3$ and derive the score:
\begin{equation*}
    \tilde{U}_{a_t,u} = F(\lvert A\cdot Z_{a_t,u}\rvert),
\end{equation*}
where $F(r)=\frac2\pi\bigl(\arcsin r+r\sqrt{1-r^2}\bigr)$.
Then for detection, it uses the corresponding \texttt{Predictive Gamma} detector.

\begin{algorithm}[H]
    \caption{RotationFirst generation}
    \label{alg:rotationfirst_generation}
    \footnotesize
    \begin{algorithmic}[1]
    \Require Key $\xi$, model distributions $p_t$, grouping $\mathcal G$, eligible classes $\mathcal E$, completion length $T$.
    \State Sample $A\sim\mathcal U(\mathbb S^3)$ \Comment{Private rotation shared throughout the completion}
    \State $F(r)\gets\frac2\pi(\arcsin r+r\sqrt{1-r^2})$
    \State $\omega\gets()$; $\mathcal A\gets\emptyset$; $\mathcal U\gets\emptyset$ \Comment{Emitted classes and initialized clocks}
    \For{$t=1,\ldots,T$}
        \State $M_t(u)\gets\sum_{v:\mathcal G(v)=u}p_t(v)$ for each class $u$
        \State $b_t\gets\mathcal G(\omega_{t-1})+1$ if $t>1$, otherwise $b_t\gets0$
        \For{each class $u$ with $M_t(u)>0$}
            \State $c_t(u)\gets0$ if $u\in\mathcal E\setminus\mathcal A$, otherwise $c_t(u)\gets b_t$; $a_t(u)\gets h(c_t(u))$
            \If{$(a_t(u),u)\notin\mathcal U$}
                \If{$b_t\ne0$ or $u\in\mathcal E$}
                    \State Recover $Z_{a_t(u),u}\in\mathbb S^3$ using key $\xi$
                    \State $\tilde U_{a_t(u),u}\gets F(|A\cdot Z_{a_t(u),u}|)$
                    \State $R_{a_t(u),u}\gets-\log\tilde U_{a_t(u),u}$
                \Else
                    \State Sample $R_{a_t(u),u}\sim\operatorname{Exp}(1)$ privately
                \EndIf
                \State $\mathcal U\gets\mathcal U\cup\{(a_t(u),u)\}$
            \EndIf
        \EndFor
        \State $u_*\gets\arg\min_{u:M_t(u)>0}R_{a_t(u),u}/M_t(u)$; $\tau_t\gets R_{a_t(u_*),u_*}/M_t(u_*)$
        \State $R_{a_t(u),u}\gets\max\{0,R_{a_t(u),u}-\tau_tM_t(u)\}$ for each $u$ with $M_t(u)>0$
        \State Sample $R_{a_t(u_*),u_*}\sim\operatorname{Exp}(1)$ privately \Comment{Refresh the winning clock}
        \State Sample $\omega_t$ from $\{v:\mathcal G(v)=u_*\}$ with probabilities $p_t(v)/M_t(u_*)$
        \State $\mathcal A\gets\mathcal A\cup\{u_*\}$; append $\omega_t$ to $\omega$
    \EndFor
    \State \Return $\omega$
    \end{algorithmic}
\end{algorithm}

\begin{algorithm}[H]
    \caption{RotationFirst detection}
    \label{alg:rotationfirst_detection}
    \footnotesize
    \begin{algorithmic}[1]
    \Require Key $\xi$, completion $\omega$, grouping $\mathcal G$, eligible classes $\mathcal E$, public unit quaternions $A_0,\ldots,A_{511}$.
    \State $D\gets()$; $\mathcal A\gets\emptyset$; $\mathcal U\gets\emptyset$
    \For{$t=1,\ldots,|\omega|$}
        \State $u_t\gets\mathcal G(\omega_t)$
        \State $b_t\gets\mathcal G(\omega_{t-1})+1$ if $t>1$, otherwise $b_t\gets0$
        \State $c_t\gets0$ if $u_t\in\mathcal E\setminus\mathcal A$, otherwise $c_t\gets b_t$; $a_t\gets h(c_t)$
        \State If $(t>1\text{ or }u_t\in\mathcal E)$ and $(a_t,u_t)\notin\mathcal U$, append $(a_t,u_t)$ to $D$ and set $\mathcal U\gets\mathcal U\cup\{(a_t,u_t)\}$
        \State $\mathcal A\gets\mathcal A\cup\{u_t\}$
    \EndFor
    \State $N\gets|D|$; if $N\le1$, \Return $1$
    \State Enumerate $D$ as $(a_i,u_i)_{i=1}^{N}$ in text order
    \State $F(r)\gets\frac2\pi(\arcsin r+r\sqrt{1-r^2})$; $d(B,Z)\gets1-F(|B\cdot Z|)$
    \State $B_1\gets A_0$; $w_{1,j}\gets1/512$ for $j=0,\ldots,511$; $S(\omega)\gets0$
    \For{$i=1,\ldots,N$}
        \State Recover $Z_i=Z_{a_i,u_i}\in\mathbb S^3$ using key $\xi$
        \State $S(\omega)\gets S(\omega)-\log d(B_i,Z_i)$ \Comment{Score before updating the predicted rotation}
        \State $w_{i+1,j}\gets w_{i,j}\,d(A_j,Z_i)^{-1/2}$ for $j=0,\ldots,511$
        \State Normalize $w_{i+1,j}\gets w_{i+1,j}/\sum_{k=0}^{511}w_{i+1,k}$
        \State $B_{i+1}\gets\operatorname{eig}_{\max}(\sum_{j=0}^{511}w_{i+1,j}A_jA_j^\top)$ \Comment{Unit eigenvector for the largest eigenvalue}
    \EndFor
    \State \Return $P(\omega)=Q(N,S(\omega))$ \Comment{Predictive Gamma upper tail}
    \end{algorithmic}
\end{algorithm}

\subsection{Additional Evaluation}

Here, we complete the evaluation from~\cref{sec:eval:best} with additional plots~(\cref{fig:selected-roc}).
Namely, we plot the TPR as a function of the significance level $\alpha$, from $10^{-6}$ to $10^{-1}$, for clean detection and paraphrasing.
At each $\alpha$, the TPR is the fraction of samples with p-value below $\alpha$. 
As in the main evaluation, we report the minimum TPR across the five watermarking seeds.

\begin{figure}[H]
    \centering
    \includegraphics[width=\linewidth]{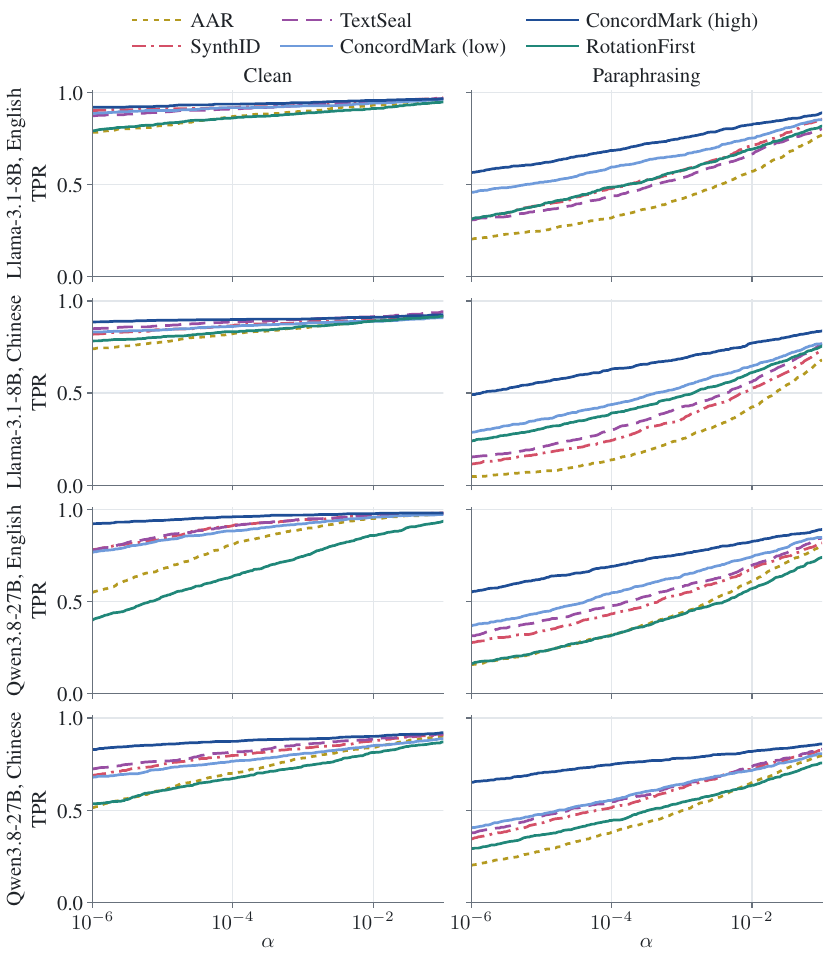}
    \caption{TPR as a function of the significance level $\alpha$, with 1k samples per model and language generated from WildChat prompts. Left is clean detection and right after paraphrasing.}
    \label{fig:selected-roc}
\end{figure}

\paragraph{Results}
Across models, languages, and significance levels, these curves confirm the conclusions of~\cref{sec:eval:best}.
ConcordMark (high) provides the strongest detection, and especially after paraphrasing.
ConcordMark (low) also retains a strong paraphrasing robustness while being more efficient at detection, as it does not require a forward pass.
RotationFirst is less robust than ConcordMark, as expected as it focuses mainly on quality, yet remains comparable to the baselines.
\flushbottom

\section{Proofs}
\label{app:proofs}

In this section, we prove the theoretical claims made throughout the paper.
In~\cref{app:proofs:prelim}, we state our assumptions and two lemmas used in the rest of the section.
In~\cref{app:proofs:given_key}, we prove that a detector that is sound over the key can be turned into a detector that is sound given a key~(\cref{sec:method:definition_validity}).
In~\cref{app:proofs:tests}, we prove that the three statistical tests of~\cref{sec:method:verification} are valid.
In~\cref{app:proofs:residual}, we prove that \texttt{ResidualRace} is distortion-free~(\cref{app:components:transform}).
Lastly, in~\cref{app:proofs:detectors}, we derive the null distribution of the \texttt{Predictive Gamma} and \texttt{Position evidence} detectors~(\cref{app:components:detection}).
For all other mathematical results, we consider the proof sufficiently trivial and do not derive them here.

\subsection{Preliminaries}
\label{app:proofs:prelim}

\paragraph{Setup}
Let $\Xi$ be the key space, and let $\xi$ be drawn from $\Xi$ as in~\cref{eq:soundness_over_key}.
For every text $\omega \in \Sigma^*$, the detector returns a p-value $P_\xi(\omega) \in [0,1]$, and we assume that $(\xi,\omega) \mapsto P_\xi(\omega)$ is measurable (which holds trivially for a finite key space).
Human texts are written without knowledge of the private key: for any human text distribution $\mathcal{H}_0$, the text $\Omega \sim \mathcal{H}_0$ is independent of $\xi$.
As in~\cref{sec:method:preliminaries}, we model the PRNG as ideal: the seed $a_t$ is a function of the context only, and, over a uniformly drawn key, the scores of distinct (seed, unit) pairs are independent, with $U_{a,u} \sim \mathcal{U}([0,1])$ (and $Z_{a,u}$ distributed according to its geometry for private-geometry scores).
We also assume that the hash function has no collision~(\cref{app:hash_function}).
Finally, we write $T(s,p,c) \mathrel{:=} \Pr_{Z \sim \mathrm{Binomial}(s,p)}[Z \ge c]$ for the binomial upper tail, as in~\cref{alg:soundness_given_key_test}.

\begin{lemma}
\label{lem:tail_pvalue}
Let $X$ be a real random variable and $\varphi: \mathbb{R} \rightarrow [0,\infty)$ a continuous nonincreasing function such that $\Pr[X \ge s] \le \varphi(s)$ for all $s \in \mathbb{R}$.
Then $\min\{1,\varphi(X)\}$ is a valid p-value, \ie $\Pr[\min\{1,\varphi(X)\} \le x] \le x$ for all $x \in [0,1]$.
\end{lemma}

\begin{proof}
The case $x = 1$ is trivial, and the case $x = 0$ follows from the case $x > 0$, since $\Pr[\min\{1,\varphi(X)\} \le 0] \le \Pr[\min\{1,\varphi(X)\} \le x]$ for every $x > 0$.
Let $x \in (0,1)$ and $I \mathrel{:=} \{s \in \mathbb{R} : \varphi(s) \le x\}$.
If $I = \emptyset$, the event is empty.
Otherwise, $I$ is an interval unbounded above, and it is bounded below, since otherwise $\Pr[X \ge s] \le x < 1$ for all $s$, which contradicts $\Pr[X \ge s] \rightarrow 1$ as $s \rightarrow -\infty$.
Let $s^* \mathrel{:=} \inf I$. By continuity, $\varphi(s^*) \le x$, and since $\{\varphi(X) \le x\} \subseteq \{X \ge s^*\}$,
\begin{equation*}
	\Pr[\min\{1,\varphi(X)\} \le x] \le \Pr[X \ge s^*] \le \varphi(s^*) \le x. \qedhere
\end{equation*}
\end{proof}

\begin{lemma}
\label{lem:binomial_pvalue}
Let $X$ be a sum of $s$ independent Bernoulli random variables with parameters at most $p \in [0,1]$.
Then $\Pr[X \ge c] \le T(s,p,c)$ for all $c \in \mathbb{N}$, and $\Pr[T(s,p,X) \le \eta] \le \eta$ for all $\eta \in [0,1]$.
\end{lemma}

\begin{proof}
Let $p_1,\dots,p_s \le p$ be the Bernoulli parameters and $V_1,\dots,V_s$ i.i.d. uniform random variables on $[0,1]$.
Then $X$ has the same distribution as $\sum_{i} \indicator\{V_i \le p_i\} \le \sum_{i} \indicator\{V_i \le p\} \sim \mathrm{Binomial}(s,p)$, which gives the first claim.
For the second claim, since $c \mapsto T(s,p,c)$ is nonincreasing and $T(s,p,s+1) = 0$, let $c^* \mathrel{:=} \min\{c \in \{0,\dots,s+1\} : T(s,p,c) \le \eta\}$.
Then $\{T(s,p,X) \le \eta\} = \{X \ge c^*\}$, and by the first claim, $\Pr[X \ge c^*] \le T(s,p,c^*) \le \eta$.
\end{proof}

\subsection{Soundness Given a Key}
\label{app:proofs:given_key}

We prove the claim of~\cref{sec:method:definition_validity}: multiplying by $1/\beta$ the p-value of a detector that is sound over the key~(\cref{eq:soundness_over_key}) gives a detector that is sound given a key~(\cref{eq:soundness_given_key}).

\begin{proposition}
\label{prop:soundness_given_key}
Let $\beta \in (0,1]$ and let $P_\xi$ be a detector that satisfies~\cref{eq:soundness_over_key}.
Define the calibrated detector
\begin{equation}
	\label{eq:calibrated_pvalue}
	P^\beta_\xi(\omega) \mathrel{:=} \min\left\{1, \frac{P_\xi(\omega)}{\beta}\right\}.
\end{equation}
Then $P^\beta_\xi$ satisfies~\cref{eq:soundness_over_key}, and for every human text distribution $\mathcal{H}_0$ and every $\alpha \in [0,1]$,
\begin{equation*}
	\Pr_{\xi}\left[ \Pr_{\Omega \sim \mathcal{H}_0} [P^\beta_{\xi}(\Omega) \le \alpha] > \alpha\right] \le \beta,
\end{equation*}
\ie $P^\beta_\xi$ also satisfies~\cref{eq:soundness_given_key}.
\end{proposition}

\begin{proof}
We first relate the rejection regions of both detectors.
For $\alpha \in [0,1)$, since $\min\{1,x\} \le \alpha < 1$ holds if and only if $x \le \alpha$, we have for every key $\xi$ and text $\omega$
\begin{equation}
	\label{eq:calibrated_region}
	P^\beta_\xi(\omega) \le \alpha \iff P_\xi(\omega) \le \alpha\beta.
\end{equation}
For $\alpha = 1$, both properties hold trivially, as all probabilities are bounded by $1$.
We therefore assume $\alpha \in [0,1)$ in the rest of the proof.

\emph{Soundness over the key.}
Let $\omega \in \Sigma^*$.
By~\cref{eq:calibrated_region} and~\cref{eq:soundness_over_key} applied at the threshold $\alpha\beta \in [0,1]$,
\begin{equation*}
	\Pr_{\xi}[P^\beta_\xi(\omega) \le \alpha] = \Pr_{\xi}[P_\xi(\omega) \le \alpha\beta] \le \alpha\beta \le \alpha.
\end{equation*}

\emph{Soundness given a key.}
Let $\mathcal{H}_0$ be a human text distribution, and define for each key the false positive rate
\begin{equation*}
	g(\xi) \mathrel{:=} \Pr_{\Omega \sim \mathcal{H}_0}[P^\beta_\xi(\Omega) \le \alpha] = \Pr_{\Omega \sim \mathcal{H}_0}[P_\xi(\Omega) \le \alpha\beta],
\end{equation*}
where the equality follows from~\cref{eq:calibrated_region}.
Because $\Omega$ is independent of $\xi$, Fubini's theorem allows us to swap the two expectations, and applying~\cref{eq:soundness_over_key} to each fixed text gives
\begin{equation*}
	\mathbb{E}_{\xi}[g(\xi)]
	= \Pr_{\xi,\Omega}[P_\xi(\Omega) \le \alpha\beta]
	= \mathbb{E}_{\Omega \sim \mathcal{H}_0}\left[\Pr_{\xi}[P_\xi(\Omega) \le \alpha\beta]\right]
	\le \alpha\beta.
\end{equation*}
If $\alpha = 0$, then $g \ge 0$ and $\mathbb{E}_\xi[g(\xi)] \le 0$, so $g(\xi) = 0$ almost surely and $\Pr_\xi[g(\xi) > 0] = 0 \le \beta$.
Otherwise, $\alpha > 0$ and Markov's inequality gives
\begin{equation*}
	\Pr_{\xi}[g(\xi) > \alpha] \le \frac{\mathbb{E}_{\xi}[g(\xi)]}{\alpha} \le \frac{\alpha\beta}{\alpha} = \beta,
\end{equation*}
which concludes the proof.
\end{proof}

\subsection{Validity of the Verification Tests}
\label{app:proofs:tests}

We prove that each of the three tests of~\cref{sec:method:verification} wrongly rejects a valid scheme with probability at most its significance level $\delta$.

\begin{proposition}
\label{prop:distortion_free_test}
Let $f$ be a distortion-free logits transformation~(\cref{eq:distortion_free}).
For each context $t$, let $v_1,\dots,v_k$ be the $k$ most likely tokens under $p_t$ and let
\begin{equation*}
	G_t(q) \mathrel{:=} \Bigl(q(v_1),\dots,q(v_k),1-\sum_{r=1}^{k}q(v_r)\Bigr)
\end{equation*}
map a distribution $q \in \Delta(\Sigma)$ to a probability vector over the head tokens and the pooled remaining tokens.
For $m$ independent keys, let $d_j \mathrel{:=} G_t(f(\zeta_t^{(j)},p_t)) - G_t(p_t)$, $\bar{d}$ their mean, $h \mathrel{:=} \lfloor m/2 \rfloor$, $u \mathrel{:=} \mathrm{sign}(h^{-1}\sum_{j=1}^{h}d_j)$ and $b \mathrel{:=} \max\{0,(m-h)^{-1}\sum_{j=h+1}^{m}u^\top d_j\}$.
\Cref{alg:distortion_free_test} combines the p-values
\begin{equation*}
	p_{\mathrm{coord}} \mathrel{:=} \min\{1,2(k+1)\exp(-2m\lVert\bar{d}\rVert_\infty^2)\}
	\quad\text{and}\quad
	p_{\mathrm{split}} \mathrel{:=} \exp(-(m-h)b^2/2)
\end{equation*}
into $q_t \mathrel{:=} \min\{1,2\min(p_{\mathrm{coord}},p_{\mathrm{split}})\}$, and rejects $f$ if $Q \mathrel{:=} \min\{1,N\min_t q_t\} \le \delta$.
Then~\cref{alg:distortion_free_test} rejects $f$ with probability at most $\delta$.
\end{proposition}

\begin{proof}
Fix a context $t$.
Because the $m$ keys are drawn independently, the scores $\zeta_t^{(1)},\dots,\zeta_t^{(m)}$ are i.i.d., and by~\cref{eq:distortion_free}, the differences $d_1,\dots,d_m$ are i.i.d. with $\mathbb{E}[d_j] = G_t(\mathbb{E}[f(\zeta_t^{(j)},p_t)]) - G_t(p_t) = 0$, as $G_t$ is affine.

\emph{Coordinate test.}
Each coordinate of $G_t(q)$ lies in $[0,1]$ for every $q \in \Delta(\Sigma)$, so each coordinate of $d_j$ lies in an interval of length $1$.
By Hoeffding's inequality and a union bound over the $k+1$ coordinates, $\Pr[\lVert\bar{d}\rVert_\infty \ge s] \le 2(k+1)\exp(-2ms^2)$ for all $s \ge 0$.
Hence, by~\cref{lem:tail_pvalue} with $\varphi(s) = 2(k+1)\exp(-2ms^2)$, $p_{\mathrm{coord}}$ is a valid p-value.

\emph{Split test.}
The direction $u$ only depends on $d_1,\dots,d_h$, so conditional on $u$, the variables $Y_j \mathrel{:=} u^\top d_j$ for $j > h$ are i.i.d. with mean zero.
Since $G_t(q)$ and $G_t(p_t)$ are probability vectors and $u \in [-1,1]^{k+1}$, $u^\top G_t(q) \in [-1,1]$ and hence $Y_j$ lies in an interval of length $2$. 
By Hoeffding's inequality, $\Pr[b \ge s \mid u] \le \exp(-(m-h)s^2/2)$ for all $s > 0$, where $m-h \ge 1$ as $m \ge 2$.
Hence, by~\cref{lem:tail_pvalue} with $\varphi(s) = \exp(-(m-h)s^2/2)$, $p_{\mathrm{split}}$ is a valid p-value conditional on $u$, and thus unconditionally.

\emph{Combination.}
By a union bound over both tests, $\Pr[q_t \le x] \le \Pr[p_{\mathrm{coord}} \le x/2] + \Pr[p_{\mathrm{split}} \le x/2] \le x$.
By a union bound over the $N$ contexts, $\Pr[Q \le \delta] \le \sum_{t=1}^{N}\Pr[q_t \le \delta/N] \le \delta$.
\end{proof}

\begin{proposition}
\label{prop:soundness_over_key_test}
Let $P_\xi$ be a detector that is sound over the key~(\cref{eq:soundness_over_key}), \ie
\begin{equation*}
	\forall \omega \in \Sigma^*, \forall \alpha \in [0,1], \quad \Pr_{\xi}[P_{\xi}(\omega) \le \alpha] \le \alpha.
\end{equation*}
Given a corpus $C = (\omega_1,\dots,\omega_n)$, thresholds $\mathcal{A}$ and $m$ independent keys $\xi_1^i,\dots,\xi_m^i$ per text, \cref{alg:soundness_over_key_test} counts the false positives $X_i(\alpha) \mathrel{:=} \sum_{j=1}^{m}\indicator\{P_{\xi_j^i}(\omega_i) \le \alpha\}$, and rejects $P_\xi$ if $q_i(\alpha) \mathrel{:=} T(m,\alpha,X_i(\alpha)) \le \eta \mathrel{:=} (1-(1-\delta)^{1/n})/|\mathcal{A}|$ for some text $\omega_i$ and threshold $\alpha \in \mathcal{A}$.
Then, for every corpus $C$,~\cref{alg:soundness_over_key_test} rejects $P_\xi$ with probability at most $\delta$.
\end{proposition}

\begin{proof}
Fix a text $\omega_i$ and a threshold $\alpha \in \mathcal{A}$.
Because the keys $\xi_1^i,\dots,\xi_m^i$ are i.i.d., $X_i(\alpha)$ is a sum of $m$ i.i.d. Bernoulli random variables with parameter $\Pr_\xi[P_\xi(\omega_i) \le \alpha] \le \alpha$ by soundness over the key.
Hence, by~\cref{lem:binomial_pvalue}, $\Pr[q_i(\alpha) \le \eta] \le \eta$, and by a union bound over the thresholds, the event $E_i \mathrel{:=} \{\exists \alpha \in \mathcal{A} : q_i(\alpha) \le \eta\}$ has probability at most $|\mathcal{A}|\eta = 1-(1-\delta)^{1/n}$.
The events $E_1,\dots,E_n$ depend on disjoint sets of independent keys, so they are independent, and
\begin{equation*}
	\Pr[\mathcal{R} \neq \emptyset] = 1 - \prod_{i=1}^{n}(1 - \Pr[E_i]) \le 1 - (1-\delta) = \delta. \qedhere
\end{equation*}
\end{proof}

\begin{proposition}
\label{prop:soundness_given_key_test}
Let $P_\xi$ be a detector that is sound given a key~(\cref{eq:soundness_given_key}) for the human text distribution $\mathcal{H}_0$, the fraction $\beta$, and the thresholds $\mathcal{A}$, \ie
\begin{equation*}
	\forall \alpha \in \mathcal{A}, \quad \Pr_{\xi}[ \Pr_{\Omega \sim \mathcal{H}_0} [P_{\xi}(\Omega) \le \alpha] > \alpha] \le \beta.
\end{equation*}
\Cref{alg:soundness_given_key_test} draws $d$ independent keys $\xi_j$ and, for each key, $n$ independent texts $\omega_1^j,\dots,\omega_n^j \sim \mathcal{H}_0$, and counts the false positives $X_j(\alpha) \mathrel{:=} \sum_{t=1}^{n}\indicator\{P_{\xi_j}(\omega_t^j) \le \alpha\}$.
For each threshold $\alpha \in \mathcal{A}$ and screening level $\gamma$ in a predeclared finite set $\mathcal{S} \subset (0,1)$, a key screens positive if $X_j(\alpha) \ge c$, where $c \mathrel{:=} \min\{k \in \{1,\dots,n+1\} : T(n,\alpha,k) \le \gamma\}$ is the smallest count that is significant at level $\gamma$ for a key with false positive rate $\alpha$.
Let $r \mathrel{:=} T(n,\alpha,c)$, $b \mathrel{:=} \beta + (1-\beta)r$, and $B$ the number of keys that screen positive.
\Cref{alg:soundness_given_key_test} rejects $P_\xi$ if $q \mathrel{:=} T(d,b,B) \le \eta \mathrel{:=} \delta/(|\mathcal{A}||\mathcal{S}|)$ for some $(\alpha,\gamma) \in \mathcal{A} \times \mathcal{S}$.
Then~\cref{alg:soundness_given_key_test} rejects $P_\xi$ with probability at most $\delta$.
\end{proposition}

\begin{proof}
Fix $(\alpha,\gamma) \in \mathcal{A} \times \mathcal{S}$, and note that $c$, $r$ and $b$ are deterministic.
For a key $\xi$, let $g(\xi) \mathrel{:=} \Pr_{\Omega \sim \mathcal{H}_0}[P_\xi(\Omega) \le \alpha]$ be its false positive rate, and say that $\xi$ is bad if $g(\xi) > \alpha$.
Given $\xi_j$, the texts $\omega_1^j,\dots,\omega_n^j$ are i.i.d., so $X_j(\alpha) \sim \mathrm{Binomial}(n, g(\xi_j))$.
If $\xi_j$ is not bad, $g(\xi_j) \le \alpha$, and~\cref{lem:binomial_pvalue} gives $\Pr[X_j(\alpha) \ge c \mid \xi_j] \le T(n,\alpha,c) = r$.
As a key is bad with probability $\pi \le \beta$ by soundness given a key,
\begin{equation*}
	\Pr[X_j(\alpha) \ge c] \le \pi + (1-\pi)r = r + \pi(1-r) \le \beta + (1-\beta)r = b.
\end{equation*}
Since all draws are independent across keys, $B$ is a sum of $d$ independent Bernoulli random variables with parameters at most $b$, and by~\cref{lem:binomial_pvalue}, $\Pr[q \le \eta] \le \eta$.
By a union bound over $\mathcal{A} \times \mathcal{S}$, $\Pr[\mathcal{R} \neq \emptyset] \le |\mathcal{A}||\mathcal{S}|\eta = \delta$.
\end{proof}

\subsection{Distortion-Freeness of Residual Clocks}
\label{app:proofs:residual}

We prove that \texttt{ResidualRace}~(\cref{eq:exponential_race}) is distortion-free even when a context is repeated within a request, whereas the \texttt{Gumbel race} requires the per-request cache for that~(\cref{app:components:context}).
The proof relies on the following classical property of exponential races.

\begin{lemma}
\label{lem:exponential_race}
Let $E_1,\dots,E_K$ be i.i.d. $\operatorname{Exp}(1)$ random variables and $m_1,\dots,m_K > 0$ with $\sum_k m_k = 1$.
Let $\tau \mathrel{:=} \min_k E_k/m_k$ and $\kappa \mathrel{:=} \arg\min_k E_k/m_k$.
Then $\Pr[\kappa = k] = m_k$, and conditional on $(\tau,\kappa)$, the residuals $(E_j - \tau m_j)_{j \neq \kappa}$ are i.i.d. $\operatorname{Exp}(1)$.
\end{lemma}

\begin{proof}
Since $E_k/m_k$ has density $m_k e^{-m_k s}$, for all $k$, $s \ge 0$ and $x_j \ge 0$,
\begin{equation*}
	\Pr[\kappa = k, \tau \in \mathrm{d}s, E_j - s m_j > x_j \ \forall j \neq k]
	= m_k e^{-m_k s}\,\mathrm{d}s \prod_{j \neq k} e^{-(s m_j + x_j)}
	= m_k \cdot e^{-s}\,\mathrm{d}s \cdot \prod_{j \neq k} e^{-x_j}.
\end{equation*}
The density factorizes, so $\kappa$, $\tau$ and the residuals are independent with the claimed distributions.
\end{proof}

\begin{proposition}
\label{prop:residual_race}
Assume that the initial scores $\tilde{U}_{a,u}$ of distinct (seed, unit) pairs are i.i.d. uniform.
Then, for \texttt{ResidualRace}, $\Pr[\omega_t = v \mid \omega_{<t}] = p_t(v)$ for every step $t$ and token $v \in \Sigma$, even if the seed $a_t$ has already been used in the request.
\end{proposition}

\begin{proof}
Let $\mathcal{F}_t$ be the information generated by $\omega_{\le t}$ and the race outcomes $(\tau_s, u_s)_{s \le t}$.
We show by induction the invariant: conditional on $\mathcal{F}_t$, the current clocks $R_{a,u}$ of all (seed, unit) pairs are i.i.d. $\operatorname{Exp}(1)$, where we identify a clock that is not yet initialized with its initial value $-\log\tilde{U}_{a,u}$.
At $t = 0$, this holds because $-\log\tilde{U}_{a,u} \sim \operatorname{Exp}(1)$.

Assume the invariant holds at $t-1$.
Given $\mathcal{F}_{t-1}$, the distribution $p_t$, the masses $M_t$ and the seed $a_t$ are fixed, and the race of~\cref{eq:exponential_race} is an exponential race over the i.i.d. $\operatorname{Exp}(1)$ clocks $R_{a_t,u}$ with weights $M_t(u)$.
By~\cref{lem:exponential_race}, $\Pr[u_t = u \mid \mathcal{F}_{t-1}] = M_t(u)$, and since the token is then sampled inside $u_t$ with private randomness,
\begin{equation*}
	\Pr[\omega_t = v \mid \mathcal{F}_{t-1}] = M_t(\mathcal{G}(v)) \cdot \frac{p_t(v)}{M_t(\mathcal{G}(v))} = p_t(v).
\end{equation*}
Because $p_t$ is a function of $\omega_{<t}$, taking the expectation given $\omega_{<t}$ proves the claim at step $t$.
For the invariant, by~\cref{lem:exponential_race}, conditional on $\mathcal{F}_{t-1}$ and $(\tau_t,u_t)$, the updated losing clocks $R_{a_t,u} - \tau_t M_t(u)$ are i.i.d. $\operatorname{Exp}(1)$.
The winning clock is replaced by a fresh $\operatorname{Exp}(1)$, and all other clocks (other seeds, or units with $M_t(u) = 0$) are unchanged and independent of the race.
The token $\omega_t$ only depends on independent private randomness, so the invariant holds at $t$.
\end{proof}

\subsection{Validity of the Detectors}
\label{app:proofs:detectors}

We derive the null distribution of the detectors of~\cref{app:components:detection} whose validity is not immediate.
Throughout, the text $\omega$ is fixed and the randomness is over the key only, so that an exact null distribution implies~\cref{eq:soundness_over_key}.

\begin{proposition}
\label{prop:predictive_gamma}
Assume that, for every fixed geometry $B$ and every score $Z$ drawn according to its geometry, $d(B,Z) \sim \mathcal{U}([0,1])$.
Then, for every text $\omega$, the \texttt{Predictive Gamma} statistic of~\cref{eq:predictive_gamma} satisfies $S(\omega) \sim \operatorname{Gamma}(N,1)$.
\end{proposition}

\begin{proof}
The retained pairs $D$ only depend on $\omega$, so $N$ is fixed and, by the ideal PRNG assumption, $Z_1,\dots,Z_N$ are i.i.d.
The estimate $B_i$ is computed from $A_0,\dots,A_{K-1}$ and $Z_1,\dots,Z_{i-1}$ only, and $Z_i$ is independent of $(Z_1,\dots,Z_{i-1})$.
Hence, conditional on $Z_1,\dots,Z_{i-1}$, $B_i$ is fixed and $d_i = d(B_i,Z_i) \sim \mathcal{U}([0,1])$.
Since $d_1,\dots,d_{i-1}$ are functions of $Z_1,\dots,Z_{i-1}$, $d_i$ is independent of $(d_1,\dots,d_{i-1})$, and by induction $d_1,\dots,d_N$ are i.i.d. uniform.
Thus, the $-\log d_i$ are i.i.d. $\operatorname{Exp}(1)$ and $S(\omega) \sim \operatorname{Gamma}(N,1)$.
\end{proof}

\begin{proposition}
\label{prop:position_evidence}
Assume that $\hat{p}_t$ only depends on $\omega_{<t}$.
Then, for every text $\omega$, the variables $(R_t)_{t \in D_{\mathrm{pos}}}$ of the \texttt{Position evidence} detector are i.i.d. uniform on $[0,1]$.
In particular, $S(\omega) \sim \operatorname{Gamma}(\lvert D_{\mathrm{pos}}\rvert,1)$, and the weighted statistic $\sum_{t \in D_{\mathrm{pos}}} -\eta_t\log R_t$ has the distribution of $\sum_{t \in D_{\mathrm{pos}}} \eta_t E_t$ for i.i.d. $E_t \sim \operatorname{Exp}(1)$.
\end{proposition}

\begin{proof}
Since $\hat{p}_t$ only depends on $\omega_{<t}$, the competitors $C_t$, the exponents $\alpha_{t,v}$, $\lambda_t$ and the weights $\eta_t$ are fixed.
The seeds $a_t$ for $t \in D_{\mathrm{pos}}$ are distinct and $\omega_t \notin C_t$, so all scores $U_{a_t,\omega_t}$ and $U_{a_t,v}$ for $v \in C_t$ are distinct (seed, token) pairs, and are therefore independent and uniform.
In particular, each $R_t$ is a function of its own disjoint set of scores, so the $R_t$ are independent, and it remains to show that each $R_t$ is uniform.

Fix $t$ and write $U \mathrel{:=} U_{a_t,\omega_t}$, $\lambda \mathrel{:=} \lambda_t$ and $W \mathrel{:=} W_t$.
Given $U = u$, the competitors are independent, so $\Pr[W = 1 \mid U = u] = \prod_{v \in C_t}\Pr[U_{a_t,v} < u^{\alpha_{t,v}}] = u^{\lambda}$.
Let $Y \mathrel{:=} W + U$, which ranks every winning outcome above every losing one, and then by the value of $U$.
$Y$ has a continuous distribution, and its survival function is, for $W = 1$ and for $W = 0$ respectively,
\begin{align*}
	\Pr[Y' \ge 1 + u] &= \int_u^1 x^\lambda\,\mathrm{d}x = \frac{1-u^{1+\lambda}}{1+\lambda}, \\
	\Pr[Y' \ge u] &= \frac{1}{1+\lambda} + \int_u^1 (1-x^\lambda)\,\mathrm{d}x = 1 - u + \frac{u^{1+\lambda}}{1+\lambda},
\end{align*}
where $Y'$ is an independent copy of $Y$.
Hence $R_t = \Pr[Y' \ge Y \mid Y]$, and by the probability integral transform, $R_t \sim \mathcal{U}([0,1])$.
\end{proof}

\section{Prompts}
\label{app:prompt}

In this section, we provide the prompt given to the research agent (\cref{app:prompt:agent}) and the prompts used for the back-translation and paraphrasing attacks during evaluation (\cref{app:prompt:attacks}).

\subsection{Agent Prompt}
\label{app:prompt:agent}

The following prompt is given to the autonomous research agent.

\begin{promptbox}[Watermark design prompt]
Your goal is to design a new distortion-free watermarking scheme in the \texttt{lm-wm-tools} library that significantly outperforms prior work across all key dimensions: detectability, robustness, and diversity. For diversity and PPL, the goal is to be as close as possible to the corresponding unwatermarked values (the leaderboard ranks based on the distance to the unwatermarked values).

You should focus on new experimental ideas and reason mathematically about watermarking to make genuine breakthroughs rather than tweaking hyperparameters.
Also, a constraint of the challenge is to focus only on distortion-free schemes.

A key challenge of watermark development is to make sure that:

\begin{itemize}[leftmargin=1.6em, itemsep=0.25em]
\item The scheme is distortion-free.
\item The scheme is robust to edits to the generated text.
\end{itemize}

Each of your submissions should be accompanied by:

\begin{itemize}[leftmargin=1.6em, itemsep=0.25em]
\item A proof that your scheme is distortion-free and well-calibrated.
\item A commit message.
\end{itemize}

Before an official submission, use the \texttt{lm-wm-tools verify} and \texttt{lm-wm-tools evaluate} commands to evaluate your proposed schemes.
We provide public configuration but you are also welcome to use your owns.
You have a limited number of official submissions (use \texttt{wm-harness limit} to see how many remain), so only evaluate schemes that you are confident (i) are sound and (ii) outperform prior work or submissions.
Soundness, in particular, is a significant constraint: your scheme should satisfy the following properties as much as possible (reported results are for baselines verification).

Let $K$ be the watermark key, $\Sigma$ the vocabulary, $\mathcal{L}$ the LLM probability distribution, $\mathcal{L}_K$ the watermarked LLM probability distribution, and $p_K(\omega)$ the detector p-value for text $\omega$ under key $K$.

\textbf{Distortion-free}

\[
\forall \omega \in \Sigma^*, \mathbb{E}_{K}[\mathcal{L}_{K}(\omega)] = \mathcal{L}(\omega)
\]

\begin{center}
\setlength{\tabcolsep}{5pt}
\begin{tabular}{lrrrrr}
\toprule
Scheme & Context 1 & Context 2 & Context 3 & Context 4 & Context 5 \\
\midrule
AAR & 1 & 0.4194 & 0.5642 & 1 & — \\
SynthID & 0.5092 & 0.6584 & 1 & 0.4254 & 1 \\
TextSeal & 1 & 1 & 1 & 1 & — \\
\bottomrule
\end{tabular}
\end{center}

\textbf{Soundness over the key}

\[
\forall \omega \in \Sigma^*, \forall \alpha \in [0, 1], \Pr_{K}[p_{K}(\omega) \le \alpha] \le \alpha
\]

\begin{center}
\setlength{\tabcolsep}{5pt}
\begin{tabular}{lrrrrr}
\toprule
Scheme & Context 1 & Context 2 & Context 3 & Context 4 & Context 5 \\
\midrule
AAR & 0.009677 & 0.02317 & 0.1124 & 0.2463 & — \\
SynthID & 0.1516 & 0.4806 & 0.5053 & 1 & 0.9043 \\
TextSeal & 0.3678 & 1 & 1 & 1 & — \\
\bottomrule
\end{tabular}
\end{center}

\textbf{Soundness for a deployed key}

\[
\forall K, \Pr_{\Omega \sim \mathcal{H}_0}[p_K(\Omega) \le \alpha] \le \alpha
\]

\begin{center}
\setlength{\tabcolsep}{5pt}
\begin{tabular}{lrrrrr}
\toprule
Scheme & Context 1 & Context 2 & Context 3 & Context 4 & Context 5 \\
\midrule
AAR & 2.981e-82 & 4.811e-14 & 3.813e-05 & 0.07999 & — \\
SynthID & 8.056e-122 & 9.575e-41 & 4.136e-09 & 6.433e-05 & 0.2573 \\
TextSeal & 5.265e-67 & 0.001826 & 1 & 1 & — \\
\bottomrule
\end{tabular}
\end{center}

You are operating autonomously.

\medskip
\textbf{Official submissions}\

When you have a submission you like, submit it for private grading through the black-box \texttt{wm-harness} command:

\begin{lstlisting}[language={},basicstyle=\ttfamily\footnotesize,numbers=none,mathescape=false,escapechar={},breaklines=true,columns=fullflexible,frame=single,rulecolor=\color{black!14},backgroundcolor=\color{black!2},xleftmargin=0pt,framexleftmargin=0pt]
wm-harness evaluate configs/MyScheme.yaml
wm-harness leaderboard
\end{lstlisting}

Run \texttt{wm-harness limit} before submitting to see the run's evaluation budget: \texttt{limit}, \texttt{used}, and \texttt{remaining} (\texttt{null} means unlimited). Each new official evaluation job consumes one call, even if verification or grading fails. Duplicate requests do not consume another call; removing a queued job does not refund it. Stop submitting when no calls remain.

The \texttt{wm-harness leaderboard} command shows you all prior evaluation results, and you should consider it trustworthy.

\texttt{wm-harness evaluate} runs official verification first, then evaluates the same committed snapshot only if every check passes. One job covers both stages. Results include \texttt{verification.all\_passed} and \texttt{verification.checks}; rejected submissions have no evaluation metrics, and the CLI exits with code 2. By default, the command waits for both stages; use \texttt{-{}-no-wait} to return the queued job immediately.

Before calling \texttt{wm-harness evaluate}:

\begin{enumerate}[leftmargin=1.6em, itemsep=0.25em]
\item Keep the working tree clean.
\item Commit the exact scheme, configuration, documentation, and required assets.
\item Submit the committed configuration path. The API grades an immutable archive of that commit and rejects dirty or uncommitted work.
\end{enumerate}

Use \texttt{wm-harness status JOB\_ID} to inspect an asynchronous job.
Status includes \texttt{progress.stage}, a readable \texttt{progress.message}, \texttt{progress.updated\_at} (last stage or completed grading pass), and \texttt{progress.heartbeat\_at} (worker heartbeat every five seconds while running). Long model startup or grading passes can leave the stage unchanged.

We automatically queue submissions, so it is better to wait until the submission is done: you can use \texttt{wm-harness wait JOB\_ID} to wait indefinitely for one job, or add \texttt{-{}-timeout <seconds>} to bound the wait (run this command in the foreground).
It is okay to wait for long periods (e.g., a few hours) to save tokens; repeatedly waiting for only one minute is inefficient.
You can also run \texttt{wm-harness list} to list all queued API queries along with their IDs.

Use \texttt{wm-harness remove JOB\_ID} to remove a job that is still queued. Running or completed jobs cannot be removed. Removed jobs disappear from the job list and can be submitted again.

Be aware that official evaluation can take a few minutes to a few hours, depending on the scheme implementation.

To help you, \texttt{lm-wm-tools verify} and \texttt{lm-wm-tools evaluate} are public subsets of the verification and evaluation stages run by \texttt{wm-harness evaluate}. You can use them locally to estimate whether your submission is worth evaluating on the private test.

API credentials for public paraphrasing evaluation are installed at \texttt{/home/sandbox/secret.sh}. Before running \texttt{lm-wm-tools evaluate}, load them in the same shell with \texttt{set -a; source /home/sandbox/secret.sh; set +a} so the CLI inherits the credentials.

\medskip
\textbf{Modifying \texttt{lm-wm-tools}}\

If you want to implement a scheme that does not fit well with the \texttt{lm-wm-tools} framework, you can modify the library more extensively.
However, for the harness to work, your changes must remain compatible with the \texttt{lm-wm-tools verify} and \texttt{lm-wm-tools evaluate} commands.
You should also make your watermark efficient and can change the library's inner workings to achieve the required efficiency.

\medskip
\textbf{Do not finish while official jobs are pending}\

When you are done, do not produce your final answer while any official evaluation job is queued or running.
Instead, run in the foreground:

\begin{lstlisting}[language={},basicstyle=\ttfamily\footnotesize,numbers=none,mathescape=false,escapechar={},breaklines=true,columns=fullflexible,frame=single,rulecolor=\color{black!14},backgroundcolor=\color{black!2},xleftmargin=0pt,framexleftmargin=0pt]
wm-harness wait
\end{lstlisting}

It waits until the official queue is empty and then returns the results of every submitted job.
Do not background it, replace it with a custom monitoring script, or end your turn while it is running.
Inspect all returned results, act on them, improve rejected or underperforming candidates, submit them with \texttt{wm-harness evaluate}, and run \texttt{wm-harness wait} again if that creates more queued work.
Only finish after the official queue is empty and you have checked the returned results.

\medskip
\textbf{Notes}\

The private configuration, datasets, seeds, raw reports, and grader filesystem are intentionally unavailable. Do not attempt to infer, access, or reproduce them: this would be considered cheating.
\end{promptbox}

\subsection{Attack Prompts}
\label{app:prompt:attacks}

Here, we detail the prompts we used for the back-translation and paraphrasing attacks.
In the example below, \texttt{<text>} is replaced with the actual text to translate or paraphrase. 
For back-translation, the prompt is applied twice, once for translating to French and the other time to translate back into the original language.

\begin{promptbox}[Back-translation prompt]
Translate the following text to \texttt{\{language\}}. Your reply should only contain the translated text.

\texttt{<text>}
\end{promptbox}

\begin{promptbox}[Paraphrasing prompt]
Please rewrite the following text and return only the rewritten text: \texttt{<text>}
\end{promptbox}

\section{Full Experimental Results}
\label{app:full_results}

In this section, we report the full numerical results behind the figures of the paper.
In~\cref{app:full_results:trajectories}, we give the results of every scheme submitted during the autoresearch trajectories of~\cref{sec:eval:trajectories}, and in~\cref{app:full_results:components}, every configuration evaluated in the component ablation of~\cref{app:components}.

\subsection{Autoresearch Trajectories}
\label{app:full_results:trajectories}

\Cref{tab:full-trajectories} reports the results of the three baselines and of the $52$ submissions of the four research agents from~\cref{fig:research-trajectories}, in submission order.

\begin{table}[t]
\centering
\caption{Full results of every submission of~\cref{fig:research-trajectories}, in submission order. Each TPR cell reports the guaranteed / empirical / no-correction settings of~\cref{app:soundness_key} (quality does not depend on the setting). Parentheses mark values without correction for schemes that then fail our soundness given a key test. Within each group, bold marks the best (unrounded) value of each column and setting.}
\label{tab:full-trajectories}
\scriptsize
\setlength{\tabcolsep}{4pt}
\resizebox{\linewidth}{!}{%
\begin{tabular}{@{}rlrrrrrrrr@{}}
\toprule
 & & & \multicolumn{4}{c}{Robustness TPR@1  $\uparrow$} & \multicolumn{3}{c}{Quality deviation (\%) $\downarrow$} \\
\cmidrule(lr){4-7}\cmidrule(l){8-10}
\# & Scheme & Clean TPR@1 $\uparrow$ & Deletion & Substitution & Back-trans. & Paraphrase & $\Delta$PPL & $\Delta$SB-2 & $\Delta$SB-3 \\
\midrule
\multicolumn{10}{@{}l}{\textbf{Baselines}} \\
 & TextSeal & 100 / \textbf{100} / \textbf{100} & \textbf{14} / \textbf{34} / \textbf{34} & \textbf{28} / \textbf{54} / \textbf{54} & \textbf{89} / \textbf{97} / \textbf{97} & \textbf{38} / \textbf{58} / \textbf{58} & 4.3 & 2.3 & 7.3 \\
 & AAR & \textbf{100} / 100 / (100) & 4 / 19 / (19) & 8 / 30 / (30) & 76 / 92 / (92) & 19 / 41 / (41) & 3.9 & 3.9 & 12 \\
 & SynthID & \textbf{100} / \textbf{100} / \textbf{(100)} & 9 / 30 / (30) & 20 / 44 / (44) & 88 / 95 / (95) & 35 / 56 / (56) & \textbf{2.5} & \textbf{2.0} & \textbf{7.1} \\
\addlinespace[2pt]
\multicolumn{10}{@{}l}{\textbf{\astra (xhigh)}} \\
1 & CascadeCopula & 100 / 100 / (100) & 60 / 79 / (80) & 60 / 78 / (79) & 95 / 99 / (99) & 58 / 76 / (76) & 2.3 & 3.9 & 10 \\
2 & LexicalCopula & 100 / 100 / (100) & 57 / 78 / (78) & 69 / 83 / (84) & 94 / 98 / (98) & 63 / 78 / (79) & 1.0 & 3.4 & 9.4 \\
3 & PortfolioRace & \textbf{100} / 100 / 100 & 51 / 67 / 67 & 70 / 80 / 80 & 94 / 98 / 98 & 59 / 71 / 71 & 3.1 & 2.1 & 4.7 \\
4 & TailRace & 99 / 100 / (100) & 44 / 66 / (67) & 55 / 73 / (73) & 87 / 94 / (94) & 50 / 70 / (70) & 4.8 & 3.1 & 8.3 \\
5 & DirectionalRenewal & \textbf{100} / 100 / 100 & 67 / 83 / 83 & 73 / 86 / 86 & 96 / \textbf{99} / \textbf{99} & 64 / 78 / 78 & 4.0 & 4.0 & 9.5 \\
6 & HaarCap & \textbf{100} / \textbf{100} / \textbf{100} & 64 / 80 / 80 & 69 / 84 / 84 & 94 / 98 / 98 & 62 / 76 / 76 & 1.5 & 1.6 & 5.9 \\
7 & HaarGamma & 100 / 100 / 100 & 67 / 80 / 80 & 74 / 84 / 84 & 96 / 98 / 98 & 66 / 76 / 76 & 1.5 & 2.3 & 6.7 \\
8 & HaarPredictive & 100 / \textbf{100} / \textbf{(100)} & 71 / 85 / (85) & 77 / 88 / (88) & \textbf{96} / 98 / (98) & 68 / 80 / (80) & 1.8 & 2.6 & 7.0 \\
9 & HaarGraph & 100 / 100 / 100 & 72 / 86 / 86 & 77 / 88 / 88 & 96 / 98 / 98 & 69 / 82 / 82 & 3.2 & 2.0 & 5.5 \\
10 & HaarLexical & 99 / 99 / (99) & 83 / 92 / (92) & 76 / 86 / (86) & 92 / 97 / (97) & 66 / 79 / (79) & 2.1 & 2.0 & 5.4 \\
11 & HaarLexicalFirst & 99 / 100 / (100) & \textbf{84} / 91 / (91) & \textbf{81} / \textbf{89} / \textbf{(89)} & 94 / 98 / (98) & 68 / 82 / (82) & 2.4 & 2.2 & 5.7 \\
12 & ProjectiveFirst & 99 / 100 / 100 & 81 / 91 / 91 & 75 / 86 / 86 & 93 / 97 / 97 & 64 / 78 / 78 & 1.7 & 1.9 & 4.2 \\
13 & AxialFirst & 99 / 100 / (100) & 83 / \textbf{92} / \textbf{(92)} & 78 / 88 / (88) & 95 / 98 / (98) & \textbf{70} / \textbf{83} / \textbf{(83)} & 3.1 & 2.3 & 5.7 \\
14 & RotationFirst & 99 / 100 / 100 & 81 / 90 / 90 & 77 / 86 / 86 & 93 / 97 / 97 & 68 / 80 / 80 & \textbf{1.0} & \textbf{1.6} & \textbf{4.0} \\
\addlinespace[2pt]
\multicolumn{10}{@{}l}{\textbf{\astra (low)}} \\
1 & FreshRace & 99 / 100 / (100) & 13 / 34 / (35) & 14 / 37 / (37) & 74 / 90 / (90) & 28 / 50 / (51) & 7.2 & 1.9 & 7.1 \\
2 & FreshRace-selective & 100 / 100 / (100) & 26 / 52 / (53) & 28 / 54 / (54) & 89 / 96 / (96) & 44 / 65 / (65) & \textbf{3.9} & 3.8 & 12 \\
3 & FreshRace-deployment & 100 / \textbf{100} / \textbf{(100)} & 39 / 61 / (62) & 43 / 69 / (69) & 92 / 97 / (97) & 53 / 72 / (72) & 4.9 & 2.9 & 11 \\
4 & FreshRace-channels & \textbf{100} / \textbf{100} / \textbf{100} & 38 / 60 / 60 & 46 / 65 / 65 & 94 / 98 / 98 & 55 / 71 / 71 & 8.8 & 2.2 & 8.0 \\
5 & FreshRace-adaptive-channels & 100 / 100 / (100) & 42 / 65 / (65) & 44 / 66 / (67) & 91 / 96 / (96) & 53 / 70 / (70) & 9.8 & \textbf{1.7} & \textbf{5.3} \\
6 & FreshRace-split & \textbf{100} / \textbf{100} / \textbf{(100)} & 58 / 75 / (75) & 60 / 78 / (78) & \textbf{96} / \textbf{98} / \textbf{(98)} & \textbf{62} / \textbf{76} / \textbf{(76)} & 8.0 & 3.3 & 9.5 \\
7 & FreshRace-antithetic-reservoir & \textbf{100} / \textbf{100} / \textbf{100} & 49 / 67 / 67 & 50 / 68 / 68 & 93 / 97 / 97 & 53 / 68 / 68 & 5.4 & 1.9 & 5.3 \\
8 & FreshRace-balanced & \textbf{100} / \textbf{100} / \textbf{100} & 40 / 60 / 60 & 50 / 67 / 67 & 94 / 97 / 97 & 55 / 70 / 70 & 7.4 & 2.8 & 9.0 \\
9 & FreshRace-lexical & \textbf{100} / \textbf{100} / \textbf{100} & 51 / 71 / 71 & 64 / \textbf{81} / \textbf{81} & 95 / 98 / 98 & 59 / 75 / 75 & 7.3 & 2.4 & 7.4 \\
10 & FreshRace-lexical & 100 / 100 / 100 & 51 / 69 / 69 & \textbf{66} / 81 / 81 & 95 / 98 / 98 & 59 / 74 / 74 & 5.2 & 2.6 & 7.7 \\
11 & RenewalRace & 95 / 98 / (98) & \textbf{80} / \textbf{90} / \textbf{(91)} & 52 / 71 / (74) & 82 / 92 / (92) & 52 / 69 / (72) & 4.7 & 3.2 & 5.8 \\
\addlinespace[2pt]
\multicolumn{10}{@{}l}{\textbf{\opus}} \\
1 & TallyMark & 100 / 100 / (100) & 37 / 63 / (64) & 57 / 77 / (78) & 86 / 94 / (95) & 44 / 63 / (64) & 4.1 & 3.1 & 8.3 \\
2 & TallyMark & 100 / 100 / (100) & 32 / 57 / (57) & 34 / 60 / (60) & 87 / 95 / (95) & 42 / 64 / (64) & 3.1 & 1.7 & 7.4 \\
3 & TallyMark & 100 / \textbf{100} / \textbf{(100)} & 33 / 61 / (62) & 60 / 80 / (81) & 88 / 95 / (95) & 47 / 67 / (67) & 5.2 & 3.6 & 10 \\
4 & TallyMark & 100 / 100 / (100) & 36 / 64 / (64) & 57 / 80 / (80) & 88 / 97 / (97) & 49 / 69 / (70) & \textbf{2.8} & 3.4 & 9.2 \\
5 & ChoirMark & 100 / \textbf{100} / \textbf{(100)} & 57 / 75 / (75) & 78 / 87 / (87) & 96 / 99 / (99) & 62 / 75 / (75) & 5.8 & 14 & 34 \\
6 & CanonMark & 100 / 100 / 100 & 43 / 59 / 59 & 66 / 78 / 78 & 93 / 96 / 96 & 51 / 64 / 64 & 4.2 & 2.6 & 4.3 \\
7 & LadderMark & 100 / 100 / 100 & 51 / 66 / 66 & 67 / 78 / 78 & 95 / 98 / 98 & 56 / 68 / 68 & 3.8 & 2.7 & 4.6 \\
8 & PreludeMark & \textbf{100} / \textbf{100} / \textbf{(100)} & \textbf{73} / \textbf{88} / \textbf{(88)} & 79 / 90 / (91) & 98 / 99 / (99) & 70 / 84 / (84) & 9.0 & 5.8 & 18 \\
9 & ConsortMark & 100 / \textbf{100} / \textbf{(100)} & 65 / 79 / (79) & 70 / 82 / (82) & 97 / 99 / (99) & 62 / 75 / (75) & 5.7 & 1.6 & 5.1 \\
10 & SestetMark & 100 / \textbf{100} / \textbf{(100)} & 66 / 82 / (82) & 78 / 86 / (86) & 97 / 99 / (99) & 67 / 78 / (78) & 9.0 & 1.3 & 5.1 \\
11 & MotetMark & 100 / 100 / (100) & 66 / 82 / (82) & 77 / 87 / (87) & 97 / 99 / (99) & 62 / 78 / (78) & 7.2 & 1.6 & 5.8 \\
12 & CodaMark & 100 / 100 / (100) & 66 / 81 / (81) & 77 / 86 / (86) & 97 / 99 / (99) & 65 / 79 / (79) & 8.0 & 1.5 & 5.4 \\
13 & AccentMark & \textbf{100} / \textbf{100} / \textbf{100} & 69 / 82 / 82 & 83 / 90 / 90 & \textbf{99} / \textbf{100} / \textbf{100} & 74 / 84 / 84 & 8.0 & 1.5 & 6.0 \\
14 & RowMark & \textbf{100} / \textbf{100} / \textbf{100} & 67 / 84 / 84 & 82 / 91 / 91 & \textbf{99} / 100 / 100 & 74 / \textbf{86} / \textbf{86} & 6.3 & 1.6 & 5.8 \\
15 & UnisonMark & 100 / \textbf{100} / \textbf{100} & 68 / 82 / 82 & 84 / 93 / 93 & 99 / 100 / 100 & 74 / 85 / 85 & 6.4 & 1.3 & 4.9 \\
16 & ConcordMark & \textbf{100} / \textbf{100} / \textbf{100} & 65 / 79 / 79 & \textbf{89} / \textbf{94} / \textbf{94} & 99 / 100 / 100 & \textbf{75} / 84 / 84 & 5.5 & \textbf{1.2} & \textbf{4.3} \\
\addlinespace[2pt]
\multicolumn{10}{@{}l}{\textbf{\gemini (high)}} \\
1 & OmniSeal & \textbf{100} / \textbf{100} / \textbf{100} & \textbf{18} / \textbf{41} / \textbf{41} & 34 / 58 / 58 & 91 / 97 / 97 & \textbf{45} / 62 / 62 & 6.8 & 5.2 & 16 \\
2 & OmniSeal & \textbf{100} / \textbf{100} / \textbf{100} & 16 / 40 / 40 & 34 / 57 / 57 & \textbf{93} / \textbf{98} / \textbf{98} & 44 / 64 / 64 & 7.4 & 5.0 & 15 \\
3 & OmniSeal & \textbf{100} / \textbf{100} / \textbf{100} & 16 / 40 / 40 & 32 / 58 / 58 & 92 / 98 / 98 & 44 / 63 / 63 & 5.6 & 4.9 & 15 \\
4 & OmniSeal & 100 / \textbf{100} / \textbf{100} & 16 / 40 / 40 & \textbf{35} / \textbf{61} / \textbf{61} & 92 / 97 / 97 & 42 / \textbf{64} / \textbf{64} & 6.9 & 5.3 & 16 \\
5 & OmniSeal-G0 & 100 / \textbf{100} / \textbf{100} & 14 / 37 / 37 & 32 / 58 / 58 & 90 / 96 / 96 & 41 / 63 / 63 & 7.6 & 4.9 & 15 \\
6 & OmniSeal-K40 & 100 / \textbf{100} / \textbf{100} & 15 / 37 / 37 & 31 / 58 / 58 & 91 / 97 / 97 & 38 / 59 / 59 & 11 & 5.7 & 17 \\
7 & OmniSeal-K45 & \textbf{100} / \textbf{100} / \textbf{100} & 16 / 40 / 40 & 34 / 58 / 58 & 92 / 98 / 98 & 42 / 62 / 62 & 9.1 & 5.3 & 16 \\
8 & OmniSeal-K35 & 100 / 100 / 100 & 13 / 36 / 36 & 29 / 56 / 56 & 92 / 98 / 98 & 38 / 60 / 60 & 14 & 6.0 & 17 \\
9 & OmniSeal-W05 & 100 / \textbf{100} / \textbf{100} & 17 / 39 / 39 & 31 / 57 / 57 & 92 / 98 / 98 & 41 / 61 / 61 & \textbf{5.1} & 5.1 & 16 \\
10 & OmniSeal-C4 & \textbf{100} / \textbf{100} / \textbf{100} & 5 / 16 / 16 & 14 / 34 / 34 & 85 / 94 / 94 & 25 / 45 / 45 & 6.1 & \textbf{3.7} & \textbf{11} \\
11 & OmniSeal-K30 & 100 / \textbf{100} / \textbf{100} & 12 / 32 / 32 & 26 / 53 / 53 & 92 / 98 / 98 & 38 / 60 / 60 & 18 & 7.0 & 19 \\
\bottomrule
\end{tabular}%
}
\end{table}

\begin{table}[t]
\centering
\caption{Watermark unit and context seeding~(\cref{app:components:unit,app:components:context}). Each cell reports identity grouping / \texttt{Lexical group}.}
\label{tab:full-components-context}
\scriptsize
\setlength{\tabcolsep}{4pt}
\resizebox{\linewidth}{!}{%
\begin{tabular}{@{}lrrrrrrrr@{}}
\toprule
 & & \multicolumn{4}{c}{Robustness TPR@1  $\uparrow$} & \multicolumn{3}{c}{Quality deviation (\%) $\downarrow$} \\
\cmidrule(lr){3-6}\cmidrule(l){7-9}
Variant & Clean TPR@1 $\uparrow$ & Deletion & Substitution & Back-trans. & Paraphrase & $\Delta$PPL & $\Delta$SB-2 & $\Delta$SB-3 \\
\midrule
Core (fixed $k=2$, \texttt{Gumbel race}, \texttt{Gamma}) & 100 / 100 & 56 / 54 & 67 / 76 & 97 / 98 & 64 / 68 & 0.1 / 1.1 & 4.5 / 4.6 & 16 / 16 \\
\addlinespace[2pt]
\multicolumn{9}{@{}l}{\textbf{Fixed-length}} \\
$k=0$ & 100 / 99 & 92 / 88 & 70 / 61 & 93 / 90 & 67 / 60 & 30 / 32 & 12 / 13 & 29 / 31 \\
$k=1$ & 99 / 99 & 57 / 51 & 51 / 55 & 89 / 89 & 50 / 49 & 2.7 / 2.5 & 5.5 / 4.1 & 14 / 11 \\
$k=3$ & 100 / 100 & 25 / 25 & 42 / 60 & 96 / 97 & 52 / 54 & 0.9 / 0.4 & 3.1 / 2.1 & 11 / 8.1 \\
\addlinespace[2pt]
\multicolumn{9}{@{}l}{\textbf{Unordered}} \\
Pair $\{\omega_{t-1}, v\}$ & 100 / 100 & 91 / 90 & 91 / 94 & 99 / 99 & 82 / 85 & 0.2 / 0.0 & 13 / 12 & 34 / 32 \\
$k=2$ & 100 / 100 & 58 / 56 & 68 / 79 & 97 / 98 & 62 / 67 & 0.6 / 1.0 & 4.1 / 3.7 & 15 / 14 \\
$k=3$ & 100 / 100 & 22 / 21 & 43 / 56 & 97 / 97 & 50 / 51 & 0.0 / 0.9 & 2.5 / 2.2 & 9.5 / 8.0 \\
\addlinespace[2pt]
\multicolumn{9}{@{}l}{\textbf{Adaptive and global length}} \\
\texttt{Adaptive length} & 100 / 100 & 73 / 70 & 77 / 79 & 98 / 98 & 71 / 71 & 2.3 / 1.9 & 5.7 / 5.4 & 17 / 16 \\
\quad + \texttt{Excess} detector & 100 / -- & 72 / -- & 78 / -- & 98 / -- & 72 / -- & 1.0 / -- & 5.6 / -- & 17 / -- \\
\texttt{Global + local} & 100 / 100 & 88 / 85 & 76 / 77 & 96 / 96 & 70 / 69 & 16 / 18 & 10 / 9.2 & 25 / 23 \\
\texttt{Global first use} & 100 / 100 & 90 / 85 & 76 / 80 & 97 / 98 & 72 / 72 & 15 / 16 & 13 / 11 & 34 / 28 \\
\addlinespace[2pt]
\multicolumn{9}{@{}l}{\textbf{Anchored}} \\
\texttt{Anchored} ($S_{\max}=4$) & 96 / 95 & 38 / 35 & 21 / 33 & 71 / 72 & 30 / 28 & 3.0 / 1.0 & 5.6 / 3.9 & 16 / 12 \\
\addlinespace[2pt]
\multicolumn{9}{@{}l}{\textbf{Suffix cascade and counting}} \\
\texttt{Suffix cascade} & 100 / 100 & 78 / 76 & 85 / 86 & 99 / 99 & 76 / 78 & 2.3 / 0.5 & 9.5 / 7.6 & 26 / 22 \\
\texttt{Tally} & 100 / 100 & 75 / 71 & 89 / 93 & 99 / 99 & 78 / 77 & 1.2 / 2.2 & 25 / 13 & 61 / 34 \\
\texttt{Ladder} $K=3$ & 100 / 100 & 72 / 72 & 80 / 90 & 98 / 99 & 75 / 73 & 4.6 / 1.5 & 25 / 14 & 61 / 36 \\
\texttt{Ladder} $K=3$ + floor & 100 / 100 & 76 / 76 & 81 / 91 & 99 / 99 & 73 / 75 & 2.8 / 1.1 & 25 / 13 & 61 / 35 \\
\texttt{Ladder} $K=4$ & 100 / 100 & 77 / 73 & 84 / 87 & 99 / 99 & 77 / 76 & 0.7 / 0.5 & 25 / 14 & 62 / 35 \\
\texttt{Ladder} $K=4$ + floor & 100 / 100 & 77 / 75 & 80 / 86 & 99 / 99 & 74 / 78 & 2.4 / 0.1 & 25 / 13 & 62 / 36 \\
\addlinespace[2pt]
\multicolumn{9}{@{}l}{\textbf{On \texttt{Signed sphere} + \texttt{ResidualRace}}} \\
Fixed-length $k=2$ & 100 / 100 & 33 / 35 & 47 / 63 & 96 / 95 & 49 / 53 & 0.7 / 1.0 & 0.4 / 0.4 & 0.4 / 0.3 \\
Fixed-length $k=1$ & 100 / 100 & 75 / 75 & 74 / 81 & 98 / 97 & 72 / 71 & 0.5 / 0.1 & 0.3 / 0.7 & 2.1 / 0.3 \\
Unordered pair & 100 / 100 & 72 / 76 & 72 / 82 & 97 / 97 & 67 / 72 & 1.2 / 1.0 & 0.1 / 0.1 & 1.7 / 0.9 \\
\texttt{Global + local} ($k=1$) & 100 / 100 & 86 / 84 & 74 / 76 & 96 / 94 & 68 / 65 & 0.5 / 0.7 & 0.6 / 0.8 & 2.5 / 2.1 \\
\texttt{Global first use} ($k=1$) & 100 / 100 & 89 / 85 & 78 / 81 & 96 / 96 & 73 / 70 & 0.7 / 1.8 & 0.1 / 0.3 & 1.1 / 1.4 \\
\bottomrule
\end{tabular}%
}
\end{table}

\begin{table}[t]
\centering
\caption{Watermark score~(\cref{app:components:score}). Each cell reports identity grouping / \texttt{Lexical group}.}
\label{tab:full-components-score}
\scriptsize
\setlength{\tabcolsep}{4pt}
\resizebox{\linewidth}{!}{%
\begin{tabular}{@{}lrrrrrrrr@{}}
\toprule
 & & \multicolumn{4}{c}{Robustness TPR@1  $\uparrow$} & \multicolumn{3}{c}{Quality deviation (\%) $\downarrow$} \\
\cmidrule(lr){3-6}\cmidrule(l){7-9}
Variant & Clean TPR@1 $\uparrow$ & Deletion & Substitution & Back-trans. & Paraphrase & $\Delta$PPL & $\Delta$SB-2 & $\Delta$SB-3 \\
\midrule
Core (fixed $k=2$, \texttt{Gumbel race}, \texttt{Gamma}) & 100 / 100 & 56 / 54 & 67 / 76 & 97 / 98 & 64 / 68 & 0.1 / 1.1 & 4.5 / 4.6 & 16 / 16 \\
\addlinespace[2pt]
\multicolumn{9}{@{}l}{\textbf{\texttt{Prelude}}} \\
On the core & 100 / 100 & 54 / 52 & 67 / 74 & 96 / 97 & 62 / 63 & 3.2 / 2.6 & 2.4 / 2.8 & 11 / 11 \\
On \texttt{Ladder} $K=4$ + floor & 100 / 100 & 74 / 70 & 78 / 84 & 99 / 98 & 74 / 73 & 3.2 / 1.0 & 6.2 / 5.3 & 18 / 15 \\
\addlinespace[2pt]
\multicolumn{9}{@{}l}{\textbf{\texttt{Gaussian copula}}} \\
$\rho=0.7$ & 99 / 100 & 23 / 20 & 29 / 39 & 83 / 84 & 34 / 37 & 0.3 / 0.4 & 0.3 / 0.6 & 1.7 / 3.5 \\
$\rho=0.8$ & 100 / 100 & 34 / 33 & 42 / 56 & 89 / 91 & 45 / 50 & 0.9 / 1.2 & 0.5 / 0.4 & 4.6 / 4.1 \\
$\rho=0.85$ & 100 / 100 & 36 / 38 & 50 / 63 & 93 / 95 & 49 / 55 & 0.1 / 3.0 & 1.4 / 0.9 & 6.5 / 5.5 \\
$\rho=0.9$ & 100 / 100 & 44 / 44 & 57 / 67 & 94 / 95 & 54 / 59 & 1.7 / 0.1 & 1.0 / 2.0 & 7.1 / 8.3 \\
$\rho=0.95$ & 100 / 100 & 52 / 52 & 63 / 71 & 97 / 97 & 62 / 62 & 0.6 / 1.4 & 2.2 / 2.5 & 9.7 / 10 \\
\addlinespace[2pt]
\multicolumn{9}{@{}l}{\textbf{\texttt{Gaussian copula} + common-pair filter}} \\
$\rho=0.7$ & 99 / 100 & 20 / 19 & 28 / 38 & 80 / 84 & 34 / 35 & 1.6 / 0.3 & 0.1 / 0.1 & 2.2 / 2.5 \\
$\rho=0.8$ & 100 / 100 & 32 / 35 & 39 / 52 & 90 / 92 & 43 / 48 & 1.1 / 0.7 & 0.3 / 0.9 & 3.5 / 4.6 \\
$\rho=0.85$ & 100 / 100 & 38 / 36 & 47 / 60 & 92 / 92 & 49 / 53 & 1.3 / 1.4 & 0.6 / 1.5 & 4.7 / 6.7 \\
$\rho=0.9$ & 100 / 100 & 44 / 43 & 50 / 64 & 94 / 95 & 53 / 55 & 0.3 / 0.8 & 0.8 / 1.5 & 5.9 / 7.0 \\
$\rho=0.95$ & 100 / 100 & 48 / 50 & 59 / 72 & 96 / 98 & 61 / 64 & 0.1 / 1.0 & 1.2 / 2.3 & 7.3 / 9.4 \\
$\rho=1$ & 100 / 100 & 57 / 56 & 64 / 75 & 97 / 98 & 63 / 68 & 2.4 / 1.6 & 3.3 / 4.0 & 13 / 14 \\
\addlinespace[2pt]
\multicolumn{9}{@{}l}{\textbf{\texttt{Private sampling}}} \\
$r=0.1$ & 100 / 100 & 47 / 45 & 58 / 67 & 95 / 97 & 57 / 62 & 0.9 / 1.3 & 2.2 / 2.6 & 9.7 / 10 \\
$r=0.2$ & 100 / 100 & 42 / 38 & 52 / 62 & 93 / 93 & 49 / 57 & 2.5 / 1.6 & 1.1 / 1.5 & 6.6 / 7.1 \\
$r=0.3$ & 99 / 99 & 30 / 30 & 39 / 52 & 87 / 89 & 44 / 48 & 4.0 / 3.5 & 0.6 / 0.9 & 4.6 / 5.1 \\
\addlinespace[2pt]
\multicolumn{9}{@{}l}{\textbf{\texttt{Tail indicator}}} \\
\texttt{Tail indicator} & 99 / 99 & 27 / 26 & 42 / 48 & 84 / 86 & 39 / 42 & 2.2 / 1.3 & 0.3 / 0.7 & 4.0 / 4.6 \\
\addlinespace[2pt]
\multicolumn{9}{@{}l}{\textbf{Independent \texttt{Channels}}} \\
$L=4$ & 100 / 100 & 51 / 47 & 60 / 71 & 96 / 97 & 60 / 61 & 2.0 / 1.4 & 2.3 / 1.9 & 8.1 / 6.9 \\
$L=16$ & 100 / 100 & 39 / 38 & 56 / 63 & 96 / 94 & 48 / 52 & 2.3 / 0.2 & 0.3 / 0.0 & 2.1 / 1.0 \\
$L=32$ & 100 / 100 & 36 / 34 & 48 / 62 & 94 / 95 & 46 / 50 & 0.0 / 1.1 & 0.5 / 0.3 & 0.4 / 0.2 \\
$L=64$ & 100 / 100 & 32 / 30 & 46 / 57 & 92 / 94 & 42 / 49 & 0.3 / 1.5 & 0.6 / 0.5 & 0.8 / 0.7 \\
$L=128$ & 100 / 100 & 29 / 28 & 41 / 57 & 91 / 92 & 42 / 48 & 0.8 / 0.4 & 0.4 / 0.6 & 0.7 / 1.3 \\
\addlinespace[2pt]
\multicolumn{9}{@{}l}{\textbf{Correlated \texttt{Channels}}} \\
\texttt{Latin} $L=4$ & 100 / 100 & 50 / 48 & 64 / 72 & 96 / 97 & 56 / 61 & 2.0 / 2.9 & 1.6 / 1.5 & 6.5 / 5.8 \\
\texttt{Antithetic} $L=4$ & 100 / 100 & 46 / 43 & 58 / 70 & 97 / 95 & 60 / 59 & 0.3 / 0.4 & 1.6 / 0.6 & 6.6 / 4.3 \\
\texttt{Latin} $L=16$ & 100 / 100 & 40 / 40 & 51 / 64 & 95 / 96 & 50 / 52 & 0.4 / 0.6 & 0.2 / 0.1 & 0.5 / 1.2 \\
\texttt{Stratified} $L=16$ & 100 / 100 & 40 / 37 & 52 / 64 & 94 / 94 & 50 / 48 & 1.0 / 1.7 & 0.1 / 0.3 & 1.0 / 0.0 \\
\addlinespace[2pt]
\multicolumn{9}{@{}l}{\textbf{\texttt{Gaussian direction} + \texttt{Gumbel race}}} \\
$d=1$ & 100 / 100 & 36 / 37 & 50 / 59 & 95 / 95 & 52 / 53 & 1.4 / 0.5 & 2.8 / 2.4 & 11 / 9.0 \\
$d=2$ & 100 / 100 & 28 / 28 & 41 / 50 & 92 / 93 & 46 / 47 & 0.2 / 0.5 & 0.5 / 0.6 & 3.5 / 3.1 \\
$d=4$ & 100 / 100 & 21 / 18 & 31 / 44 & 89 / 90 & 37 / 38 & 2.3 / 1.4 & 0.8 / 0.5 & 0.6 / 0.3 \\
$d=8$ & 100 / 100 & 15 / 15 & 24 / 37 & 83 / 86 & 31 / 33 & 1.8 / 1.0 & 0.1 / 0.7 & 0.0 / 0.7 \\
\addlinespace[2pt]
\multicolumn{9}{@{}l}{\textbf{\texttt{Gaussian direction} + \texttt{ResidualRace}}} \\
$d=1$ & 100 / 100 & 44 / 43 & 58 / 72 & 97 / 98 & 55 / 62 & 1.5 / 2.1 & 5.4 / 3.6 & 18 / 13 \\
$d=2$ & 100 / 100 & 34 / 33 & 48 / 65 & 96 / 96 & 51 / 55 & 0.9 / 1.1 & 1.1 / 1.1 & 5.1 / 5.1 \\
$d=4$ & 100 / 100 & 29 / 27 & 39 / 53 & 94 / 95 & 45 / 50 & 0.2 / 0.8 & 0.2 / 0.3 & 1.3 / 0.7 \\
$d=8$ & 100 / 100 & 20 / 20 & 32 / 48 & 91 / 92 & 38 / 41 & 1.5 / 2.7 & 0.9 / 0.6 & 0.9 / 0.3 \\
\addlinespace[2pt]
\multicolumn{9}{@{}l}{\textbf{\texttt{Spherical} + \texttt{Gumbel race}}} \\
\texttt{Signed sphere} & 100 / 100 & 29 / 28 & 41 / 57 & 90 / 94 & 43 / 46 & 2.2 / 1.2 & 0.9 / 1.0 & 1.6 / 1.9 \\
\texttt{Unoriented axis} & 100 / 100 & 30 / 27 & 41 / 56 & 91 / 92 & 45 / 45 & 2.4 / 0.5 & 0.2 / 0.5 & 0.2 / 0.7 \\
\texttt{Complex ray} & 100 / 100 & 24 / 23 & 34 / 50 & 88 / 91 & 39 / 42 & 0.0 / 1.6 & 1.5 / 0.5 & 2.7 / 0.8 \\
\texttt{Rotation} & 100 / 100 & 25 / 24 & 39 / 51 & 89 / 91 & 39 / 45 & 1.1 / 2.1 & 0.9 / 1.1 & 1.5 / 2.3 \\
\addlinespace[2pt]
\multicolumn{9}{@{}l}{\textbf{\texttt{Spherical} + \texttt{ResidualRace}}} \\
\texttt{Signed sphere} & 100 / 100 & 33 / 35 & 47 / 63 & 96 / 95 & 49 / 53 & 0.7 / 1.0 & 0.4 / 0.4 & 0.4 / 0.3 \\
\texttt{Unoriented axis} & 100 / 100 & 34 / 33 & 48 / 62 & 94 / 95 & 46 / 53 & 3.3 / 0.1 & 0.6 / 0.7 & 1.1 / 1.1 \\
\texttt{Complex ray} & 100 / 100 & 29 / 28 & 44 / 56 & 93 / 94 & 45 / 48 & 0.4 / 0.0 & 1.1 / 1.0 & 2.1 / 1.7 \\
\texttt{Rotation} & 100 / 100 & 33 / 32 & 47 / 63 & 94 / 95 & 47 / 52 & 1.0 / 1.2 & 0.4 / 0.9 & 0.5 / 1.7 \\
\bottomrule
\end{tabular}%
}
\end{table}

\begin{table}[p]
\centering
\caption{Logits transformation~(\cref{app:components:transform}). Each cell reports identity grouping / \texttt{Lexical group}.}
\label{tab:full-components-transform}
\scriptsize
\setlength{\tabcolsep}{4pt}
\resizebox{\linewidth}{!}{%
\begin{tabular}{@{}lrrrrrrrr@{}}
\toprule
 & & \multicolumn{4}{c}{Robustness TPR@1  $\uparrow$} & \multicolumn{3}{c}{Quality deviation (\%) $\downarrow$} \\
\cmidrule(lr){3-6}\cmidrule(l){7-9}
Variant & Clean TPR@1 $\uparrow$ & Deletion & Substitution & Back-trans. & Paraphrase & $\Delta$PPL & $\Delta$SB-2 & $\Delta$SB-3 \\
\midrule
Core (fixed $k=2$, \texttt{Gumbel race}, \texttt{Gamma}) & 100 / 100 & 56 / 54 & 67 / 76 & 97 / 98 & 64 / 68 & 0.1 / 1.1 & 4.5 / 4.6 & 16 / 16 \\
\addlinespace[2pt]
\multicolumn{9}{@{}l}{\textbf{Races with private randomness}} \\
\texttt{BandRace} & 98 / 98 & 16 / 16 & 20 / 29 & 73 / 75 & 29 / 29 & 2.0 / 0.7 & 0.3 / 0.1 & 1.7 / 2.5 \\
\texttt{PoolRace} ($K=16$) & 100 / 100 & 33 / 27 & 44 / 50 & 93 / 92 & 46 / 52 & 3.2 / 0.2 & 0.8 / 1.0 & 5.3 / 5.6 \\
\texttt{MomentBridge} ($K=8$) & 88 / 88 & 3 / 4 & 6 / 7 & 37 / 36 & 8 / 11 & 1.9 / 1.2 & 0.5 / 0.8 & 0.2 / 0.7 \\
\addlinespace[2pt]
\multicolumn{9}{@{}l}{\textbf{Quantile}} \\
\texttt{QuantileRace} & 93 / 91 & 4 / 5 & 7 / 12 & 44 / 46 & 12 / 13 & 1.4 / 0.3 & 2.6 / 1.6 & 10 / 7.5 \\
\quad + \texttt{Wheel} & 85 / 86 & 0 / 0 & 1 / 2 & 19 / 21 & 3 / 4 & 2.7 / 0.8 & 0.5 / 0.7 & 1.3 / 1.4 \\
\addlinespace[2pt]
\multicolumn{9}{@{}l}{\textbf{Transport ($L=K$ independent channels)}} \\
\texttt{TransportRace} $K=32$ & 99 / 99 & 14 / 11 & 26 / 35 & 83 / 85 & 31 / 32 & 0.3 / 0.5 & 0.3 / 0.6 & 0.7 / 1.4 \\
\texttt{TransportRace} $K=128$ & 100 / 100 & 23 / 20 & 32 / 48 & 87 / 88 & 38 / 40 & 1.8 / 0.9 & 0.8 / 1.2 & 2.1 / 2.5 \\
\texttt{QuotaTransport} $K=32$ & 99 / 100 & 19 / 18 & 30 / 40 & 87 / 87 & 37 / 35 & 2.2 / 0.2 & 1.1 / 1.0 & 2.4 / 1.9 \\
\texttt{QuotaTransport} $K=128$ & 100 / 100 & 22 / 20 & 34 / 48 & 88 / 90 & 41 / 39 & 0.3 / 0.3 & 0.6 / 1.4 & 1.4 / 2.8 \\
\addlinespace[2pt]
\multicolumn{9}{@{}l}{\textbf{Residual clocks}} \\
\texttt{ResidualRace} & 100 / 100 & 63 / 63 & 72 / 82 & 98 / 99 & 70 / 74 & 1.5 / 0.9 & 9.0 / 6.3 & 27 / 21 \\
\addlinespace[2pt]
\multicolumn{9}{@{}l}{\textbf{Skeleton}} \\
\texttt{Last important} + \texttt{Content gate} & 99 / 98 & 54 / 47 & 26 / 30 & 84 / 80 & 46 / 44 & 0.5 / 0.6 & 1.0 / 0.1 & 3.0 / 1.0 \\
\texttt{Interleaved} + \texttt{Interleaved gate} & 99 / 98 & 49 / 45 & 42 / 46 & 88 / 85 & 48 / 46 & 0.3 / 1.3 & 2.3 / 1.4 & 6.0 / 4.1 \\
\bottomrule
\end{tabular}%
}
\end{table}

\begin{table}[p]
\centering
\caption{Detection~(\cref{app:components:detection}), on fixed generated texts: only the detector changes, so the quality is that of the generation. Each cell reports identity grouping / \texttt{Lexical group}.}
\label{tab:full-components-detection}
\scriptsize
\setlength{\tabcolsep}{4pt}
\resizebox{\linewidth}{!}{%
\begin{tabular}{@{}lrrrrrrrr@{}}
\toprule
 & & \multicolumn{4}{c}{Robustness TPR@1  $\uparrow$} & \multicolumn{3}{c}{Quality deviation (\%) $\downarrow$} \\
\cmidrule(lr){3-6}\cmidrule(l){7-9}
Variant & Clean TPR@1 $\uparrow$ & Deletion & Substitution & Back-trans. & Paraphrase & $\Delta$PPL & $\Delta$SB-2 & $\Delta$SB-3 \\
\midrule
\multicolumn{9}{@{}l}{\textbf{Core texts, no forward pass}} \\
\texttt{Gamma} & 100 / 100 & 56 / 54 & 67 / 76 & 97 / 98 & 64 / 68 & 0.1 / 1.1 & 4.5 / 4.6 & 16 / 16 \\
\texttt{Gaussian} & 100 / 100 & 40 / 40 & 51 / 64 & 95 / 96 & 54 / 57 & 0.1 / 1.1 & 4.5 / 4.6 & 16 / 16 \\
\texttt{Clustering} & 100 / 100 & 57 / 55 & 67 / 77 & 97 / 98 & 64 / 68 & 0.1 / 1.1 & 4.5 / 4.6 & 16 / 16 \\
\texttt{Excess} & 100 / 100 & 57 / 55 & 68 / 77 & 97 / 98 & 64 / 68 & 0.1 / 1.1 & 4.5 / 4.6 & 16 / 16 \\
\addlinespace[2pt]
\multicolumn{9}{@{}l}{\textbf{Core texts, forward pass}} \\
\texttt{Weighted Gamma} & 100 / 100 & 59 / 56 & 73 / 81 & 99 / 100 & 73 / 79 & 0.1 / 1.1 & 4.5 / 4.6 & 16 / 16 \\
\texttt{Position evidence} & 100 / 100 & 57 / 56 & 70 / 80 & 98 / 99 & 68 / 74 & 0.1 / 1.1 & 4.5 / 4.6 & 16 / 16 \\
\quad + weighting & 100 / 100 & 60 / 60 & 76 / 85 & 99 / 100 & 76 / 82 & 0.1 / 1.1 & 4.5 / 4.6 & 16 / 16 \\
\addlinespace[2pt]
\multicolumn{9}{@{}l}{\textbf{Independent \texttt{Channels}}} \\
$L=4$, Bonferroni & 100 / 100 & 51 / 47 & 60 / 71 & 96 / 97 & 60 / 61 & 2.0 / 1.4 & 2.3 / 1.9 & 8.1 / 6.9 \\
$L=4$, \v{S}id\'ak & 100 / 100 & 51 / 47 & 60 / 71 & 96 / 97 & 60 / 61 & 2.0 / 1.4 & 2.3 / 1.9 & 8.1 / 6.9 \\
$L=16$, Bonferroni & 100 / 100 & 39 / 38 & 56 / 63 & 96 / 94 & 48 / 52 & 2.3 / 0.2 & 0.3 / 0.0 & 2.1 / 1.0 \\
$L=16$, \v{S}id\'ak & 100 / 100 & 39 / 38 & 56 / 63 & 96 / 94 & 48 / 52 & 2.3 / 0.2 & 0.3 / 0.0 & 2.1 / 1.0 \\
$L=32$, Bonferroni & 100 / 100 & 36 / 34 & 48 / 62 & 94 / 95 & 46 / 50 & 0.0 / 1.1 & 0.5 / 0.3 & 0.4 / 0.2 \\
$L=32$, \v{S}id\'ak & 100 / 100 & 36 / 34 & 48 / 62 & 94 / 95 & 46 / 50 & 0.0 / 1.1 & 0.5 / 0.3 & 0.4 / 0.2 \\
$L=64$, Bonferroni & 100 / 100 & 32 / 30 & 46 / 57 & 92 / 94 & 42 / 49 & 0.3 / 1.5 & 0.6 / 0.5 & 0.8 / 0.7 \\
$L=64$, \v{S}id\'ak & 100 / 100 & 32 / 30 & 46 / 57 & 92 / 94 & 42 / 49 & 0.3 / 1.5 & 0.6 / 0.5 & 0.8 / 0.7 \\
$L=128$, Bonferroni & 100 / 100 & 29 / 28 & 41 / 57 & 91 / 92 & 42 / 48 & 0.8 / 0.4 & 0.4 / 0.6 & 0.7 / 1.3 \\
$L=128$, \v{S}id\'ak & 100 / 100 & 29 / 28 & 41 / 57 & 91 / 92 & 42 / 48 & 0.8 / 0.4 & 0.4 / 0.6 & 0.7 / 1.3 \\
\addlinespace[2pt]
\multicolumn{9}{@{}l}{\textbf{\texttt{Signed sphere} + \texttt{ResidualRace} texts}} \\
\texttt{Predictive Gamma} & 100 / 100 & 33 / 35 & 47 / 63 & 96 / 95 & 49 / 53 & 0.7 / 1.0 & 0.4 / 0.4 & 0.4 / 0.3 \\
\texttt{Resultant length} & 100 / 100 & 19 / 19 & 27 / 43 & 91 / 92 & 37 / 40 & 0.7 / 1.0 & 0.4 / 0.4 & 0.4 / 0.3 \\
\texttt{Spherical caps} & 100 / 100 & 22 / 24 & 31 / 50 & 89 / 91 & 37 / 40 & 0.7 / 1.0 & 0.4 / 0.4 & 0.4 / 0.3 \\
Length + caps & 100 / 100 & 24 / 26 & 35 / 53 & 92 / 93 & 40 / 44 & 0.7 / 1.0 & 0.4 / 0.4 & 0.4 / 0.3 \\
\texttt{Anchor distance} & 100 / 100 & 32 / 32 & 43 / 59 & 94 / 94 & 44 / 49 & 0.7 / 1.0 & 0.4 / 0.4 & 0.4 / 0.3 \\
\addlinespace[2pt]
\multicolumn{9}{@{}l}{\textbf{\texttt{Gaussian direction} + \texttt{Gumbel race} texts}} \\
\texttt{Gaussian norm} $d=1$ & 100 / 100 & 36 / 37 & 50 / 59 & 95 / 95 & 52 / 53 & 1.4 / 0.5 & 2.8 / 2.4 & 11 / 9.0 \\
\texttt{Residual energy} $d=1$ & 100 / 100 & 50 / 51 & 65 / 70 & 97 / 97 & 61 / 63 & 1.4 / 0.5 & 2.8 / 2.4 & 11 / 9.0 \\
\texttt{Gaussian norm} $d=2$ & 100 / 100 & 28 / 28 & 41 / 50 & 92 / 93 & 46 / 47 & 0.2 / 0.5 & 0.5 / 0.6 & 3.5 / 3.1 \\
\texttt{Residual energy} $d=2$ & 100 / 100 & 36 / 36 & 49 / 58 & 93 / 94 & 51 / 52 & 0.2 / 0.5 & 0.5 / 0.6 & 3.5 / 3.1 \\
\texttt{Circle search} $d=2$ & 100 / 100 & 43 / 41 & 55 / 65 & 96 / 96 & 56 / 57 & 0.2 / 0.5 & 0.5 / 0.6 & 3.5 / 3.1 \\
\texttt{Gaussian norm} $d=4$ & 100 / 100 & 21 / 18 & 31 / 44 & 89 / 90 & 37 / 38 & 2.3 / 1.4 & 0.8 / 0.5 & 0.6 / 0.3 \\
\texttt{Residual energy} $d=4$ & 100 / 100 & 24 / 22 & 35 / 47 & 89 / 90 & 40 / 40 & 2.3 / 1.4 & 0.8 / 0.5 & 0.6 / 0.3 \\
\texttt{Gaussian norm} $d=8$ & 100 / 100 & 15 / 15 & 24 / 37 & 83 / 86 & 31 / 33 & 1.8 / 1.0 & 0.1 / 0.7 & 0.0 / 0.7 \\
\texttt{Residual energy} $d=8$ & 100 / 100 & 16 / 14 & 24 / 36 & 82 / 85 & 31 / 31 & 1.8 / 1.0 & 0.1 / 0.7 & 0.0 / 0.7 \\
\bottomrule
\end{tabular}%
}
\end{table}

\paragraph{Setup}
We use the private evaluation of~\cref{sec:eval}: for each metric, we report the worst result over the $5$ watermark keys, each evaluated on $1000$ replies.
We evaluate every scheme under the three p-value corrections of~\cref{app:soundness_key}: \emph{guaranteed} (p-values multiplied by $1/\beta = 20$), \emph{empirical} (smallest multiplier $\kappa \ge 1$ that passes our soundness tests,~\cref{fig:soundness_key_multipliers}), and \emph{no correction} ($\kappa = 1$).
As in~\cref{app:soundness_key}, we first remove the outer multiplier chosen by the agent (if any) and keep the rest of the detector unchanged.
For each watermark key, all settings share the same generated replies and the same attacked texts: only the multiplier applied to the p-values changes.
Hence, the quality columns do not depend on the setting.

\paragraph{Reading the table}
Each TPR cell reports the guaranteed / empirical / no-correction results, and the first column gives the submission index (\ie the x-axis of~\cref{fig:research-trajectories}).
Parentheses mark the results without correction of the schemes that then fail our soundness given a key test (\ie whose empirical multiplier is above $1$).
Note that~\cref{fig:research-trajectories} instead uses the multiplier chosen by each agent: $20$ for \astra (except for the first two submissions of the low variant), between $1$ and $6$ for \opus, and none for \gemini and the baselines.
Moreover, as we evaluate the settings on newly generated replies (with the same watermark keys), the results can differ by a few points from those of~\cref{fig:research-trajectories}.

\subsection{Component Ablation}
\label{app:full_results:components}

\Cref{tab:full-components-context,tab:full-components-score,tab:full-components-transform,tab:full-components-detection} report all the configurations of the component ablation of~\cref{app:components}, following the same structure: watermark unit and context seeding, watermark score, logits transformation, and detection.

\paragraph{Setup}
We use the experimental setup of~\cref{app:components:prelim}: each configuration swaps one or a few components of the simple baseline scheme (identity grouping, \texttt{Fixed-length} context seeding with $k=2$, $\tilde{U} = U$, the \texttt{Gumbel race}, and the \texttt{Gamma} detector).
All configurations are evaluated on the private evaluation of~\cref{sec:eval} with a single watermark key, and their p-values are multiplied by $1/\beta = 20$ so that they are sound given a key~(\cref{sec:method:definition_validity}).
Unlike the bar plots of~\cref{app:components}, which show differences, we report absolute values with the same columns as~\cref{tab:selected-schemes}: the TPR@1\%FPR on clean text and under each attack, and the relative distance (in \%) to the unwatermarked model in perplexity and bigram and trigram Self-BLEU.

\paragraph{Reading the tables}
Each cell reports the result with the identity grouping (left) and with the \texttt{Lexical group}~(\cref{app:components:unit}) (right), so that comparing the two sides of a cell gives the effect of the watermark unit on that configuration.
The first row of each table is the baseline scheme, which every other configuration modifies.
In~\cref{tab:full-components-detection}, only the detector changes: each block rescores the exact same generated texts, so the quality columns are those of the corresponding generation.
A dash indicates a configuration that we only evaluated with the identity grouping.

\fi

\end{document}